\documentclass{article}
\usepackage[english]{babel}
\usepackage{amsmath,amsfonts,amssymb,amsthm}
\usepackage{enumerate}
\usepackage{tikz}
\usepackage[hyphens]{url}
\usepackage[breaklinks=true]{hyperref}
\hypersetup{breaklinks=true}

\usepackage{skull}

\DeclareMathOperator{\reduce}{red}
\DeclareMathOperator{\dist}{d}
\DeclareMathOperator{\gs}{gs}
\DeclareMathOperator{\hreduce}{\widehat{red}}
\DeclareMathOperator{\hdist}{\hat{d}}
\DeclareMathOperator{\fRep}{Rep^{\mathit{f}}}
\DeclareMathOperator{\Hgrid}{\mathit{H}^{grid}}
\DeclareMathOperator{\Hbound}{\mathit{H}^{bound}}

\newcommand{\norm}[1]{{\left\| #1 \right\|}}
\newcommand{\abs}[1]{{\left|{#1}\right|}}
\newcommand{\floor}[1]{{\left\lfloor{#1}\right\rfloor}}
\newcommand{\ceil}[1]{{\left\lceil{#1}\right\rceil}}

\newcommand{\ket}[1]{{\left| #1 \right\rangle}}

\newcommand{\<}{\langle}
\renewcommand{\>}{\rangle}

\DeclareMathOperator{\volume}{vol}
\DeclareMathOperator{\mspan}{span}

\newenvironment{enuma}{\begin{enumerate}[\upshape (a)]}{\end{enumerate}}
\newenvironment{enumi}{\begin{enumerate}[\upshape (i)]}{\end{enumerate}}

\newcommand{\N}{\mathbb{N}}
\newcommand{\Z}{\mathbb{Z}}
\newcommand{\Q}{\mathbb{Q}}
\newcommand{\R}{\mathbb{R}}
\newcommand{\C}{\mathbb{C}}
\newcommand{\F}{\mathbb{F}}
\DeclareMathOperator{\ha}{\hat{\mathbf{a}}}
\DeclareMathOperator{\hb}{\hat{\mathbf{b}}}
\DeclareMathOperator{\hx}{\hat{\mathbf{x}}}
\DeclareMathOperator{\fa}{\mathbf{a}}
\DeclareMathOperator{\fb}{\mathbf{b}}
\DeclareMathOperator{\fe}{\mathbf{e}}
\DeclareMathOperator{\ta}{\tilde{\mathbf{a}}}
\DeclareMathOperator{\tb}{\tilde{\mathbf{b}}}
\DeclareMathOperator{\tx}{\tilde{\mathbf{x}}}
\DeclareMathOperator{\fz}{\mathbf{z}}
\DeclareMathOperator{\fm}{\mathbf{m}}

\newcommand{\calA} {\mathcal{A}}
\newcommand{\calB} {\mathcal{B}}
\newcommand{\calF} {\mathcal{F}}
\newcommand{\calI} {\mathcal{I}}
\newcommand{\calM} {\mathcal{M}}
\newcommand{\calO} {\mathcal{O}}
\newcommand{\calV} {\mathcal{V}}
\newcommand{\calR} {\mathcal{R}}
\newcommand{\calW} {\mathcal{W}}

\newcommand{\fraka} {\mathfrak{a}}

\newcommand{\ASMB}{\textbf{A}}
\newcommand{\ASME}{)}
\newcommand{\ASM}[1]{\ASMB#1\ASME}

\theoremstyle{plain}
\newtheorem{definition}{Definition}[section]
\newtheorem{theorem}[definition]{Theorem}
\newtheorem{proposition}[definition]{Proposition}
\newtheorem{lemma}[definition]{Lemma}
\newtheorem{corollary}[definition]{Corollary}
\theoremstyle{definition}
\newtheorem{remark}[definition]{Remark}

\newenvironment{algorithm}{\paragraph{Algorithm}}{}

\title{Quantum Algorithm for Computing the Period Lattice of an Infrastructure}
\author{Felix Fontein\footnote{Insitute of Mathematics, University of Zurich, Winterthurerstrasse 190,
    8057 Zurich, Switzerland; \texttt{felix.fontein@math.uzh.ch}}
  \and Pawel Wocjan\footnote{Department of Electrical Engineering and Computer Science, University
    of Central Florida, Orlando, FL 32816-2362; \texttt{wocjan@eecs.ucf.edu}}}

\begin{document}
  \maketitle
  
	We present a quantum algorithm for computing the period lattice of infrastructures of fixed
  dimension.  The algorithm applies to infrastructures that satisfy certain conditions.  The latter
  are always fulfilled for infrastructures obtained from global fields, i.e., algebraic number
  fields and function fields with finite constant fields, as described in \cite{ff-tioagfoaur}.  
  
  The first of our main contributions is an exponentially better method for sampling approximations of vectors of the dual
  lattice of the period lattice than the methods outlined in the works of {\sc Hallgren} and {\sc Schmidt and Vollmer}.
  This new method improves the success probability by a factor of at least $2^{n^2-1}$ where $n$ is the dimension.
  The second main contribution is a rigorous and complete proof that the running time of the
  algorithm is polynomial in the logarithm of the determinant of the period lattice and exponential
  in $n$.  The third contribution is the determination of an explicit lower bound on the success probability of 
  our algorithm which greatly improves on the bounds given in the above works.
  
  The exponential scaling seems inevitable because the best currently known methods for carrying out
  fundamental arithmetic operations in infrastructures obtained from algebraic number fields take
  exponential time.  In contrast, the problem of computing the period lattice of infrastructures
  arising from function fields can be solved without the exponential dependence on the dimension~$n$
  since this problem reduces efficiently to the abelian hidden subgroup problem.  This is also true
  for other important computational problems in algebraic geometry.  The running time of the best
  classical algorithms for infrastructures arising from global fields increases subexponentially
  with the determinant of the period lattice.

  \newpage
  
  \tableofcontents
  
  \section{Introduction}
  \label{sec:intro}
  
  \subsection{Informal definition of an infrastructure and the problem of computing the period lattice}
  
  An \emph{$n$-dimensional infrastructure~$\calI$} is a finite set of distinguished
  points on an $n$-dimensional torus $\R^n / \Lambda$, where $\Lambda$ is a lattice of full rank in
  $\R^n$. To every of these finitely many distinguished points, we assign a region on the torus, so
  that every point on the torus lies in exactly one such region. If $x$ is such a distinguished
  point, every point~$y$ in this region can be represented by the difference $t := y - x$ together
  with $x$, i.e., as the pair $(x, t)$. These tuples $(x, t)$ are essentially the $f$-representations
  of the infrastructure. Infrastructures can be obtained, for example, from global fields, i.e., from
  algebraic number fields as well as function fields with finite constant fields; in this case, the
  lattice corresponds to the free part of the unit group. We explain later that such
  infrastructures satisfy all assumptions we make on infrastructures in this paper.

  We present a quantum algorithm for computing the period lattice $\Lambda$ of infrastructures of
  fixed dimension $n$ and provide a rigorous and detailed proof of its performance.  We focus our
  attention on non-discrete infrastructures.  An infrastructure is called discrete if its period
  lattice is integral and the coordinates of the distinguished points are integral (or more
  generally, if everything can be made integral by a suitable rescaling).  Discrete and non-discrete
  infrastructures arise from function fields and number fields, respectively.  The problem of
  computing the period lattice of discrete infrastructures is easy since this problem can be solved
  by using the same approach as for the abelian hidden subgroup problem.  The reason is that the
  quantum algorithm for solving the abelian HSP can also be viewed as computing a hidden lattice in
  $\Z^n$.
  
  \subsection{Intuition behind the quantum algorithm and brief summary of new contributions}
  The idea behind the quantum algorithm for computing the period lattice of a (non-discrete) infrastructures is a follows. It is possible to define
  a function from the window $\calV=\{0,\ldots,qN-1\}^n\subset\Z^n$ into a certain finite set, whose elements are related to $f$-representations, so that
  \[
  	f(v) = f(v') \Leftrightarrow \frac{v-v'}{N} \approx \lambda \mbox{ for some $\lambda\in\Lambda$}.
  \]
  In words, there is a collision iff the two values $v$ and $v'$ differ approximately 
  by an integer multiple of a lattice vector of the period lattice.  This implies that 
  the elements of the preimage $f^{-1}(v)$ have the special form
  \[
  	v' = v + N \lambda + \xi_\lambda,
  \]
  where $\lambda\in\Lambda$ and $\xi_\lambda$ is a certain error vector from $(-1,1)^n$ such that $v'\in\calV$.  
  Moreover, for a constant fraction of $v$ the cardinality of the corresponding preimage is $f^{-1}(v)$ is close 
  to $\tfrac{q^n}{\det(\Lambda)}$, which corresponds to the natural density of the lattice $\Lambda$ in $\R^n$.
  
  We prove that such function $f$ exists and can always be evaluated correctly at all points
  of $\calV$ with constant probability.  Our analysis takes into account the special nature of the shapes of the regions of the distinguished points and the
  way how these regions interlock with each other.  This analysis closes a gap in the work \cite{hallgrenUnitgroup}.  The works \cite{schmidt-vollmer,arthurDiss}
  chose a different approach.  They showed that it is not necessary that the function $f$ can always be evaluated correctly.  However, their resulting analysis 
  leads to a significantly worse overall running time.
  
  Efficiency means here that we can evaluate this
  function in time that is polynomial in the logarithm of the determinant of the period lattice
  $\Lambda$ and exponential in the dimension $n$.  This exponential scaling seems inevitable because
  the best methods for carrying out fundamental arithmetic operations in such infrastructures take
  exponential time.
  
  Following the quantum algorithm for the abelian HSP, we start by evaluating the function $f$ in superposition over the window $\calV$ and measuring
  the output register.  The resulting post-measurement state is a ``pseudo-periodic'' state, i.e., a uniform superposition of the above $v'$.  It is 
  important that this superposition contains sufficiently many values of the form $v'$.  The pseudo-periodic state corresponds to a 
  uniform superposition of a randomly translated rectangular portion of the rescaled lattice $N\Lambda$ such that only few of its points are missing and the remaining points are
  only slightly perturbed.  We present a new method for precisely 
  analyzing the probability of obtaining a pseudo-periodic state with sufficiently values of the $v'$.  This analysis also closes a gap in the work \cite{hallgrenUnitgroup}.
  
  Similarly to the situation in the abelian HSP, we effectively
  remove the undesired random offset $v$ by applying a multidimensional quantum Fourier transform.  This allows us to sample approximations of lattice vectors of 
  the dual lattice $\Lambda^\ast$.  To mitigate the perturbations effects caused by the error vectors $\xi_\lambda$, we have to perform the quantum Fourier transform 
  over a larger window $\calW$.  But this comes at the price of an exponential decay of the success probability with increasing dimension $n$.  The idea to use
  a larger window goes back to \cite{hallgrenUnitgroup} and \cite{schmidt-vollmer,arthurDiss}.  We obtain here a new better method for sampling approximations improving
  the success probability by the exponential factor $2^{n-1}$ compared to the less efficient methods in \cite{hallgrenUnitgroup} and \cite{schmidt-vollmer,arthurDiss}.  This 
  is not just an improvement in the analysis, but an improvement of the algorithm.
  
  We present lattice and group theoretic results, making it possible to prove that $2n+1$ approximations obtained by the above sampling process form an approximate 
  generating set of $\Lambda^*$ with constant probability for fixed dimension.  No such bound on the number of required samples was proved in the previous works. 
  Once we have such approximate generating set, we recover an approximate basis of $\Lambda^*$.  We describe an improved
  method for this purpose.  We then determine an approximate basis of $\Lambda$ from such approximate basis of $\Lambda^*$.  
  
  Finally, we discuss in detail how to choose all parameters to obtain an approximate basis of the period lattice $\Lambda$ that has the desired approximation quality.
  We obtain an explicit lower bound on the success probability of our algorithm, which reveals precisely how the complexity depends on the various
  parameters. We compare this probability to the ones presented in the works of {\sc Schmidt and
  Vollmer} and \textsc{Schmidt} and conclude that our probability is exponentially better by at least $2^{{n^2}-1}$.
  The work of {\sc Hallgren} gives no explicit probability.
  
  Note that in the one-dimensional case more specialized algorithms lead to a much better
  probability of success; see, for example, \cite{hallgrenPell,pradeep-pawel}.
  
  \subsection{Efficient quantum algorithms for problems in arithmetic geometry}
  
  We conclude the introduction with some comments on the existence of efficient quantum algorithms for certain computationally hard problems in 
  algebraic geometry.  Readers not familiar with algebraic geometry may not be aware that many interesting problems can be reduced to the abelian HSP efficiently.  
  The understanding of these reductions does require some specialized knowledge in algebraic geometry, but the necessary results are fairly standard.  As noted previously, infrastructures 
  obtained from function fields are easier to handle than general infrastructures. As shown in Theorem~7.1 of \cite{ff-tioagfoaur}, such infrastructures embed in a
  natural way into the divisor class group of degree zero, which is a finite abelian group in the case of function fields
  with finite constant field. There are polynomial time classical algorithms to do arithmetic in this group, for instance, the ``algebraic'' algorithm by
  F.~He\ss\ \cite{hessRR,diem-habil}.  Therefore, one can directly apply the standard algorithm for
  the abelian HSP \cite{cheung-mosca} to compute the period lattice.  Other important problems, such as determining discrete logarithms in 
  the infrastructure, computing the whole divisor class group and the ideal
  class group, solving the principal ideal problem, as well as computing the Zeta function, can all be treated in the same way.  The latter problem
  was solved in \cite{kedlaya-zeta} using this approach, while relying on a less efficient ``geometric'' 
  method based on the Brill-Noether algorithm to do arithmetic. 
  
  Arithmetic geometry provides a unifying understanding and treatment of problems related to global fields.  On the one hand, the discussion above shows that the algebro-geometric problems for 
  function fields with finite constant fields (i.e., function fields of curves over finite fields) can be reduced to the abelian HSP.  This presents an elegant and efficient quantum solution.  
  On the other hand, the analysis of the quantum algorithms for the corresponding number-theoretic problems is significantly more challenging.  We believe that our rigorous and improved
  treatment of the problem of computing the period lattice of non-discrete infrastructures can serve as a valuable starting point for addressing other number-theoretic problems and also finding more
  efficient quantum algorithms for them.  A first stepping stone is our new method for sampling approximations of vectors of the dual lattice, which improves the success probability by an exponential 
  factor.
  
  \section{Formal definition of an infrastructure}
  \label{sec:infra}

  An \emph{$n$-dimensional infrastructure~$\calI$} consists of 
  \begin{itemize}
  	\item a lattice~$\Lambda$ of full rank, called the \emph{period lattice}, 
  	\item a finite non-empty set $X$, an injective map $\dist : X \to \R^n / \Lambda$, and 
  	\item a set of \emph{$f$-representations}~$\fRep(\calI)$, i.e., a subset $\fRep(\calI)
  	\subseteq X \times \R^n$ with $X \times \{ 0 \} \subseteq \fRep(\calI)$ such that the function
  	\[ 
  		\Phi_\calI : \fRep(\calI) \to \R^n / \Lambda,
  		\qquad (x, t) \mapsto \dist(x) + t 
  	\] 
  	is a bijection. 
  \end{itemize}
  Such a set of $f$-representations yields a
  \emph{reduction map}~$\reduce : \R^n/\Lambda \to X$ satisfying $\reduce(\Phi_\calI(x, t)) = x$ for all
  $(x, t) \in \fRep(\calI)$, as well as a \emph{giant step} operation~$\gs : X \times X \to X$ by $\gs(x,
  y) = \reduce(\dist(x) + \dist(y))$. Note that the set of $f$-representations has a natural group
  structure using the pull-back of the group operation of $\R^n / \Lambda$ via $\Phi_\calI$: $(x, t)
  + (x', t') := \Phi_\calI^{-1}(\Phi_\calI(x, t) + \Phi_\calI(x', t'))$.
  
  Given such a set of $f$-representations, we can \emph{unroll} the infrastructure. Let $\pi : \R^n
  \to \R^n/\Lambda$ be the canonical projection, and set 
  \[
  	\hat{X} := \pi^{-1}(\dist(X)).
  \]
  This is a discrete non-empty subset of $\R^n$ satisfying $\hat{X} + \Lambda = \hat{X}$. Define
  $\hdist(\hat{x}) = \hat{x}$ for all $\hat{x} \in \hat{X}$ and 
  \[
  	\hat{V}_{\hat{x}} := \{ \hdist(\hat{x}) +
  	t \mid (\dist^{-1}(\pi(\hat{x})), t) \in \fRep(\calI) \}
  \]
  for every $\hat{x} \in \hat{X}$. Then $\R^n$ is the disjoint union of all $\hat{V}_{\hat{x}}$, $\hat{x} \in \hat{X}$, and 
  one can define $\hreduce : \R^n \to \hat{X}$ by $\hreduce(v) = \hat{x}$ if $v \in \hat{V}_{\hat{x}}$.
  
  The unrolled infrastructure is periodic with period lattice $\Lambda$ in the sense that for
  $\hat{x} \in \hat{X}$, $v \in \R^n$ and $\lambda \in \Lambda$, we have $\hat{x} + \lambda \in
  \hat{X}$, $\hat{V}_{\hat{x} + \lambda} = \hat{V}_{\hat{x}} + \lambda$, $\hreduce(v + \lambda) = \hreduce(v) +
  \lambda$ and $\hdist(\hat{x} + \lambda) = \hdist(\hat{x}) + \lambda$. Moreover, $\pi(\hat{x}) =
  \pi(\hat{y})$ for $\hat{x}, \hat{y} \in \hat{X}$ if, and only if, $\hat{y} - \hat{x} \in \Lambda$.
  
  For $s, t\in \R^n$, we write $[s, t]$ for $\{ r \in \R^n \mid s \le r \le t \}$, where ``$\le$''
  denotes the component-wise inequality on $\R^n$. We say that a subset~$U \subseteq \R^n$ is
  \emph{cornered} with \emph{corner}~$s \in \R^n$ if $s \in U$ and for every~$t \in U$, $t \in [s,
  t] \subseteq U$. In other words, $U = \bigcup_{t \in U} [s, t]$. Note that every cornered subset
  of $\R^n$ has exactly one corner, which is its minimal element with respect to $\le$. We say that
  $\calI$ is \emph{cornered} if for all $\hat{x} \in \hat{X}$, $\hat{V}_{\hat{x}}$ is cornered with
  corner~$\hat{x}$.
  
  We make the following assumptions:
  \begin{enumerate}[\ASMB1\ASME]
    \item There exists a constant $A > 0$ such that for every $\hat{x} \in \hat{X}$, \[
    \hat{V}_{\hat{x}} \subseteq \hat{x} + [0, A]^n. \]
    \item There exist constants $C, D > 0$ such that for every $r \in \R^n$, the set \[ (r + [0,
    C]^n) \cap \hat{X} \] contains at most $D$ elements.
    \item There exists a polynomial-time algorithm such that for given $k \in \N$ and $u \in \Z^n$, one
    can compute $(x, t) \in X \times 2^{-k} \Z^n$ such that
    \begin{enuma}
      \item $\norm{\hat{x} + t - 2^{-k}u}_\infty \le 2^{-k}$ for some $\hat{x} \in \hat{X}$ with
      $\dist^{-1}(\pi(\hat{x})) = x$;
      \item $\bigl(2^{-k}u + (-2^{-k}, 2^{-k})^n\bigr) \cap \hat{V}_{\hat{x}} \neq \emptyset$.
    \end{enuma}
    The running time is polynomial in $k$ and $\log \norm{u}_\infty$ when the dimension $n$ is held constant. 
  \end{enumerate}
  
  \begin{proposition}
    \label{prop:infrafromFFNF}
    Let $K$ be a global field. Then any infrastructure obtained from $K$ in the sense of
    \cite[Section~6]{ff-tioagfoaur} has $f$-representations in a natural way and is
    cornered. Moreover, it satisfies \ASM1 to \ASM3 with explicit constants $A, C, D$:
    
    If $K$ is a number field of discriminant~$\Delta$ and degree~$d = [K : \Q]$, then one can choose
    $A = \tfrac{1}{2} \log \abs{\Delta}$, $C = \log 2$ and $D = 4^d$. If $K$ is a function field of
    genus~$g$ and degree~$d = [K : k(x)]$, then one can choose $A = g + d - 1$, $C = 1 -
    \varepsilon$ for any $\varepsilon\in (0,1)$, and $D = 1$.
  \end{proposition}
  
  \begin{proof}[Sketch of Proof.]
    Assume that the infrastructure is $\calI = (X^\fraka, \dist^\fraka, \reduce^\fraka)$ in the
    notation of \cite{ff-tioagfoaur}. Here, $\fraka$ is an ideal of the ring of integers~$\calO$ (or
    the ring of holomorphic functions in case $K$ is a function field), and $X^\fraka$ is
    essentially the set of reduced ideals equivalent to $\fraka$. If $\abs{\bullet}_1, \dots,
    \abs{\bullet}_{n+1}$ are the pairwise different absolute values of $K$, we define $\Lambda := \{
    (\log \abs{\varepsilon}_1, \dots, \log \abs{\varepsilon}_n) \mid \varepsilon \in \calO^* \}$,
    which is isomorphic to the free part of the finitely generated abelian group~$\calO^*$ of units
    of $\calO$. The definition of $f$-representations is rather technical, whence we do not repeat
    it here, but just refer to Definition~6.3 of \cite{ff-tioagfoaur}. For every $\hat{x} \in
    \hat{X}$,
    \[ 
    	\hat{V}_{\hat{x}} = \hat{x} + W(d^{-1}(\pi(\hat{x}))), \quad \text{where } W(x) = \{ t \in \R^n
        \mid (x, t) \in \fRep(\fraka) \} \text{ for } x \in X.
    \] 
    It is clear from Definition~6.3 in \cite{ff-tioagfoaur} that $W(x)$ is cornered with
    corner~$0$. Hence, $\calI$ is a cornered infrastructure. Our assumption \ASM1 follows from
    Proposition~8.1 of \cite{ff-tioagfoaur}. The second assumption \ASM2 holds trivially for
    function fields; for number fields, it follows from Lemma~3.2 in
    \cite{buchmann-ontheperiodlength}.
    
    Assumption~\ASM3 will be discussed in an upcoming paper of the first author and
    M.~J.~Jacobson,~Jr. In the case of function fields, the algorithms are of polynomial running
    time with respect to the genus of the function field as well as the size of its
    representation. In the case of number fields, the algorithms are polynomial with respect to the
    logarithm of the discriminant of the number field, but exponential in its degree~$d = [K : \Q]$,
    as one has to find shortest vectors in lattices of dimension~$d$.
  \end{proof}
  
  Note that the algorithm we plan to use for \ASM3 is exponential in $n$, but significantly more
  efficient than the algorithms that were proposed in \cite{hallgrenUnitgroup} and
  \cite{schmidt-vollmer}. These are based on \cite[Chapter~5 and 6]{thielDiss}, which essentially
  uses Buchmann's baby step algorithm \cite{genlagrange, buchmann-habil}. The latter is known for
  being practically unusable \cite{buchmannjuntgenpohst-practialGLA}. Even on modern computers,
  computing all minima of one reduced ideal can take a long time for not too large number field
  degrees, say~$[K : \Q] = 8$ (which yields $n = 7$); the first author verified this in 2010 when
  implementing that algorithm.

  Note that Schoof's Algorithm~10.7 in \cite{schoofArakelov} is also 
  mentioned in \cite{hallgrenUnitgroup} as a more efficient alternative to Buchmann's algorithm. Unfortunately, 
  Schoof's algorithm uses a different distance function from the one used by Hallgren and by us.  Therefore, 
  Schoof's algorithm cannot be applied without non-trivial modifications if one wants to obtain a provably 
  polynomial-time quantum algorithm for computing the period lattice.
  
  Observe that \ASM3 follows from the existence of two simpler algorithms.  Before we list these,
  we need to define what an ``approximate $f$-representation of error at most~$\varepsilon$'' of a point~$r
  \in \R^n$ is. This is a pair~$(x, t) \in X \times \R^n$ satisfying
  \begin{enuma}
    \item $\norm{\hat{x} + t - r}_\infty \le \varepsilon$ for some $\hat{x} \in \hat{X}$ with
    $\dist^{-1}(\pi(\hat{x})) = x$;
    \item $\bigl(r + (-\varepsilon, \varepsilon)\bigr)^n \cap \hat{V}_{\hat{x}} \neq \emptyset$,
  \end{enuma}
  Now we can describe the characteristics of the two simpler algorithms, which can be combined to obtain such an
  algorithm as described in \ASM3:
  \begin{enuma}
    \item one algorithm which, given $\ell\in\N$ and $r \in 2^{-\ell} \{ -2^\ell, -2^\ell + 1, \dots, 2^\ell \}^n\subset[-1,1]^n$, computes an approximate 
    $f$-representation $(x, t)$ of error at most $2^{-\ell}$ such that
    $\norm{\dist(x) + t - r}_\infty \le 2^{-\ell}$ in time polynomial in $\ell$;
    \item a second algorithm which, given two approximate $f$-representations of error at
    most~$2^{\ell'}$, computes an approximate $f$-representation of their sum of error at
    most~$2^{\ell'+1}$ in time polynomial in $\ell'$.
  \end{enuma}
  One can compute an approximate $f$-representation of any $r \in\R^n$ of error at most $2^{-k}$ in time polynomial in $\log \norm{r}_\infty$ and $k$.
  This is done by using a double and add technique and by calling these algorithms to obtain approximate $f$-representations of error at most
  $2^{-(k+k')}$, where $k' = O(\log \norm{r}_\infty)$. 
  
  The formal definition of the problem of computing the period lattice is as follows.  
  
  \begin{definition}
  Given $\gamma\in (0,1)$, the task is to find
  $\tilde{\lambda}_1,\ldots,\tilde{\lambda}_n\in\R^n$ such that there exists a basis $\lambda_1,\ldots,\lambda_n$ of $\Lambda$ with 
  \[
  	\| \tilde{\lambda}_j - \lambda_j \|_2 \le \gamma
  \]
  for $j=1,\ldots,n$.  We call such $\tilde{\lambda}_1,\ldots,\tilde{\lambda}_n$ a $\gamma$-approximate basis of $\Lambda$.
  \end{definition}
  
  We present a quantum algorithm with running time polynomial in $\log \det(\Lambda)$ and
  $\log(1/\gamma)$ when $A, 1/C, D$ and $1 / \lambda_1(\Lambda)$ can be bounded polynomially in
  terms of $\log \det(\Lambda)$. Here, $\lambda_1(\Lambda)$ denotes the first consecutive minimum of
  $\Lambda$, i.e., the length of a shortest non-zero vector in $\Lambda$. Note that for number
  fields, $\lambda_1(\Lambda)$ can be bounded from below by a bound depending only on $n$; see
  Satz~5.6 in \cite{buchmann-habil}.
  
  In the case of computing units of a global field, computing a $\gamma$-approximate basis of
  $\Lambda$ yields approximations of the logarithms of the absolute values of the units. These
  approximations can be refined to arbitrary precision in polynomial time. Note that one can also
  relatively efficiently recover the corresponding units themselves; since their representation is
  not of size polynomial in the genus respectively logarithm of the discriminant, explicitly
  computing them cannot be done in polynomial time. What can be done is computing a so-called
  \emph{compact representation} of a unit, which was presented for number fields in
  \cite{thielDiss,thiel-comprep} and for function fields in \cite{hallgren-eisentraeger}; one can
  modify the quantum algorithm to output such compact representations of the units and still run in
  polynomial time.
  
  Finally, we want to mention that our algorithm can be interpreted as an algorithm for solving certain
  instances of a \emph{Hidden Subgroup Problem} for the group~$G = \R^n$ provided that the group
  operation in $\fRep(\calI)$ is effective. In case the infrastructure is obtained from a global
  field as in the above proposition, the group operation is effective and is described explicitly in
  Theorem~7.3 of \cite{ff-tioagfoaur}.

  Now one can consider the group homomorphism $f : \R^n \to \fRep(\calI)$ as the composition of the
  canonical projection $\pi : \R^n \to \R^n/\Lambda$ with $\Phi_\calI^{-1}$. This map can be
  effectively computed -- ignoring rounding and approximation issues -- and it hides the
  lattice~$\Lambda$ by $\ker f = \Lambda$.

  \section{Detailed outline of the quantum algorithm and new contributions}
  \label{sec:algoutline}
  
  Let $N \in \N$ and $s \in \R^n$ be fixed. Consider the function 
  \[ 
  	f : \R^n \to X \times \Z^n,
  	\qquad v \mapsto (x, \floor{N t}) \text{ if } \Phi_\calI^{-1}\bigl(\pi(s + \tfrac{1}{N} v )\bigr) = (x, t). 
  \] If
  $f(v) = f(v')$ for $v,v' \in \Z^n$, then $v - v'$ lies close to an element of $N \Lambda$. We
  want to use the quantum computer to find such collisions.
  
  Let $\calV = \{0,\ldots,qN-1\}^n$ and $\calW = \{0,\ldots,2nqN-1\}^n$ where $q$ and $N$ are positive integers that will be fixed later. 
  Set $V = \abs{\calV}$ and $W=\abs{\calW}$.
  The input register is $\C^W=\left(\C^{2nqN}\right)^{\otimes n}$.  
  The output register is $\C^d$ with $d$ sufficiently large
  so it can store any element of the image $f(\calV)$.  
  In the following we use $f$ to denote the restriction of $f$ to $\calV$.
  We assume that we have a reversible version $U_f$ of $f$ that acts on the above input and output registers.
  
\newpage
  
  \begin{algorithm}
  \begin{enumerate}
  	\item We start by preparing the state
  	\[ 
  		\frac{1}{\sqrt{V}} \sum_{v\in \calV} \ket{v} \ket{f(v)}. 
  	\] 
  	Note that we evaluate $f$ only on the subset $\calV$ of $\calW$. 
  	\item We measure the output register and denote the outcome by $f(v)$ for some $v\in\calV$.
  	The post-measurement state is then 
  	\[ 
  		\frac{1}{\sqrt{M}} \sum_{v' \in \calM} \ket{v'} \ket{f(v)}
  	\] 
  	where $\calM := \{ v' \in \calV \mid f(v') = f(v) \}$ and $M=\abs{\calM}$.
  	\item We apply the $n$-fold tensor product of the quantum Fourier transform of size $2nqN$ on the input register
  	and obtain the state
  	\[ 
  		\frac{1}{\sqrt{M \, W}} \sum_{w \in\calW} \sum_{v \in \calM} 
  		\exp\biggl(
  		2 \pi i \, v' \cdot \frac{w}{2 n q N} 
  		\biggr) \ket{w} \ket{(f(v)}
  	\]
  	where $\cdot$ denotes the inner product on $\R^n$.
  	\item Finally, we measure the input register and denote the outcome by $w$.  
  \end{enumerate}
  \end{algorithm}
  This quantum procedure is repeated $2n+1$ many times to obtain the samples $w_1,\ldots,w_{2n+1}$.  A subsequent classical post-processing step makes it
  possible to extract an approximate basis of $\Lambda$ from these samples with a probability that
  can be bounded from below by a positive constant.
  
  \paragraph{Organization of the paper and outline of technical results}\ \medskip\\
  In Section~\ref{sec:computef}, we prove that with constant
  probability all evaluation points $v/N+s$ ($v\in\calV$) are sufficiently far away from the
  boundary of $\hat{V}_{\hat{x}}$ for all $\hat{x}\in\hat{X}$.  This is achieved by choosing the
  shift $s$ uniformly at random from a certain finite set. This ensures that we can compute $f(v)$
  correctly for all $v\in\calV$ even though we may only determine approximate $f$-representations.
  
  In Section~\ref{sec:probperiodic}, we show that the probability for post-measurement states being
  periodic states can be bounded from below by a constant. Roughly speaking, a periodic state
  corresponds to a (randomly) translated and perturbed finite portion of the lattice $N \Lambda$
  that may be missing some points.  In particular, we establish a lower bound on $M$ showing that
  not too many points are missing in the superposition.
  
  To derive the results in Sections~\ref{sec:computef} and \ref{sec:probperiodic}, it is absolutely
  indispensable to take into account that the infrastructure is cornered.  Relying only a lower
  bound on the minimal distance between two elements of $\hat{X}$ is not sufficient because the
  union of $\varepsilon$-neighborhoods of the boundaries of $\hat{V}_{\hat{x}}$ of all
  $\hat{x}\in\hat{X}$ could still fill out too much of $\R^n$.  In the one-dimensional case, the
  regions $\hat{V}_{\hat{x}}$ are intervals. In contrast to that, in the $n$-dimensional case, their
  shapes can take on much more complicated forms. This makes the analysis more difficult.  This
  problem was mentioned, but not resolved in \cite{hallgrenUnitgroup}, while in
  \cite{schmidt-vollmer,arthurDiss}, this problem was solved differently by relaxing the conditions
  of the quantum algorithm on the function~$f$.
  
  In Section~\ref{sec:fouriersampling}, we show that the last step of the above quantum procedure
  yields an approximation of an element of the dual lattice~$\Lambda^\ast = \{ \lambda^\ast \in \R^n
  \mid \forall \lambda \in \Lambda : \langle \lambda^\ast, \lambda \rangle \in \Z \}$ with a certain
  probability.  It becomes essential here that the Fourier transform is taken over the larger window
  $\calW$, while $f$ is only evaluated inside $\calV$.  This makes it possible to mitigate the
  perturbation effects.
  
  More precisely, we determine a lower bound on the probability the outcome $w$ obtain in the final
  step is contained in the set $\calR_{\lambda^\ast}$, where
  \[
  	\calR_{\lambda^\ast} := \Big\{ (w_1,\ldots,w_n) \Bigm| w_k \in \{ \floor{2nq \lambda_k^*},
        \floor{2nq\lambda^\ast_k} + 1 \} \text{ for $k=1,\ldots,n$} \Big\}
  \]
  and $\lambda^* = (\lambda_1^*, \dots, \lambda_n^*) \in \Lambda^*$.  Such elements yield good
  approximations of $\lambda^\ast$ since
  \[
  	\Big\| \frac{w}{2nq} - \lambda^* \Big\|_2 \le \frac{1}{2\sqrt{n}q} 
  \]
  for all $w\in\calR_{\lambda^\ast}$.  
  
  The works \cite{hallgrenUnitgroup} nor \cite{schmidt-vollmer,arthurDiss} consider only elements of the more restrictive form $[2nq\lambda^\ast]$, 
  where $[u]$ means that we round each coefficient of $u\in\R^n$ to the closest integer.  This is why our method improves the success probability of obtaining a single 
  good approximation by the exponential factor $2^{n-1}$.  It can be shown that at least $n+1$ samples are needed so our method provably leads to an 
  overall improvement of the success probability by the factor $2^{n^2-1}$.
    
  In Section~\ref{sec:part1}, we present lattice and group theoretic results, yielding a lower bound on the probability that $n$ lattice vectors drawn uniformly at random 
  from $L\cap [0,b)^n$ and $n+1$ lattice vectors drawn uniformly at random from $L \cap [0,b_0)^n$ generate together 
  the entire lattice $L$, where $L$ is a full-rank lattice in $\R^n$ and $b<b_0$ are 
  sufficiently large. Neither \cite{hallgrenUnitgroup} nor \cite{schmidt-vollmer,arthurDiss} provide
  an explicit and proven upper bound on the complexity of generating a lattice by drawing samples. 
  But this is a crucial result, directly affecting the success probability of the algorithm.
      
  In Section~\ref{sec:approximateDualLattice}, we specialize these lattice-theoretic results to $L:=\Lambda^\ast$ and present an explicit lower bound on the probability that 
  the $2n+1$ samples $w_1,\ldots,w_{2n+1}$ output by our quantum algorithm yield an approximate generating set for the dual lattice $\Lambda^\ast$.
  
  In Section~\ref{sec:part2}, we first present technical results based on \cite{BK:93} showing how to construct an approximate basis of $L$ from 
  an approximate generating set of $L$.  Then, we show how to recover an approximate basis of the dual lattice $L^\ast$ from the previously determined approximate basis of $L$.  
    
  Finally, in Section~\ref{sec:final}, we combine all results from the previous sections and show to find an approximate basis for the period lattice $\Lambda$.  We explain in detail 
  how to choose all parameters.  We also bound the success probability of our algorithm from below.  There is a classical method for checking whether the computed basis vectors
  are indeed close to elements of $\Lambda$. If that is the case, we have computed $\Lambda$ with a high
  probability.
  
  Unfortunately, the success probability of this algorithm decreases exponentially in the dimension $n$ of the infrastructure.
  This is a common problem of such algorithms which also applies to the algorithms described in \cite{hallgrenUnitgroup} and \cite{schmidt-vollmer} (see also 
  \cite[p.~122]{arthurDiss}).  However the success probability of our algorithm decreases less rapidly than that of the previous works.  It is better by the exponential factor
  $2^{n^2-1}$.
  
  \section{Computing the function $f$ that hides the period lattice $\Lambda$}
  \label{sec:computef}
  
  We consider a computable version~$\tilde{f}$ of $f$ and show under which conditions $f(v) =
  \tilde{f}(v)$ holds for all $v\in\calV$ with high probability. Recall that $v$ corresponds to the
  point $s+\tfrac{v}{N}$, where $s$ is a random offset. We show that if $s$ is chosen uniformly
  random at random from a certain finite set, then with high probability none of these evaluation points
  $u:=s+\tfrac{1}{N} v$ (for $v\in\calV$) falls into regions in which  
  the method \ASM3 may return a result that leads to a wrong evaluation of $f(v)$.
  
  Let $v \in \calV$ yield $f(v) = (x, \floor{N t})$ with $(x, t) = \Phi_\calI^{-1}\bigl(\pi(u)\bigr)$. Let $\hat{x}
  \in \hat{X}$ with $u \in \hat{V}_{\hat{x}}$; then $\pi(\hat{x}) = \dist(x)$ and $u - \hat{x} =
  t$. If $u$ is sufficiently far away from $\partial \hat{V}_{\hat{x}}$, then the oracle in \ASM3 returns
  the correct $x \in X$. Moreover, if $t = (t_1, \dots, t_n) \in \R^n$ has no coordinate which comes
  close to an integer multiple of $\frac{1}{N}$, then the coordinates of $N t$ are bounded away from
  integers and $\floor{N t} = \floor{N t'}$ for all $t'$ which are close enough to $t$. This ensures
  that the oracle in \ASM3 outputs an approximation~$(x, t')$ of
  $\Phi_\calI^{-1}(u) = (x, t)$ such that $(x, \floor{N t}) = (x, \floor{N t'})$.
  
  A \emph{boundary point} of $\hat{\calI}$ is a point $u \in \R^n$ such that every neighborhood of $u$
  contains points from at least two different $\hat{V}_{\hat{x}}$. Denote the set of all boundary
  points by $H$; then 
  \[ 
  	H = \bigcup_{\hat{x} \in \hat{X}} \partial \hat{V}_{\hat{x}}. 
  \] 
  For a given $\varepsilon > 0$, define the enhanced boundary
  \[ 
  	H(\varepsilon) := H + [-\varepsilon, \varepsilon]^n. 
  \] 
  Observe that $\hat{X} \subset H \subset H(\varepsilon)$ since by assumption all
  $\hat{V}_{\hat{x}}$ are cornered sets with corner $\hat{x}$.  An example of how cornered sets
  could tile the plane $\R^2$ is shown in Figure~\ref{fig:planetiling}, in which the enhanced boundary
  $H(\varepsilon)$ is highlighted. 
  
  \tikzstyle{background} = [dashed,white!67!black,fill=white!100!black]
  \tikzstyle{highlight} = [thick,fill=white!67!black]
  \tikzstyle{lines} = [black]
  \tikzstyle{fatlines} = [thick,black]
  \tikzstyle{dots} = [fill=black]
  \tikzstyle{fatdots} = [draw=black,thick,fill=black]
  
  \begin{figure}[Hh]
    \begin{center}
      \begin{tikzpicture}[y=-1cm, scale=0.5]
        \draw[background] (9.5,1) rectangle (27.5,13.5);
        \path[highlight] (10.4,1) rectangle (10.6,3.6);
        \path[highlight] (10.4,3.6) rectangle (12.6,3.4);
        \path[highlight] (12.4,4.1) rectangle (12.6,2.4);
        \path[highlight] (11.6,1) rectangle (11.4,2.6);
        \path[highlight] (11.4,2.6) rectangle (14.6,2.4);
        \path[highlight] (14.4,1) rectangle (14.6,5.1);
        \path[highlight] (12.4,4.1) rectangle (14.6,3.9);
        \path[highlight] (14.4,5.1) rectangle (15.1,4.9);
        \path[highlight] (9.5,6.1) rectangle (11.6,5.9);
        \path[highlight] (11.4,7.1) rectangle (11.6,3.4);
        \path[highlight] (11.4,7.1) rectangle (15.1,6.9);
        \path[highlight] (14.9,2.9) rectangle (15.1,10.1);
        \path[highlight] (14.4,2.9) rectangle (17.6,3.1);
        \path[highlight] (17.4,1.4) rectangle (17.6,3.6);
        \path[highlight] (15.9,1.6) rectangle (16.1,1);
        \path[highlight] (15.9,1.6) rectangle (20.6,1.4);
        \path[highlight] (20.4,1) rectangle (20.6,2.1);
        \path[highlight] (14.9,3.4) rectangle (22.6,3.6);
        \path[highlight] (20.4,2.1) rectangle (23.6,1.9);
        \path[highlight] (23.4,1) rectangle (23.6,3.6);
        \path[highlight] (22.4,1.9) rectangle (22.6,5.1);
        \path[highlight] (22.4,5.1) rectangle (24.6,4.9);
        \path[highlight] (19.4,6.1) rectangle (19.6,3.4);
        \path[highlight] (19.4,6.1) rectangle (23.1,5.9);
        \path[highlight] (23.4,3.6) rectangle (26.6,3.4);
        \path[highlight] (25.4,1) rectangle (25.6,2.1);
        \path[highlight] (25.4,2.1) rectangle (26.6,1.9);
        \path[highlight] (26.4,1) rectangle (26.6,4.6);
        \path[highlight] (26.4,4.6) rectangle (27.5,4.4);
        \path[highlight] (24.4,3.4) rectangle (24.6,7.6);
        \path[highlight] (24.4,7.6) rectangle (27.5,7.4);
        \path[highlight] (21.4,5.9) rectangle (21.6,8.6);
        \path[highlight] (21.4,8.6) rectangle (22.1,8.4);
        \path[highlight] (21.9,5.9) rectangle (22.1,10.6);
        \path[highlight] (21.9,10.6) rectangle (23.1,10.4);
        \path[highlight] (22.9,4.9) rectangle (23.1,13.1);
        \path[highlight] (22.9,13.1) rectangle (27.5,12.9);
        \path[highlight] (14.9,10.1) rectangle (22.1,9.9);
        \path[highlight] (18.4,9.9) rectangle (18.6,12.6);
        \path[highlight] (18.4,12.6) rectangle (23.1,12.4);
        \path[highlight] (13.4,6.9) rectangle (13.6,13.5);
        \path[highlight] (20.4,12.4) rectangle (20.6,13.5);
        \path[highlight] (20.9,13.5) rectangle (21.1,12.4);
        \path[highlight] (24.4,12.9) rectangle (24.6,13.5);
        \path[highlight] (26.4,13.5) rectangle (26.6,12.9);
        \draw[lines] (15,10) -- (15,3.5) -- (19.5,3.5) -- (19.5,6) -- (21.5,6) -- (21.5,8.5) -- (22,8.5) -- (22,10) -- cycle;
        \draw[lines] (18.5,12.5) -- (23,12.5) -- (23,10.5) -- (22,10.5) -- (22,10) -- (18.5,10) -- cycle;
        \draw[lines] (22,10.5) rectangle (23,6);
        \draw[lines] (21.5,8.5) rectangle (22,6);
        \draw[lines] (19.5,6) -- (23,6) -- (23,5) -- (22.5,5) -- (22.5,3.5) -- (19.5,3.5) -- cycle;
        \draw[lines] (22.5,5) -- (24.5,5) -- (24.5,3.5) -- (23.5,3.5) -- (23.5,2) -- (22.5,2) -- cycle;
        \draw[lines] (11.5,7) -- (15,7) -- (15,5) -- (14.5,5) -- (14.5,4) -- (12.5,4) -- (12.5,3.5) -- (11.5,3.5) -- cycle;
        \draw[lines] (14.5,5) rectangle (15,3);
        \draw[lines] (12.5,4) rectangle (14.5,2.5);
        \draw[lines] (15,3.5) rectangle (17.5,3);
        \draw[lines] (17.5,3.5) -- (22.5,3.5) -- (22.5,2) -- (20.5,2) -- (20.5,1.5) -- (17.5,1.5) -- cycle;
        \draw[lines] (20.5,1) -- (20.5,2) -- (23.5,2) -- (23.5,1);
        \draw[lines] (16,1) -- (16,1.5) -- (17.5,1.5) -- (17.5,3) -- (14.5,3) -- (14.5,1);
        \draw[lines] (14.5,1) -- (14.5,2.5) -- (11.5,2.5) -- (11.5,1);
        \draw[lines] (10.5,1) -- (10.5,3.5) -- (11.5,3.5) -- (11.5,6) -- (9.5,6);
        \draw[lines] (25.5,1) -- (25.5,2) -- (26.5,2);
        \draw[lines] (26.5,1) -- (26.5,4.5) -- (27.5,4.5);
        \draw[lines] (24.5,3.5) -- (26.5,3.5) -- (26.5,4.5) -- (27.5,4.5);
        \draw[lines] (24.5,3.5) -- (24.5,7.5) -- (27.5,7.5);
        \draw[lines] (13.5,13.5) -- (13.5,7);
        \draw[lines] (20.5,12.5) -- (20.5,13.5);
        \draw[lines] (21,13.5) -- (21,12.5);
        \draw[lines] (23,12.5) -- (23,13) -- (27.5,13);
        \draw[lines] (24.5,13.5) -- (24.5,13);
        \draw[lines] (26.5,13.5) -- (26.5,13);
      \end{tikzpicture}
      \caption{Demonstrating the tiling of $\R^2$ by cornered sets~$\hat{V}_{\hat{x}}$, $\hat{x} \in
        \hat{X}$. The enhanced boundary region $H(\varepsilon)$ is highlighted.}
      \label{fig:planetiling}
    \end{center}
  \end{figure}
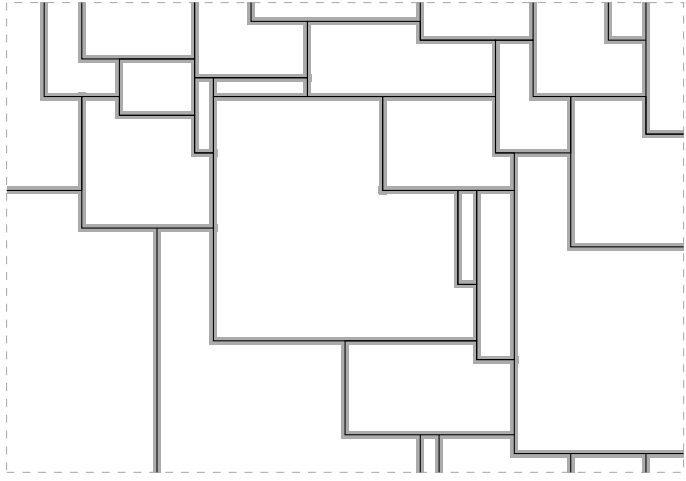
  
  If $u=s+\tfrac{v}{N} \not\in H(\varepsilon)$, then the oracle in \ASM3 can be
  used to correctly compute the $x$ part of $f(v) = (x, \floor{N t})$.
  To ensure that the $\floor{N t}$ part of $f(v) = (x, \floor{N t})$ is also correctly computed, we need $v$
  to avoid a larger set. Formally, we define 
  \[ 
  	\Hgrid(\varepsilon) := 
  	\bigcup_{\hat{x} \in \hat{X}} \Bigl( (\tfrac{1}{N} \N^n + \partial \hat{V}_{\hat{x}}) 
  	\cap
  	\overline{\hat{V}_{\hat{x}}} \Bigr) + [-\varepsilon, \varepsilon]^n. 
  \] 
  Clearly, we have $H(\varepsilon) \subseteq \Hgrid(\varepsilon)$ for all~$N \geq 1$. An example of what $\Hgrid(\varepsilon)$
  may look like is shown in Figure~\ref{fig:Hgridsketch}.
  
  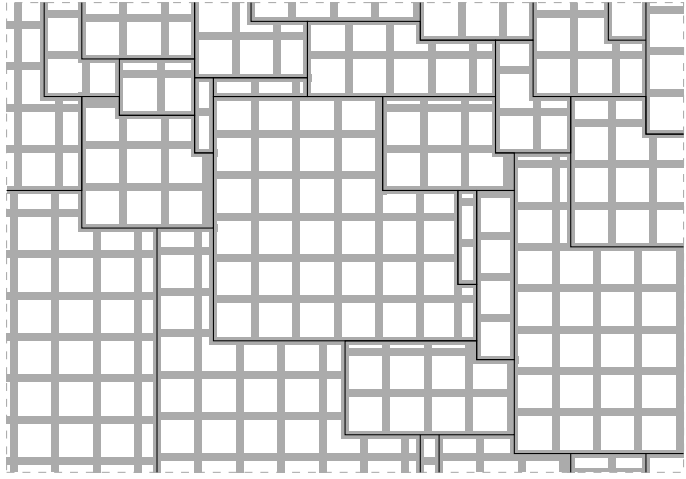
\begin{figure}[Hh]
    \begin{center}
      \begin{tikzpicture}[y=-1cm, scale=0.5]
        \draw[background] (9.5,1) rectangle (27.5,13.5);
        \path[highlight] (12.5,7) rectangle (12.7,4);
        \path[highlight] (13.6,7) rectangle (13.8,4);
        \path[highlight] (14.7,7) rectangle (14.9,5);
        \path[highlight] (11.5,6) rectangle (15,5.8);
        \path[highlight] (11.5,4.9) rectangle (14.5,4.7);
        \path[highlight] (11.5,3.8) rectangle (12.5,3.6);
        \path[highlight] (11.5,3.5) rectangle (11.7,2.5);
        \path[highlight] (10.5,2.5) rectangle (11.5,2.3);
        \path[highlight] (10.5,1.4) rectangle (11.5,1.2);
        \path[highlight] (13.5,4) rectangle (13.7,2.5);
        \path[highlight] (12.5,3) rectangle (14.4,2.8);
        \path[highlight] (14.5,4) rectangle (15,3.8);
        \path[highlight] (16,3.5) rectangle (16.2,3);
        \path[highlight] (17.1,3) rectangle (17.3,3.5);
        \path[highlight] (15.5,3) rectangle (15.7,1);
        \path[highlight] (16.6,3) rectangle (16.8,1.5);
        \path[highlight] (14.5,2) rectangle (17.5,1.8);
        \path[highlight] (17,1.5) rectangle (17.2,1);
        \path[highlight] (18.1,1.5) rectangle (18.3,1);
        \path[highlight] (19.2,1.4) rectangle (19.4,1);
        \path[highlight] (20.3,1.4) rectangle (20.5,1);
        \path[highlight] (18.5,3.5) rectangle (18.7,1.4);
        \path[highlight] (19.6,1.5) rectangle (19.8,3.5);
        \path[highlight] (20.7,3.4) rectangle (20.9,2);
        \path[highlight] (21.8,3.5) rectangle (22,2);
        \path[highlight] (17.5,2.5) rectangle (22.5,2.3);
        \path[highlight] (23.5,5) rectangle (23.7,3.5);
        \path[highlight] (22.5,4) rectangle (24.5,3.8);
        \path[highlight] (22.5,2.9) rectangle (23.5,2.7);
        \path[highlight] (21.5,2) rectangle (21.7,1);
        \path[highlight] (22.6,2) rectangle (22.8,1);
        \path[highlight] (24.5,3.5) rectangle (24.7,1);
        \path[highlight] (25.6,3.5) rectangle (25.8,2);
        \path[highlight] (23.5,2.5) rectangle (26.5,2.3);
        \path[highlight] (23.5,1.4) rectangle (25.5,1.2);
        \path[highlight] (26.5,3.5) rectangle (27.5,3.3);
        \path[highlight] (26.5,2.4) rectangle (27.5,2.2);
        \path[highlight] (26.5,1.3) rectangle (27.5,1.1);
        \path[highlight] (25.5,7.5) rectangle (25.7,3.5);
        \path[highlight] (26.6,7.5) rectangle (26.8,4.4);
        \path[highlight] (24.5,6.5) rectangle (27.5,6.3);
        \path[highlight] (24.5,5.4) rectangle (27.5,5.2);
        \path[highlight] (24.5,4.3) rectangle (26.5,4.1);
        \path[highlight] (20.5,6) rectangle (20.7,3.5);
        \path[highlight] (21.6,6) rectangle (21.8,3.5);
        \path[highlight] (22.7,6) rectangle (22.9,5);
        \path[highlight] (19.5,5) rectangle (22.5,4.8);
        \path[highlight] (19.5,3.9) rectangle (22.5,3.7);
        \path[highlight] (21.5,7.5) rectangle (22,7.3);
        \path[highlight] (21.5,6.4) rectangle (22,6.2);
        \path[highlight] (22,9.5) rectangle (23,9.3);
        \path[highlight] (22,8.4) rectangle (23,8.2);
        \path[highlight] (22,7.3) rectangle (23,7.1);
        \path[highlight] (23,6) rectangle (22,6.2);
        \path[highlight] (15,9) rectangle (22,8.8);
        \path[highlight] (15,7.9) rectangle (21.5,7.7);
        \path[highlight] (15,6.8) rectangle (21.5,6.6);
        \path[highlight] (15,5.5) rectangle (19.5,5.7);
        \path[highlight] (15,4.4) rectangle (19.5,4.6);
        \path[highlight] (16,10) rectangle (16.2,3.5);
        \path[highlight] (17.1,10) rectangle (17.3,3.5);
        \path[highlight] (18.2,9.9) rectangle (18.4,3.5);
        \path[highlight] (19.3,10) rectangle (19.5,3.5);
        \path[highlight] (20.4,10) rectangle (20.6,6);
        \path[highlight] (21.4,8.5) rectangle (21.6,10);
        \path[highlight] (19.5,12.5) rectangle (19.7,10);
        \path[highlight] (20.6,12.5) rectangle (20.8,10);
        \path[highlight] (21.7,12.5) rectangle (21.9,10);
        \path[highlight] (22.8,12.5) rectangle (23,10.5);
        \path[highlight] (18.5,11.5) rectangle (23,11.3);
        \path[highlight] (18.5,10.4) rectangle (22,10.2);
        \path[highlight] (24,13) rectangle (24.2,5);
        \path[highlight] (25.1,13) rectangle (25.3,7.5);
        \path[highlight] (26.2,13) rectangle (26.4,7.5);
        \path[highlight] (27.3,13) rectangle (27.5,7.5);
        \path[highlight] (23,12) rectangle (27.5,11.8);
        \path[highlight] (23,10.9) rectangle (27.5,10.7);
        \path[highlight] (23,9.8) rectangle (27.5,9.6);
        \path[highlight] (23,8.5) rectangle (27.5,8.7);
        \path[highlight] (23,7.6) rectangle (24.5,7.4);
        \path[highlight] (24.4,7.6) rectangle (27.5,7.5);
        \path[highlight] (23,6.5) rectangle (24.5,6.3);
        \path[highlight] (23,5.4) rectangle (24.5,5.2);
        \path[highlight] (25.5,13.5) rectangle (25.7,13);
        \path[highlight] (22,12.6) rectangle (22.2,13.5);
        \path[highlight] (23.1,13.5) rectangle (23.3,13);
        \path[highlight] (24.2,13.5) rectangle (24.4,13);
        \path[highlight] (21,13) rectangle (24.5,13.2);
        \path[highlight] (24.5,13.4) rectangle (26.5,13.5);
        \path[highlight] (20.5,13.1) rectangle (21,13.3);
        \path[highlight] (13.5,7.5) rectangle (15,7.7);
        \path[highlight] (15,8.8) rectangle (13.5,8.6);
        \path[highlight] (15,9.9) rectangle (13.5,9.7);
        \path[highlight] (13.5,11) rectangle (18.5,10.8);
        \path[highlight] (13.5,11.9) rectangle (18.5,12.1);
        \path[highlight] (13.5,13) rectangle (20.5,13.2);
        \path[highlight] (14.5,7) rectangle (14.7,13.5);
        \path[highlight] (15.6,10) rectangle (15.8,13.5);
        \path[highlight] (16.7,10) rectangle (16.9,13.5);
        \path[highlight] (18,10) rectangle (17.8,13.5);
        \path[highlight] (18.9,13.5) rectangle (19.1,12.5);
        \path[highlight] (20,13.5) rectangle (20.2,12.5);
        \path[highlight] (12.5,2.5) rectangle (12.7,1);
        \path[highlight] (13.6,2.5) rectangle (13.8,1);
        \path[highlight] (11.5,1.5) rectangle (14.5,1.3);
        \path[highlight] (10.8,6) rectangle (11,3.5);
        \path[highlight] (9.9,6) rectangle (9.7,1);
        \path[highlight] (11.5,5) rectangle (9.5,4.8);
        \path[highlight] (11.5,3.9) rectangle (9.5,3.7);
        \path[highlight] (9.5,2.6) rectangle (10.5,2.8);
        \path[highlight] (9.5,1.5) rectangle (10.5,1.7);
        \path[highlight] (12.9,7) rectangle (13.1,13.5);
        \path[highlight] (12,7) rectangle (11.8,13.5);
        \path[highlight] (10.9,6) rectangle (10.7,13.5);
        \path[highlight] (9.8,6) rectangle (9.6,13.5);
        \path[highlight] (11.5,6.7) rectangle (9.5,6.5);
        \path[highlight] (9.5,7.6) rectangle (13.5,7.8);
        \path[highlight] (9.5,8.9) rectangle (13.5,8.7);
        \path[highlight] (9.5,10) rectangle (13.5,9.8);
        \path[highlight] (13.6,10.9) rectangle (9.5,11.1);
        \path[highlight] (13.5,12) rectangle (9.5,12.2);
        \path[highlight] (13.5,13.1) rectangle (9.5,13.3);
        \path[highlight] (10.4,1) rectangle (10.6,3.6);
        \path[highlight] (10.4,3.6) rectangle (12.6,3.4);
        \path[highlight] (12.4,4.1) rectangle (12.6,2.4);
        \path[highlight] (11.6,1) rectangle (11.4,2.6);
        \path[highlight] (11.4,2.6) rectangle (14.6,2.4);
        \path[highlight] (14.4,1) rectangle (14.6,5.1);
        \path[highlight] (12.4,4.1) rectangle (14.6,3.9);
        \path[highlight] (14.4,5.1) rectangle (15.1,4.9);
        \path[highlight] (9.5,6.1) rectangle (11.6,5.9);
        \path[highlight] (11.4,7.1) rectangle (11.6,3.4);
        \path[highlight] (11.4,7.1) rectangle (15.1,6.9);
        \path[highlight] (14.9,2.9) rectangle (15.1,10.1);
        \path[highlight] (14.4,2.9) rectangle (17.6,3.1);
        \path[highlight] (17.4,1.4) rectangle (17.6,3.6);
        \path[highlight] (15.9,1.6) rectangle (16.1,1);
        \path[highlight] (15.9,1.6) rectangle (20.6,1.4);
        \path[highlight] (20.4,1) rectangle (20.6,2.1);
        \path[highlight] (14.9,3.4) rectangle (22.6,3.6);
        \path[highlight] (20.4,2.1) rectangle (23.6,1.9);
        \path[highlight] (23.4,1) rectangle (23.6,3.6);
        \path[highlight] (22.4,1.9) rectangle (22.6,5.1);
        \path[highlight] (22.4,5.1) rectangle (24.6,4.9);
        \path[highlight] (19.4,6.1) rectangle (19.6,3.4);
        \path[highlight] (19.4,6.1) rectangle (23.1,5.9);
        \path[highlight] (23.4,3.6) rectangle (26.6,3.4);
        \path[highlight] (25.4,1) rectangle (25.6,2.1);
        \path[highlight] (25.4,2.1) rectangle (26.6,1.9);
        \path[highlight] (26.4,1) rectangle (26.6,4.6);
        \path[highlight] (26.4,4.6) rectangle (27.5,4.4);
        \path[highlight] (24.4,3.4) rectangle (24.6,7.6);
        \path[highlight] (24.4,7.6) rectangle (27.5,7.4);
        \path[highlight] (21.4,5.9) rectangle (21.6,8.6);
        \path[highlight] (21.4,8.6) rectangle (22.1,8.4);
        \path[highlight] (21.9,5.9) rectangle (22.1,10.6);
        \path[highlight] (21.9,10.6) rectangle (23.1,10.4);
        \path[highlight] (22.9,4.9) rectangle (23.1,13.1);
        \path[highlight] (22.9,13.1) rectangle (27.5,12.9);
        \path[highlight] (14.9,10.1) rectangle (22.1,9.9);
        \path[highlight] (18.4,9.9) rectangle (18.6,12.6);
        \path[highlight] (18.4,12.6) rectangle (23.1,12.4);
        \path[highlight] (13.4,6.9) rectangle (13.6,13.5);
        \path[highlight] (20.4,12.4) rectangle (20.6,13.5);
        \path[highlight] (20.9,13.5) rectangle (21.1,12.4);
        \path[highlight] (24.4,12.9) rectangle (24.6,13.5);
        \path[highlight] (26.4,13.5) rectangle (26.6,12.9);
        \draw[lines] (15,10) -- (15,3.5) -- (19.5,3.5) -- (19.5,6) -- (21.5,6) -- (21.5,8.5) -- (22,8.5) -- (22,10) -- cycle;
        \draw[lines] (18.5,12.5) -- (23,12.5) -- (23,10.5) -- (22,10.5) -- (22,10) -- (18.5,10) -- cycle;
        \draw[lines] (22,10.5) rectangle (23,6);
        \draw[lines] (21.5,8.5) rectangle (22,6);
        \draw[lines] (19.5,6) -- (23,6) -- (23,5) -- (22.5,5) -- (22.5,3.5) -- (19.5,3.5) -- cycle;
        \draw[lines] (22.5,5) -- (24.5,5) -- (24.5,3.5) -- (23.5,3.5) -- (23.5,2) -- (22.5,2) -- cycle;
        \draw[lines] (11.5,7) -- (15,7) -- (15,5) -- (14.5,5) -- (14.5,4) -- (12.5,4) -- (12.5,3.5) -- (11.5,3.5) -- cycle;
        \draw[lines] (14.5,5) rectangle (15,3);
        \draw[lines] (12.5,4) rectangle (14.5,2.5);
        \draw[lines] (15,3.5) rectangle (17.5,3);
        \draw[lines] (17.5,3.5) -- (22.5,3.5) -- (22.5,2) -- (20.5,2) -- (20.5,1.5) -- (17.5,1.5) -- cycle;
        \draw[lines] (20.5,1) -- (20.5,2) -- (23.5,2) -- (23.5,1);
        \draw[lines] (16,1) -- (16,1.5) -- (17.5,1.5) -- (17.5,3) -- (14.5,3) -- (14.5,1);
        \draw[lines] (14.5,1) -- (14.5,2.5) -- (11.5,2.5) -- (11.5,1);
        \draw[lines] (10.5,1) -- (10.5,3.5) -- (11.5,3.5) -- (11.5,6) -- (9.5,6);
        \draw[lines] (25.5,1) -- (25.5,2) -- (26.5,2);
        \draw[lines] (26.5,1) -- (26.5,4.5) -- (27.5,4.5);
        \draw[lines] (24.5,3.5) -- (26.5,3.5) -- (26.5,4.5) -- (27.5,4.5);
        \draw[lines] (24.5,3.5) -- (24.5,7.5) -- (27.5,7.5);
        \draw[lines] (13.5,13.5) -- (13.5,7);
        \draw[lines] (20.5,12.5) -- (20.5,13.5);
        \draw[lines] (21,13.5) -- (21,12.5);
        \draw[lines] (23,12.5) -- (23,13) -- (27.5,13);
        \draw[lines] (24.5,13.5) -- (24.5,13);
        \draw[lines] (26.5,13.5) -- (26.5,13);
      \end{tikzpicture}
      \caption{The set $H(\epsilon)$ from Figure~\ref{fig:planetiling} is depicted in black.  The corresponding set $\Hgrid(\varepsilon)$ is highlighted in gray.}
      \label{fig:Hgridsketch}
    \end{center}
  \end{figure}
  
  \begin{lemma}
    Let $L, N, q \in \N$ with $L, N, q \ge 1$ and $\varepsilon$ with
    $0 < \varepsilon \le \frac{1}{2 N L}$ be given.  Let 
    \[
    	S := \tfrac{1}{N L} \{0,\ldots,L-1 \}^n.
    \]
    For $s \in S$, consider the shifted grid 
    \[ 
    	G(s) := \{ s + \tfrac{1}{N} v \mid v \in \calV \} \,.
    \]
    Assume that for some $s\in S$ we have $G(s) \cap \Hgrid(\varepsilon) = \emptyset$. 
    Then, this implies the following two conditions:
    \begin{enumerate}
      \item For every $v \in \calV$, there is exactly one $\hat{x} \in \hat{X}$ with $\hat{V}_{\hat{x}} \cap
        (s + \frac{1}{N} v + (-\varepsilon, \varepsilon)^n) \neq \emptyset$.
      \item Let 
      	$T:=\{ t \in \R^n \mid \exists v\in\calV : \Phi_\calI^{-1}\bigl( s + \tfrac{1}{N} v \bigr) = (x, t) \}$. 
      	Then, we have $T\cap \bigl( [-\varepsilon,\varepsilon]^n + \tfrac{1}{N} \N^n \bigr) = \emptyset$.
    \end{enumerate}
    These conditions show that we can compute $f$ correctly using \ASM3 if the precision~$2^{-k}$
    used there is at most $\frac{\varepsilon}{2}$: the first condition ensures that the $x$ part of
    $f(v) = (x, \floor{N t})$ can be computed exactly, and the second condition ensures that
    $\floor{N t}$ is exact.
  \end{lemma}
  
  \begin{proof}
    Observe that $G(s) \cap \Hgrid(\varepsilon) = \emptyset$ implies 
    $G(s) \cap H(\varepsilon) = \emptyset$.  We show that the latter implies the
    first condition of the lemma. The more general $G(s) \cap \Hgrid(\varepsilon) = \emptyset$ is needed to
    eliminate some sporadic cases in the second condition.
    \begin{enumerate}
      \item Since $\bigcup_{\hat{x} \in \hat{X}} \hat{V}_{\hat{x}} = \R^n$ there must be at least
        one such $\hat{x}$. Let $\hat{x} \in \hat{X}$ be one such element. In the case that $s +
        \frac{1}{N} v + (-\varepsilon, \varepsilon)^n$ is not completely contained in
        $\hat{V}_{\hat{x}}$, the translated open disc $s + \frac{1}{N} v + (-\varepsilon,
        \varepsilon)^n$ must contain some $y \in \partial \hat{V}_{\hat{x}}$. But this implies that
        $G(s) \ni s + \frac{1}{N} v \in y + [-\varepsilon, \varepsilon]^n \subseteq H(\varepsilon)$,
        contradicting $G(s) \cap H(\varepsilon) = \emptyset$. Thus $s + \frac{1}{N} v +
        (-\varepsilon, \varepsilon)^n \subseteq \hat{V}_{\hat{x}}$ and we are done.
      \item Assume that $t \in T$ can be written as $t = \frac{1}{N} w + e$ with $w \in \N^n$ and $e
        \in [-\varepsilon, \varepsilon]^n$, i.e., $t \in T \cap ([-\varepsilon, \varepsilon]^n
        + \frac{1}{N} \N^n)$. As $t \in T$ there exists some $u \in G(s)$ with $\Phi_\calI^{-1}\bigl(\pi(u)\bigr) = (x,
        t)$ for $x \in X$. Let $\hat{x} \in \hat{X}$ with $u \in \hat{V}_{\hat{x}}$; then $u =
        \hat{x} + \frac{1}{N} w + e$. But this yields $u \in (\frac{1}{N} \N^n + \partial
        \hat{V}_{\hat{x}}) \cap \overline{\hat{V}_{\hat{x}}} + [-\varepsilon, \varepsilon]^n$ and
        thus $u \in \Hgrid(\varepsilon)$. Hence, $u \in G(s) \cap \Hgrid(\varepsilon)$. \qedhere
    \end{enumerate}
  \end{proof}
  
  We now determine a lower bound on the probability that the desired condition $G(s) \cap \Hgrid(\varepsilon)=\emptyset$ holds 
  when $s$ is chosen uniformly at random in $S$ and $L$ is sufficiently large.
    
  \begin{proposition}
    \label{prop:s-avoidborder}
    Let $q,N\in\N$ with $q,N\ge 1$ and $p \in (0, 1)$ be given.  Choose $L$ and $\varepsilon$ such that
    \[ 
    	L \ge \frac{2 n D (q + A + C + 2)^n}{(1 - p) C^n} \;\; \text{and} \;\; \varepsilon \le \tfrac{1}{2NL}.
    \]
    If we select $s\in S$ uniformly at random, then 
    \[
    	\Pr\bigl( G(s) \cap \Hgrid(\varepsilon) = \emptyset \bigr) \ge p.
    \]
  \end{proposition}
  
  The main idea behind the proof of this proposition is a follows: while $\partial
  \hat{V}_{\hat{x}}$ for a single $\hat{x}$ can be difficult to describe, the union of all $\partial
  \hat{V}_{\hat{x}}$, where $\hat{x}$ ranges over all $\hat{x} \in \hat{X}$, has a much simpler
  structure. For instance, in the case of $n = 2$, i.e., in the plane, it suffices to consider only two
  faces of $\partial \hat{V}_{\hat{x}}$, namely, the ones incident with $\hat{x}$. Let us
  call these two faces the \emph{principal boundaries} of $\hat{V}_{\hat{x}}$. Every boundary point
  is an element of a principal boundary of at least one $\hat{V}_{\hat{x}}$. The principal
  boundaries of some $\hat{V}_{\hat{x}}$ from Figure~\ref{fig:planetiling} are shown in
  Figure~\ref{fig:principalborders}; note that we capped off the ends of the principal boundaries to
  make it possible to distinguish between different principal boundaries. The corners of the sets
  are marked by large dots, and the principal boundaries by thick lines. The proof works by covering
  the principal boundaries by larger sets which are known to cover them -- for this, we need
  assumption~\ASM1.
  
  \begin{figure}[Hh]
    \begin{center}
      \begin{tikzpicture}[y=-1cm, scale=0.5]
        \draw[background] (14.2,3.2) rectangle (24,11.2);
        \path[highlight] (14.9,3.4) rectangle (22.6,3.6);
        \path[highlight] (19.4,6.1) rectangle (19.6,3.4);
        \path[highlight] (21.4,5.9) rectangle (21.6,8.6);
        \path[highlight] (21.4,8.6) rectangle (22.1,8.4);
        \path[highlight] (21.9,5.9) rectangle (22.1,10.6);
        \path[highlight] (21.9,10.6) rectangle (23.1,10.4);
        \path[highlight] (14.9,10.1) rectangle (22.1,9.9);
        \path[highlight] (14.2,7.1) rectangle (15.1,6.9);
        \path[highlight] (14.2,4.1) rectangle (14.6,3.9);
        \path[highlight] (14.4,3.2) rectangle (14.6,5.1);
        \path[highlight] (14.9,3.2) rectangle (15.1,10.1);
        \path[highlight] (17.4,3.2) rectangle (17.6,3.6);
        \path[highlight] (22.4,3.2) rectangle (22.6,5.1);
        \path[highlight] (22.4,5.1) rectangle (24,4.9);
        \path[highlight] (22.9,4.9) rectangle (23.1,11.2);
        \path[highlight] (18.4,9.9) rectangle (18.6,11.2);
        \path[highlight] (14.4,5.1) rectangle (15.1,4.9);
        \path[highlight] (19.4,6.1) rectangle (23.1,5.9);
        \path[fatdots] (15,10) circle (0.1cm);
        \path[fatdots] (14.5,5) circle (0.1cm);
        \path[fatdots] (15,3.5) circle (0.1cm);
        \path[fatdots] (17.5,3.5) circle (0.1cm);
        \path[fatdots] (19.5,6) circle (0.1cm);
        \path[fatdots] (21.5,8.5) circle (0.1cm);
        \path[fatdots] (22.5,5) circle (0.1cm);
        \path[fatdots] (22,10.5) circle (0.1cm);
        \draw[fatlines] (19.5,3.6) -- (19.5,6) -- (23,6);
        \draw[fatlines] (21.9,8.5) -- (21.5,8.5) -- (21.5,6.1);
        \draw[fatlines] (23,10.5) -- (22,10.5) -- (22,6.1);
        \draw[fatlines] (24,5) -- (22.5,5) -- (22.5,3.2);
        \draw[fatlines] (15,3.2) -- (15,3.5) -- (17.3,3.5);
        \draw[fatlines] (14.5,3.2) -- (14.5,5) -- (14.9,5);
        \draw[fatlines] (21.9,10) -- (15,10) -- (15,3.7);
        \draw[fatlines] (17.5,3.2) -- (17.5,3.5) -- (22.4,3.5);
      \end{tikzpicture}
      \caption{Example showing the principal boundaries of some of the cornered sets of $H$ from Figure~~\ref{fig:planetiling}.  The principal boundaries are depicted in black.  We capped them in the 	
      figure to make clear to which corners they belong.}
      \label{fig:principalborders}
    \end{center}
  \end{figure}
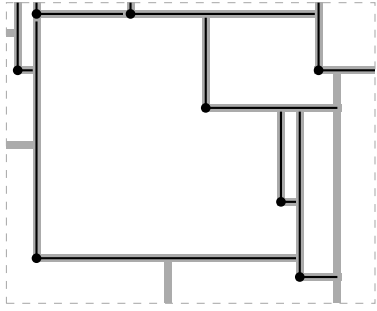
  
  \begin{proof}[Proof of Proposition~\ref{prop:s-avoidborder}.]
    Define
    \begin{eqnarray*}
    	\hat{X}' & := & \hat{X} \cap [-A - \varepsilon, q + 1 + \varepsilon]^n \\
    	F(s)     & := & (s + \frac{1}{N} \Z^n) \cap (\bigcup_{\hat{x} \in \hat{X}'}
                      \partial \hat{V}_{\hat{x}} + [-\varepsilon, \varepsilon]^n).
    \end{eqnarray*}
    We first show that $F(s) = \emptyset$ implies the desired condition $G(s) \cap \Hgrid(\varepsilon)=\emptyset$ 
    and then bound the number of $s$ in $S$ for which it may be the case that $F(s) \neq \emptyset$.
    
    We prove the first part by considering the contraposition of the implication.  
    Assume that some $u \in G(s) \cap \Hgrid(\varepsilon)$ exists. Then there
    exists some $\hat{x} \in \hat{X}$ and some $e \in [-\varepsilon, \varepsilon]^n$ such that $u -
    e \in (\tfrac{1}{N} \N^n + \partial \hat{V}_{\hat{x}}) \cap \overline{\hat{V}_{\hat{x}}}$, and
    $u = s + \frac{1}{N} v$ with $v \in \calV = \Z^n \cap [0, q N - 1]^n$. In particular, $u \in [0,
      q - \frac{1}{NL}]^n$ and hence $u - e \in [-\varepsilon, q + \varepsilon]^n$.
    
    We have $\overline{\hat{V}_{\hat{x}}} \subseteq \hat{x} + [0, A]^n$ since 
    $\hat{V}_{\hat{x}}$ is cornered with corner $\hat{x}$ and \ASM1 holds. This implies 
    $\hat{x} \in u - e - [0,A]^n \subseteq [-\varepsilon - A, q + \varepsilon]^n$ and 
    hence $\hat{x} \in \hat{X}'$. As $u - e \in \tfrac{1}{N} \N^n + \partial \hat{V}_{\hat{x}}$, 
    we can write $u - e = b + \frac{1}{N}w$ with $b \in \partial \hat{V}_{\hat{x}}$ and $w \in \N^n$. 
    But this yields that $b + e\in F(s)$ and hence $F(s) \neq \emptyset$.  
    
    We now bound the number of $s \in S$ with $F(s) \neq \emptyset$. For $\hat{x} \in
    \hat{X}$, define
    \[ 
    	H'(\hat{x}, \varepsilon, i) := \{ \hat{x} + (t_1, \dots, t_n) \mid 
    	-\varepsilon \le t_j \le A + \varepsilon, \; -\varepsilon \le t_i \le \varepsilon \} 
    \] 
    and $H'(\hat{x}, \varepsilon) := \bigcup_{i=1}^n H'(\hat{x}, \varepsilon, i)$. This set
    $H'(\hat{x}, \varepsilon)$ covers the enhanced principal boundaries of $\hat{V}_{\hat{x}}$.  The
    observation on which this proof is based (see Figure~\ref{fig:principalborders}) can now be
    expressed by
    \[
    	H(\varepsilon) \subseteq \bigcup_{\hat{x} \in \hat{X}} H'(\hat{x}, \varepsilon), 
    \] 
    which implies 
    \[
    	F(s) \subseteq \bigcup_{\hat{x} \in \hat{X}'} H'(\hat{x}, \varepsilon) \cap (s +
        \tfrac{1}{N} \Z^n).
    \]
    We count the number of $s$ for which it may be the case that $H'(\hat{x}, \varepsilon) \cap (s +
    \frac{1}{N} \Z^n) \neq \emptyset$ for a fixed $\hat{x}$, and then multiply this by an upper bound
    on the the number of elements in $\hat{X}'$. This allows us to obtain the formula from the theorem statement.
    
    To obtain a upper bound on the cardinality $|\hat{X}'|$, we use \ASM2. Since $\hat{X}'$ is contained in at
    most $\bigl( \frac{q + A + 1 + 2 \varepsilon}{C} + 1 \bigr)^n$ blocks of size $C$, $|\hat{X}'|
    \le D \cdot \bigl( \frac{q + A + 1 + 2 \varepsilon}{C} + 1 \bigr)^n$ by \ASM2.
    
    Now let $\hat{x} \in \R^n$ be arbitrary. As $\varepsilon \le \frac{1}{2 N L}$, there are at most
    $2 L^{n-1}$ choices for $s \in S$ with $H'(\hat{x}, \varepsilon, i) \cap (s + \frac{1}{N} \Z^n)
    \neq \emptyset$. This shows that there are at most 
    \[ 
    	2 L^{n-1} \cdot n \cdot D \cdot \left(
    		\frac{q + A + 1 + 2 \varepsilon}{C} + 1 \right)^n 
    \] 
    bad choices for $s \in S$, while $\abs{S} = L^n$. 
    This, together with $\varepsilon \le \frac{1}{2 N L}$, yields that the probability for
    $G(s) \cap H(\varepsilon) = \emptyset$ is at least 
    \[ 1 - \frac{1}{L} \cdot 2 n D \left( \frac{q
      + A + 2 + C}{C} \right)^n. \qedhere 
    \]
  \end{proof}
  
  \begin{corollary}
    \label{cor:computefcorrectly}
    Let $N, q\in\N$ with $q,N\ge 1$ be given.  Choose $L$ and $\varepsilon$ such that
    \begin{equation}\tag{I}\label{eq:I}
    	L \ge \frac{4 n D (q + A + C + 2)^n}{C^n} \mbox{ and } \varepsilon\le\tfrac{1}{2NL}.
    \end{equation}
    If we select $s \in S$ uniformly random, then
    \[
    	\Pr\bigl( G(s) \cap \Hgrid(\varepsilon) = \emptyset \bigr) \ge \frac{1}{2}.
    \]
    This implies that we can compute $f(v)$ correctly
    for all $v\in \calV$ and thus prepare the state
    \[
    	\frac{1}{\sqrt{V}} \sum_{v\in\calV} |v\> |f(v)\>
    \]
    in step 1 with probability greater or equal to $\tfrac{1}{2}$.  
  \end{corollary}

  \section{Preparing periodic states}
  \label{sec:probperiodic}
  
  The original function $\Phi_\calI^{-1} \circ \pi : \R^n \to \R^n/\Lambda \to \fRep(\calI)$ is
  perfectly periodic with period lattice~$\Lambda$: if $u \in \R^n$ maps to $(x, t) \in
  \fRep(\calI)$, then $u + \lambda$ will also map to $(x, t)$ for all $\lambda \in \Lambda$. 
  
  Due to precision issues we have to work with the function~$f : \Z^n \to \fRep(\calI)$ defined by $v \mapsto (x, \floor{N t})$ if
  $(\Phi_\calI^{-1} \circ \pi)(s + \frac{1}{N} v) = (x, t)$. As $N \lambda$ will most certainly not
  have integral coordinates for $\lambda \in \Lambda$, we cannot directly obtain the collision
  $f(v) = f(v + N \lambda)$. And, if we round the coordinates of $N \lambda$ down or up to the nearest integers,
  it might happen that $f(v + \lfloor N \lambda \rceil)$ yields an $f$-representation~$(x', \floor{N t'})$ with
  $x \neq x'$ -- no matter to which integers we round the coordinates of $N \lambda$.
  
  The first proposition of this section establishes a lower bound on the fraction of grid points for which this
  problem does not occur. For these grid points, the corresponding $f$-representation $(x, t)$ is
  sufficiently far away from the boundaries, meaning that we remain in the same (translated) region~$\hat{V}_{\hat{x}}$ when adding 
  a suitably rounded version of $N \lambda$. 
  
  Similarly to the argument used in the
  proof of Proposition~\ref{prop:s-avoidborder} in the previous section, we estimate the number of
  grid points lying in regions that are too close to the principal boundaries.  The union of all such regions is denoted by $\Hbound$. An example for $n=2$ is shown in
  Figure~\ref{fig:hatHfigure} with $\Hbound$ highlighted.
  
  \begin{proposition}
    Assume that $s \in S$ with $H(\varepsilon) \cap G(s) = \emptyset$ in the notation of the previous
    section. Consider $\Hbound := (-\tfrac{1}{N}, 0]^n + H$. Then \[ \frac{|G(s) \setminus
    \Hbound|}{\abs{G(s)}} \ge 1 - \frac{1}{N} \cdot \frac{n D (q + 1 + A + C)^n (A + 2/N)^{n-1}}{(C
    q)^n}. \]
  \end{proposition}
  
  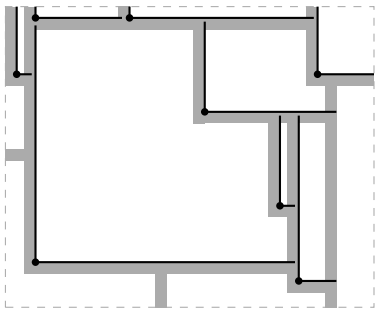
\begin{figure}[Hh]
    \begin{center}
      \newcommand{\mythickness}{0.3}
      \begin{tikzpicture}[y=-1cm, scale=0.5]
        \draw[background] (14.2,3.2) rectangle (24,11.2);
        \path[highlight] (15-\mythickness,3.5) rectangle (22.5,3.5+\mythickness);
        \path[highlight] (14.5-\mythickness,5+\mythickness) rectangle (15,5);
        \path[highlight] (14.5-\mythickness,3.2) rectangle (14.5,5+\mythickness);
        \path[highlight] (17.5-\mythickness,3.2) rectangle (17.5,3.5+\mythickness);
        \path[highlight] (19.5-\mythickness,6+\mythickness) rectangle (19.5,3.5);
        \path[highlight] (19.5-\mythickness,6+\mythickness) rectangle (23,6);
        \path[highlight] (22.5-\mythickness,3.2) rectangle (22.5,5+\mythickness);
        \path[highlight] (23-\mythickness,5) rectangle (23,11.2);
        \path[highlight] (22.5-\mythickness,5+\mythickness) rectangle (24,5);
        \path[highlight] (21.5-\mythickness,6) rectangle (21.5,8.5+\mythickness);
        \path[highlight] (21.5-\mythickness,8.5+\mythickness) rectangle (22,8.5);
        \path[highlight] (22-\mythickness,6) rectangle (22,10.5+\mythickness);
        \path[highlight] (22-\mythickness,10.5+\mythickness) rectangle (23,10.5);
        \path[highlight] (18.5-\mythickness,10) rectangle (18.5,11.2);
        \path[highlight] (15-\mythickness,3.2) rectangle (15,10+\mythickness);
        \path[highlight] (15-\mythickness,10+\mythickness) rectangle (22,10);
        \path[highlight] (14.2,7+\mythickness) rectangle (15,7);
        \path[dots] (15,10) circle (0.1cm);
        \path[dots] (14.5,5) circle (0.1cm);
        \path[dots] (15,3.5) circle (0.1cm);
        \path[dots] (19.5,6) circle (0.1cm);
        \path[dots] (21.5,8.5) circle (0.1cm);
        \path[dots] (22.5,5) circle (0.1cm);
        \path[dots] (22,10.5) circle (0.1cm);
        \path[dots] (17.5,3.5) circle (0.1cm);
        \draw[fatlines] (21.9,8.5) -- (21.5,8.5) -- (21.5,6.1);
        \draw[fatlines] (23,10.5) -- (22,10.5) -- (22,6.1);
        \draw[fatlines] (24,5) -- (22.5,5) -- (22.5,3.2);
        \draw[fatlines] (15,3.2) -- (15,3.5) -- (17.3,3.5);
        \draw[fatlines] (21.9,10) -- (15,10) -- (15,3.7);
        \draw[fatlines] (17.5,3.2) -- (17.5,3.5) -- (22.4,3.5);
        \draw[fatlines] (19.5,3.6) -- (19.5,6) -- (23,6);
        \draw[fatlines] (14.5,3.2) -- (14.5,5) -- (14.9,5);
      \end{tikzpicture}
      \caption{The region $\Hbound$ is highlighted. For some of the cornered sets, the principal
        boundaries are shown.}
      \label{fig:hatHfigure}
    \end{center}
  \end{figure}
  
  \begin{proof}
    Clearly $\abs{G(s)} = (N q)^n$. The cardinality of $G(s) \cap \Hbound$ can be estimated by
    counting the number of cubes of the form $(-\frac{1}{N}, 0]^n$ needed to cover $\Hbound \cap [0,
    q]^n$. For $\hat{x} \in \hat{X}$, define
    \[ 
    	H''(\hat{x}, i) := \{ \hat{x} + (t_1, \dots, t_n) \mid -\tfrac{1}{N} < t_j \le A, \;
    		-\tfrac{1}{N} < t_i \le 0 \} 
    \]
    and $\hat{X}'' := \hat{X} \cap [-A, q + 1]^n$. Then
    \[ 
    	\Hbound           \subseteq \bigcup_{\hat{x} \in \hat{X}}   \bigcup_{i=1}^n H''(\hat{x}, i) \quad \text{and} \quad 
    	G(s) \cap \Hbound \subseteq \bigcup_{\hat{x} \in \hat{X}''} \bigcup_{i=1}^n H''(\hat{x}, i). 
    \]
    Now $H''(\hat{x}, i)$ can be covered by $\ceil{N A + 1}^{n-1}$ such cubes, whence the total
    number of cubes needed is less or equal to
    \[ 
    	|\hat{X''}| \cdot n \cdot \ceil{N A + 1}^{n-1}. 
    \]
    As above, $\hat{X}''$ is contained in at most $(\frac{q + 1 + A}{C} + 1)^n$ blocks of size~$C$,
    whence
    \[ 
    	\bigl|\hat{X}'' \bigr| \le \frac{D}{C^n} (q + 1 + A + C)^n. 
    \]
    Therefore,
    \[ 
    	\frac{|G(s) \cap \Hbound|}{\abs{G(s)}} \le \frac{\frac{D}{C^n} (q + 1 + A + C)^n \cdot n
      \cdot \ceil{N A + 1}^{n-1}}{(N q)^n}. \qedhere 
    \]
  \end{proof}
  
  We now give an explicit lower bound on $N$ guaranteeing that the
  fraction of grid points not contained in $\Hbound$ is sufficiently large.
  
  \begin{corollary}
    \label{corr:insideprob12}
    We may assume w.l.o.g.~that $C \le 1$.  Choose $q$ and $N$ such that
    \begin{equation}\tag{II}\label{eq:II}
    	q \ge 9 \max\{ 1, A \}, \;\; \text{and} \;\; N \ge \max\left\{ \frac{4}{A}, \;
        \frac{8 (n + 1) n \cdot 2^n D A^{n-1}}{3 C^n} \right\}.
    \end{equation}
    If $s \in S$ is such that $\Hgrid(\varepsilon) \cap G(s) = \emptyset$, then 
    \[ 
    	\frac{1}{N} \le \frac{A}{4} \qquad \text{and} \qquad \frac{|G(s) \setminus
      \Hbound|}{\abs{G(s)}} \ge 1 - \frac{1}{4 (n + 1)}. 
     \]
  \end{corollary}
  
  Note that we can always decrease $C$ without invalidating assumption~\ASM2. Moreover, the
  recommended choice of $C$ in Proposition~\ref{prop:infrafromFFNF} for infrastructures obtained
  from function fields or number fields satisfies $C \le 1$.
  
  \begin{proof}
    Using $N \ge \frac{4}{A}$, the complement probability can be bounded by
    \begin{align*}
      & \frac{1}{N} \cdot \frac{n D (q + 1 + A + C)^n (A + 2/N)^{n-1}}{(C q)^n} \\
      {}\le{} & \frac{1}{N} \cdot \frac{n D (1 + 1/q + A/q + C/q)^n (A + A/2)^{n-1}}{C^n}.
    \end{align*}
    As $q \ge 9 \max\{ 1, A \}$, we have $1/q + A/q \le \frac{2}{9}$, and as $C \le 1$, we have $C/q
    \le \frac{1}{9}$. Therefore, $1 + 1/q + A/q + C/q \le 1 + \frac{1}{3} = \frac{4}{3}$, whence $(1
    + 1/q + A/q + C/q)^n (1 + \frac{1}{2})^{n-1} \le 2^n \cdot \frac{2}{3}$. Finally, the choice of
    $N$ ensures that the complement probability is bounded by $\frac{1}{4 (n + 1)}$ from above.
  \end{proof}
  
  Let $\calF\subset X \times \N$ be the set of rounded $f$-representations. More precisely, it is defined by
  \[
  	(x,k) \in \calF \quad\text{iff}\quad \text{there exists $v\in\calV$ with $f(v)=(x,k)$ and $s+\frac{1}{N} v \notin\Hbound$.}
  \]
  
  \begin{lemma}\label{lem:awayfromborder}
  	Choose $q$ and $N$ according to (II).  Assume $s \in S$ is such that $G(s) \cap H(\varepsilon) = \emptyset$. 
    Let $(x,k)$ be the measurement outcome obtained in step 2 of the quantum algorithm.
    Then, 
    \[
    	\Pr\bigl( (x,k)\in\calF \bigr) \ge 1 - \frac{1}{4(n+1)}.
    \]
  \end{lemma}
  
  \begin{proof}
  	For a fixed pair~$(x', k')$, the probability that this pair is sampled is
    $\frac{1}{\abs{\calV}}\abs{f^{-1}(x', k')}$. Let $\calA$ be the set of elements~$v \in \calV$ with
    $s + \frac{1}{N} v \not\in \Hbound$; then 
    \[ \calA \subseteq \bigcup_{(x', t') \in \calF}
    f^{-1}(x', t'), 
    \] 
    whence the probability we want to estimate can be bounded from below by
    $\frac{1}{\abs{\calV}} \abs{\calA}$. But this quantity equals $\frac{|G(s) \setminus
      \Hbound|}{|G(s)|}$, and by Corollary~\ref{corr:insideprob12} it can be bounded from below by
    $1 - \frac{1}{4 (n + 1)}$. 
  \end{proof}
  
  The next proposition makes precise statements on the periodicity of grid elements outside $\Hbound$. 
  First, we show that if $f(v) = f(v')$, then $\frac{1}{N} (v' - v)$ yields an approximation 
  of some element $\lambda \in \Lambda$. Second, we show that for every $\lambda \in \Lambda$ such that 
  $v+N \lambda$ stays within the boundaries of the grid there exists a unique $v'$ with $f(v) = f(v')$ and $\frac{1}{N} (v' - v) \approx \lambda$. Finally,
  we estimate the number of collisions for one specific $v$, i.e., the numbers of $v'$
  in the grid such that $f(v) = f(v')$.
  
  \begin{proposition}
    \label{prop:Mestimate}
    Choose $q$ and $N$ such that
    \begin{align*}
      N & \ge \frac{2 \sqrt{n}}{\lambda_1(\Lambda)} \tag{III}\label{eq:III}\\
      q & >   2n\nu(\Lambda) + \frac{3n}{N}\,. \tag{IV}\label{eq:IV}
    \end{align*}
    Assume that $s\in S$ is such that $G(s) \cap \Hgrid(\frac{1}{2 N L}) = \emptyset$.  Let $v\in\calV$ be such that
    $f(v)$ is equal to the measurement outcome.  Assume that $s + \frac{1}{N} v \not\in \Hbound$.  Let 
    $\calM = \{ v' \in \calV \mid f(v') = f(v) \}$ and $M=\abs{\calM}$.
    \begin{enumi}
      \item Let $v' \in \calM$. We have $\norm{(v - v') - N
      \lambda}_\infty < 1 - \frac{1}{L}$ for a unique $\lambda \in \Lambda$.
      \item Let $\lambda \in \Lambda$ such that $v + N \lambda \in [1,q N - 2]^n$. Then, there exists
      a unique $v' \in \calM$ satisfying $\norm{(v - v') - N \lambda}_\infty < 1 - \frac{1}{L}$.
      \item We have $M \ge M_\ell$, where 
      \begin{eqnarray*}
      	M_\ell 
      	& = & 
      	\frac{q^n}{\det(\Lambda)} \left(
      		1 - \frac{3n}{qN} - \frac{2n\nu(\Lambda)}{q} 
      	\right).
      \end{eqnarray*}
    \end{enumi}
  \end{proposition}
  
  \begin{proof}\hfill
    \begin{enumi}
      \item Let $(\Phi_\calI^{-1} \circ \pi)(s + \tfrac{1}{N} v) = (x, t)$ and 
      $(\Phi_{\calI}^{-1} \circ \pi)(s + \tfrac{1}{N} v') = (x', t')$; 
      then $f(v) = (x, \floor{N t})$ and $f(v') = (x', \floor{N t'})$. 
      Note that in $\R^n/\Lambda$, we have
      $\dist(x) + t   = s + \tfrac{1}{N} v$ and 
      $\dist(x') + t' = s + \tfrac{1}{N} v'$, 
      whence $\dist(x) + t - (\dist(x') + t') = \frac{1}{N} (v - v')$.
      
      We have $f(v) = f(v')$. Therefore, $x = x'$ and $\floor{N t} = \floor{N t'}$, which
      yields $\norm{t - t'}_\infty < \frac{1}{N}$. By the assumption that $G(s) \cap
      \Hgrid(\frac{1}{2 N L}) = \emptyset$, we have that the coefficients and $N t$ and $N t'$ are
      bounded away from an integer by at least $\frac{1}{2 L}$ (compare
      Corollary~\ref{cor:computefcorrectly}~2), whence we actually have $\norm{t - t'}_\infty <
      \frac{1}{N} - \frac{1}{N L}$.
      
      Now $t - t' = \dist(x) + t - (\dist(x') + t') = \frac{1}{N} (v - v')$ in $\R^n / \Lambda$,
      whence there exists some $\lambda \in \Lambda$ such that $v - v' = N (t - t') + N \lambda$.
      \item Let $(\Phi_\calI^{-1} \circ \pi)(s + \tfrac{1}{N} v) = (x, t)$; then $f(v) = (x,
      \floor{N t})$ and $\dist(x) + t = s + \tfrac{1}{N} v$.  Set $u := v + N \lambda$; then
      $\dist(x) + t = s + \frac{1}{N} u$ as an element of $\R^n / \Lambda$, whence $(x, t) =
      (\Phi_\calI^{-1} \circ \pi)(s + \frac{1}{N} u)$.
      
      There are at most two choices for each coordinate of the vector $e \in (-1, 1)^n$ such that $u
      + e$ has only integral coefficients. For each coordinate, there is exactly one choice if only
      $0$ can be chosen; otherwise, there exists one choice~$a \in (-1, 0)$ and the other is $1 +
      a$. Hence, there exists a unique $e \in (-1, 1)^n$ such that $\floor{N t} = \floor{N t + e}$
      and $v' := u + e \in \Z^n$.
      
      Clearly, $t + \frac{1}{N} e \ge 0$. First, $(x, t + \frac{1}{N} e) \in \fRep(\calI)$ since $s
      + \frac{1}{N} v \not\in \Hbound$. Second, $\dist(x) + (t + \frac{1}{N} e) = s + \frac{1}{N}
      v'$ implies $f(v') = (x, \floor{N t + e}) = (x, \floor{N t}) = f(v)$. Third, $v' \in \calV$
      since $u \in [1, q N - 2]^n$.
      
      It remains to show that $v'$ is unique. Assume that $v', v'' \in \calV$ satisfy $f(v') =
      f(v'')$, $\| (v - v') - N \lambda \|_\infty < 1 - \tfrac{1}{L}$, and $\| (v - v'') - N \lambda
      \|_\infty < 1 - \tfrac{1}{L}$.
      
      By (i) of this proposition, the condition $f(v)=f(v')$ implies that there exists some
      $\lambda' \in \Lambda$ with $\norm{(v' - v'') - N \lambda'}_\infty < 1$. By the triangle
      inequality, the two above conditions on the norms imply that $\norm{v' - v''}_\infty <
      2$. Since $v' - v'' \in \Z^n$, this yields $\norm{v' - v''}_\infty \le 1$.
      
      By applying the triangle inequality again and dividing by $N$, we conclude that
      $\norm{\lambda'}_\infty < \frac{2}{N}$.  Now, if $v' \neq v''$, then $\norm{(v' - v'') - N
      \lambda'}_\infty < 1$ would imply that $\lambda' \neq 0$. Then, $0 < \norm{\lambda'}_2 <
      \sqrt{n} \cdot \frac{2}{N}$ would hold.  But, this would violate $\lambda_1(\Lambda) \ge
      \frac{2 \sqrt{n}}{N}$, which follows from (III).  Therefore, we must have $v' = v''$ and,
      thus, $v'$ is unique.
      \item Using (ii), we see that a lower bound $M_\ell$ on $M$ is given by the cardinality of $N \Lambda
      \cap (-v + [1, q N - 2]^n)$. Let $\nu(N\Lambda)$ be the covering radius of $N\Lambda$.  
      Let $\lambda\in N\Lambda$.  If $\lambda\in (-v + [1+\nu(N\Lambda), q N - 2 - \nu(N\Lambda)]^n)$, then the
      Voronoi cell $\hat{V}_{N\Lambda}(\lambda)$ of $\lambda$ is entirely contained in $\bigl(-v + [1, q N
        - 2]^n\bigr)$.  As the volume of $\hat{V}_{N\Lambda}(\lambda)$ is $\det(N \Lambda)$, this
      yields the lower bound
      \begin{eqnarray*}
      	\frac{\bigl(qN-3-2\nu(N\Lambda)\bigr)^n}{\det(N\Lambda)} & \ge &
      	\frac{q^n}{\det(\Lambda)} \left(
      		1 - \frac{3n}{qN} - \frac{2n\nu(\Lambda)}{q} 
      	\right),
      \end{eqnarray*}
      which is greater than $0$ provided that assumption~(III) holds.
      \qedhere
    \end{enumi}
  \end{proof}
  
  \section{Sampling approximations of vectors of the dual lattice $\Lambda^*$}
  \label{sec:fouriersampling}
  
  \subsection{Sampling in dimension greater than one}
	
	We present here our new method of sampling approximation of the vectors of the dual lattice $\Lambda^\ast$, which improves the success probability of the overall algorithm by at least the exponential
	factor $2^{n^2-1}$. 
	
	We determine the probability that the quantum algorithm outputs a $w\in\calW$ such that 
  $\frac{1}{2 n q} w$ is sufficiently close to some $\lambda^* \in \Lambda^*$.
  We have to impose certain conditions on $w$ to be able to show that the probability of observing a good approximation is bounded away from $0$.  For $\lambda^*\in\Lambda^*$, let 
	\[
  	\calR_{\lambda^*} = \Big\{(w_1,\ldots,w_n) \Bigm| w_k \in \{ \floor{2nq \lambda^*_k},
        \floor{2nq \lambda^*_k} + 1 \} \mbox{ for $k=1,\ldots,n$} \Big\}.
  \]  
  Observe that for all $w\in\calR_{\lambda^*}$, we have
  \[
  	\norm{\frac{w}{2nq} - \lambda^*}_2 \le \frac{1}{2\sqrt{n}q}.
  \]
  The following proposition gives a lower bound on the probability of observing elements of $\calR_{\lambda^*}$ provided that 
  $\calR_{\lambda^*}\subset [0,2nq\kappa N]^n$, where $\kappa\in (0,1)$.
  
  In the remainder of this section, we make the two following assumptions:
  \begin{enumi}
  	\item the random shift $s\in S$ is such that $G(s)\cap\Hgrid(\tfrac{1}{2NL})=\emptyset$ and
  	\item all measurement outcomes $f(v)$ are such that $s+\tfrac{v}{N}\not\in\Hbound$.
  \end{enumi}
  The relevant results can be stated in a more direct way if we do not have to include these two assumptions in the formulation of the propositions.
  Note that we can estimate the probabilities that they are satisfied with the help of Corollary~\ref{cor:computefcorrectly} and
  Lemma~\ref{lem:awayfromborder}.  These will be included in the final analysis of the algorithm.
  
	\begin{proposition}
    \label{prop:lbprobone}
    Choose $q$ and $N$ according to (III) and (IV).  Choose $\kappa$ 
    such that 
    \begin{equation}\tag{V}\label{eq:V}
    	\kappa < \frac{1}{8n} - \frac{1}{4 n q N}.
    \end{equation}
    Then, for all $\lambda^\ast\in\Lambda^\ast$ with $\calR_{\lambda^*}\subset [0,2qn\kappa N]^n$, we have the lower bound
  	\begin{eqnarray*}
  		\Pr\bigl( \calR_{\lambda^*} \bigr) 
  		& = &
  		\sum_{w\in\calR_{\lambda^*}} \Pr(w) \\
  		& = &
  		\sum_{w\in\calR_{\lambda^*}} 
  		\left|
  			\frac{1}{\sqrt{M W}} 
  			\sum_{v' \in \calM} 
  			\exp\biggl(2 \pi i \, v' \cdot \frac{w}{2 n q N} \biggr) 
  		\right|^2 \\
  		& \ge & 
  		\frac{2^{n-1} M_\ell}{W} \cos^2\Bigl( \pi \big(\tfrac{1}{4} + \tfrac{1}{2qN} + 2 \kappa n \big) \Bigr).
  	\end{eqnarray*}
  \end{proposition}
		
	\begin{proof}
	  Let $v'$ be an arbitrary but fixed element of $\calM$. Proposition~\ref{prop:Mestimate} (i) shows that $\norm{v' - v - N
    \lambda}_\infty < 1 - \frac{1}{L}$ for some $\lambda \in \Lambda$ since condition (III) is satisfied.
    Define the error terms $e_1(v') = v' - v - N \lambda$ and $e_2(w) = \frac{w}{2 n q N} - \frac{\lambda^\ast}{N}$ for $w\in\calR_{\lambda^*}$.
    Both error types arise because both the rescaled lattice $N\Lambda$ and the dual lattice $\Lambda^*$ are not necessarily integral.
    
    To be able to show that the probability of observing a $w\in\calR_{\lambda^*}$ is bounded away from zero by a constant, 
    we have (i) to carry out the Fourier transform over a larger window and (ii) to disregard $w$ whose infinity-norm is too large.  These two 
    measures makes it possible to mitigate the effects of the first and second errors, respectively.  Unfortunately, both measures are also responsible for the
    exponentially decreasing success probability with increasing dimension $n$.
    
    To understand the effects of these error terms, we expand the inner product $v' \cdot \tfrac{y}{2 n q N} $ as follows
    \begin{eqnarray*}
      v' \cdot \frac{w}{2 n q N} 
      & = & 
      \bigl(v + N \lambda + e_1(v')\bigr) \cdot \frac{w}{2 n q N} \\
      & = &
      (v + N \lambda) \cdot \frac{w}{2 n q N} + e_1(v') \cdot \frac{w}{2 n q N} \\
      & = &
      (v + N \lambda) \cdot \frac{\lambda^*}{N} + (v + N \lambda) \cdot e_2(w) + e_1(v') \cdot \frac{w}{2 n q N} \\
      & = &
      v \cdot \frac{\lambda^*}{N} + \lambda \cdot \lambda^* + (v + N \lambda) \cdot e_2(w) + e_1(v') \cdot \frac{w}{2 n q N}. \\
    \end{eqnarray*}
    Since $v \cdot \tfrac{\lambda^*}{N}$ is constant and $\lambda\cdot\lambda^*\in\Z$, we only have to consider the inner products 
    $e_1(v') \cdot \tfrac{w}{2 n q N}$ and $(v + N \lambda) \cdot e_2(w)$.
  	
  	Using the upper bound $\norm{e_1(v')}_\infty \le 1 - \frac{1}{L}$,   	
    the absolute value of the first error term is seen to be bounded from above by
    \[
    	\Bigl|e_1(v') \cdot \frac{w}{2 n q N}\Bigr| \le \frac{n}{2 n q N} \norm{e_1(v')}_\infty
    	\norm{w}_\infty \le \frac{1}{2 q N} (1 - \tfrac{1}{L}) 2 n q \kappa N < \kappa n.
    \]  	
  	To bound the norm of the second error term, we set 
  	\[
  		p_k = 2nq\lambda^*_k - \floor{2nq\lambda^*_k}
  	\]
  	for $k=1,\ldots,n$.  In words, the values $p_k$ correspond to
		the errors caused by rounding down the coefficients of $2nq\lambda^*$ to the nearest integer.  Set 
		\[
			A       = \{ k \, : \, w_k = \floor{2nq\lambda^*_k} \} \quad\mbox{and}\quad 
			\bar{A} = \{ \ell \, : \, w_\ell = \floor{2nq\lambda^*_\ell} + 1 \}.
		\]
		Observe that for $k\in A$ the $k$th coefficient of the error vector $e_2(w):=\tfrac{w}{2nqN} - \tfrac{\lambda^*}{N}$ is equal to 
		$\tfrac{-p_k}{2nqN}$ and for $\ell\in\bar{A}$ the $\ell$th coefficient is equal to $\tfrac{1-p_\ell}{2nqN}$.
		Set 
		\[
			L_A = \frac{1}{2n} \biggl( \sum_{k\in A} p_k + \sum_{\ell\in \bar{A}} (1 - p_\ell) \biggr),
		\]
		which is equal to $qN \| e_2(w) \|_1$.
	
		Since $v+ N \lambda \in [-1+\tfrac{1}{L},qN-\tfrac{1}{L}]^n\subset (-1,qN)^n$, we have
		\begin{eqnarray*}
			- \frac{1}{2n} \sum_{k\in A} p_k - \frac{1}{2nqN} \sum_{\ell\in\bar{A}} (1-p_\ell)
			& \le & 
			(v + N \lambda) \cdot e_2(w) \\
			\frac{1}{2n} \sum_{\ell\in \bar{A}} (1-p_\ell) + \frac{1}{2nqN} \sum_{k\in A} p_k 
			& \ge &
			(v + N \lambda) \cdot e_2(w) 
		\end{eqnarray*}
		Therefore, the sum $(v + N \lambda) \cdot e_2(w) + e_1(v') \cdot \frac{w}{2 n q N}$ of both error terms ranges over an interval of length at most
		\[
			L_A + \frac{1}{2qN} + 2\kappa n.
		\]
		Clearly, the identity
		\[
			L_{\bar{A}} = \frac{1}{2} - L_A
		\]
		holds for all $A\subseteq\{1,\ldots,n\}$.  This simple fact implies the crucial inequality 
		\[
			\min\{ L_A, L_{\bar{A}} \} \le \frac{1}{4}.
		\]
		The latter holds because otherwise we would have $L_A>\tfrac{1}{4}$ and $L_{\bar{A}}=\tfrac{1}{2} - L_A > \tfrac{1}{4}$, which would lead
		to the contradiction $\tfrac{1}{2} > \tfrac{1}{2}$.  
	
		In the remainder of the proof, without loss of generality $A$ always denotes a subset of $\{1,\ldots,n\}$ with
		$L_A\le\frac{1}{4}$.  
	
		Let $A$ be such subset and $w$ the corresponding approximation of $2nqN\lambda^*$.
		This means that the sum we want to estimate can be written as
    \[ 
    	\sum_{v'\in\calM} \exp(2 \pi i (\alpha + \beta_{v'})) =
    	\exp(2 \pi i \alpha) \sum_{v'\in\calM} \exp(2 \pi i \beta_{v'}) 
    \] 
    with $\alpha, \beta_{v'} \in \R$ and
    $-\tfrac{1}{2} L_{\mathrm{phase}} \le \beta_{v'} \le \tfrac{1}{2} L_{\mathrm{phase}}$, where $L_{\mathrm{phase}}=L_A+\tfrac{1}{2qN}+2\kappa n$.
    Hence, the real part of every term $\exp(2 \pi i \beta_{v'})$ is $\cos(2 \pi \beta_{v'}) \ge \cos(\pi L_{\mathrm{phase}})$ since $L_{\mathrm{phase}} <
    \tfrac{1}{2}$ due to $L_A\le\tfrac{1}{4}$ and the special choice of $\kappa$ in (IV).
    
    This implies that the absolute value of the sum is bounded from below by $M \cos(\pi L_{\mathrm{phase}})$ for this 
    particular $w$.  Finally, we obtain the desired claim
		\begin{eqnarray*}
			\Pr(\calR_{\lambda^*}) 
			& \ge &
			\frac{M}{W} \sum_{A \, : \, L_A\le \tfrac{1}{4}} \cos^2\Bigl( \pi \bigl( L_A + \frac{1}{2qN} + 2\kappa n \bigr) \Bigr) \\
			& \ge &
			\frac{2^{n-1} M}{W} \cos^2\Bigl( \pi \bigl( \frac{1}{4} + \frac{1}{2qN} + 2\kappa n \bigr) \Bigr)
		\end{eqnarray*}
		by noting that there are at least $2^{n-1}$ subsets $A$ with $L_A\le\frac{1}{4}$.
		\end{proof}

		\subsection{Sampling in dimension one}
		\begin{remark}[One-dimensional infrastructures]
			In the special case of one-dimensional infrastructures, it is better to work
                        with the sets
			\[
				\calR_{\lambda^\ast} = \{ w \mid w= [2q\lambda^\ast] \}
			\]
			for $\lambda^\ast\in\Lambda^\ast$.  This is because we may then choose a slightly larger $\kappa$.  The upper bound can be increased to
			\[
				\kappa < \frac{1}{8} - \frac{1}{8 q N},
			\]
			which leads to the higher lower bound on the success probability
			\[
				\Pr\bigl( \calR_{\lambda^\ast} \bigr) \ge \frac{M}{W} \cos^2\Bigl( \pi \bigl( \frac{1}{4} + \frac{1}{4qN} + 2\kappa n \bigr) \Bigr).
			\]
			This bound is established by using the same arguments as in the proof of the above proposition and by observing that the upper bound on $|e_2(w)|$ is reduced by a factor 
			of $2$.  The latter statement is due to the fact that 
			for all $\lambda^\ast\in\Lambda^\ast$, we have the better approximation
			\[
				\Bigl| \frac{w}{2q} - \lambda^\ast \Bigr| \le \frac{1}{4q},
			\]
			where $w\in\calR_{\lambda^\ast}$.
		\end{remark}
	
 	\section{Lattice theoretic tools -- Part 1}\label{sec:part1}

  \subsection{Lattices of dimension greater than one}  
  
  We now show how to obtain a generating set of a full-rank lattice $L$ in $\R^n$ by first sampling $n$ lattice vectors that are contained in the window $[0,b)^n$ and then
  $n+1$ lattice vectors that are contained in the larger window $[0,b_0)^n$.  If we chose $b$ to be a sufficiently larger than the covering radius of $L$, then 
  the first $n$ lattice vectors generate a full-rank sublattice $L_0$ of $L$ with probability greater or equal to $\tfrac{1}{4}$ (Subsection~\ref{subsec:fullRank}).
  Once we have such sublattice $L_0$, the next $n+1$ lattice vectors that we sample from the larger window $[0,b_0)$ generate together with the first $n$ vectors 
  the entire lattice $L$ with probability greater or equal to $\hat{\zeta}-\frac{1}{4} \ge 0.184$, where $\hat{\zeta}$ is a certain constant (Subsection~\ref{subsec:entireLattice}).
  
  Our current proof requires that we use two windows.  We think that it is possible to prove a similar result, while relying only on one window.  
  
  Note that these results will be used with $L = \Lambda^*$ throughout the rest of the paper.
  
  \subsubsection{Probability of generating a full-rank sublattice $L_0$ of $L$}\label{subsec:fullRank}
  
  Let $L$ be a lattice in $\R^n$ of full rank. For $\lambda \in L$, let $V_L(\lambda)$ be its (open)
  Voronoi cell. We know that $V_L(\lambda)$ is contained in an open sphere of radius $\nu(L)$
  centered around $\lambda$, where $\nu(L)$ is the covering radius of $L$, and that the volume of
  $V_L(\lambda)$ is $\det (L)$. Moreover, if $\lambda \neq \lambda'$, $V_L(\lambda) \cap
  V_L(\lambda') = \emptyset$, and $\bigcup_{\lambda \in L} \overline{V_L(\lambda)} = \R^n$.
  
  \begin{lemma}
    If $b > 2 \nu(L)$. Then 
    \[ 
    	\frac{(b - 2\nu(L))^n}{\det(L)} \le \abs{L \cap [0, b)^n} \le \frac{(b + 2\nu(L))^n}{\det(L)}. 
    \]
  \end{lemma}
  
  \begin{proof}
    If $\lambda \in L$ satisfies $V_L(\lambda) \cap [\nu, b - \nu)^n \neq \emptyset$,
    then we must have $\lambda \in [0, b)^n$. Therefore, $ (b - 2
    \nu)^n / \det(L) \le \abs{L \cap [0, b)^n}$.
    
    If $\lambda \in L \cap [0,b)^n$, then we must have $V_L(\lambda) \subseteq [-\nu, b + \nu)^n$.  
    Therefore, $\abs{L \cap [0, b)^n} \le (b + 2\nu)^n$.
  \end{proof}
  
  \begin{lemma}
    Let $b > 0$ and $H$ be a $k$-dimensional hyperplane, $1 \le k < n$. Then \[ \abs{L \cap H \cap
    [0, b)^n} \le \frac{n^{k/2} (b + 2 \nu(L))^k (2 \nu(L))^{n-k}}{\det(L)}. \]
  \end{lemma}
  
  \begin{proof}
    Let $\lambda \in L \cap H \cap [0, b)^n$. Then $V_L(\lambda) \subseteq X := [-\nu, b +
    \nu)^n \cap (H + B_\nu(0))$, where $B_\nu(0)$ is a sphere of radius~$\nu$ centered around
    0. Therefore, $\abs{L \cap H \cap [0, b)^n} \le \volume(X) / \det(L)$, and we have to
    estimate $\volume(X)$.
    
    Clearly, if $\volume_k(Y)$ denotes the $k$-dimensional volume of $Y := H \cap [-\nu, b +
    \nu)^n$, we have that $\volume(X) \le \volume_k(Y) \cdot (2 \nu)^{n - k}$. (In fact, we can
    replace $(2 \nu)^{n - k}$ by the volume of an $(n - k)$-dimensional sphere of radius~$\nu$.)
    
    Let $b_1, \dots, b_k$ be an orthonormal basis of $H$. Set $T := \{ (x_1, \dots, x_k) \in \R^k
    \mid \sum_{i=1}^k x_i b_i \in [-\nu, b + \nu)^n \}$; then $\volume(T) = \volume_k(Y)$. A point
    $y \in Y$ corresponds to $(\< y, b_1 \>, \dots, \< y, b_k \>) \in T$. Write $b_i = (b_{i1},
    \dots, b_{in})$ and $y = (y_1, \dots, y_n) \in [-\nu, b + \nu)^n$, set $A_{ij} := b + \nu$ if
    $b_{ij} \ge 0$ and $A_{ij} := \nu$ if $b_{ij} < 0$. Then 
    \[ 
    	\sum_{j=1}^n |b_{ij}| (A_{ij} - (b + 2\nu)) \le \< y, b_i \> = 
    	\sum_{j=1}^n y_j b_{ij} \le \sum_{j=1}^n |b_{ij}| A_{ij}, 
    \] 
    implying that $\< y, b_i \>$ ranges over an interval of length $\norm{b_i}_1 (b + 2 \nu) \le \sqrt{n} (b + 2
    \nu)$. Therefore, 
    \[ 
    	\volume(T) \le n^{k/2} (b + 2 \nu)^k. 
    \]
  \end{proof}
  
  \begin{corollary}\label{cor:probGenFullRankSubLattice}
    Assume that $b \ge \max\{ 8 n - 2, n^{(n-1)/2} 2^{n+1} - 2 \} \cdot \nu(L)$. Let
    \begin{align*}
      X :={} & (L \cap [0, b)^n)^n \\
      \text{and} \quad Y :={} & \{ (y_1, \dots, y_n) \in X \mid \mspan_\R(y_1, \dots, y_n) =
      \R^n \}.
    \end{align*}
    Then \[ \abs{Y} > 0.289 \abs{X} > \frac{1}{4} \abs{X}. \]
  \end{corollary}
  
  Note that $\max\{ 8 n - 2, n^{(n-1)/2} 2^{n+1} - 2 \} = n^{(n-1)/2} 2^{n+1} - 2$ unless $n \le 2$,
  in which case the maximum is $8 n - 2$.
  
  The proof of this corollary is similar to the proof of the first part of Satz~2.4.23 in
  \cite{arthurDiss}. Note that the proof in \cite{arthurDiss} is not correct: the quantity
  $\frac{\abs{M_i \cap \calB}}{\abs{M_{i-1} \cap \calB}}$ in the proof can be $> \frac{1}{2}$; for
  example, consider $r = 3$, $M = \Z^3$, $n > 0$ arbitrary (in \cite{arthurDiss}, $n \nu(M)$ is what
  we denote by $b$, i.e., $\calB = [0, n \nu(M))^n$), $x_1 = (1, n \nu(M) - 1, -1)$, $x_2 = (0, 1, n
  \nu(M) - 1)$, $x_3 = (0, 0, 1)$; then $M_1 \cap \calB$ contains three elements, while $M_2 \cap
  \calB$ contains five elements. The problem is that $\det(M_i)$ cannot be bounded in terms of
  $\nu(M)$ and $\det(M_{i-1})$, as it was claimed in that proof. We proceed differently by
  considering the quantity $\frac{\abs{M_i \cap \calB}}{\abs{M \cap \calB}}$ directly, and our bound
  on the minimal size of $\calB$ is in fact better than the bound given in \cite{arthurDiss}.
  
  Also, note that for specific small~$n$, one can obtain better bounds of $\abs{Y}$ in term of
  $\abs{X}$. As the proof will show, a lower bound on $\abs{Y}$ is given by $\abs{X} \cdot
  \prod_{i=1}^{n-1} (1 - 2^{-i})$. The following table gives explicit values for this factor for
  small values of $n$, rounded down to a precision of $10^{-3}$:
  \begin{center}
    \begin{tabular}{c|ccccc}
      $n$ & $2$ & $3$ & $4$ & $5$ & $6$ \\\hline
      $\prod_{i=1}^{n-1} (1 - 2^{-i})$ & 0.500 & 0.375 & 0.328 & 0.307 & 0.298 \\
    \end{tabular}
  \end{center}
  
  \begin{proof}
    Assume that $y_1, \dots, y_k \in X$ are linearly independent. We have to compute the probability
    that $y_{k+1} \in X$ is not contained in the hyperplane generated by $y_1, \dots, y_k$, which is
    of dimension~$k$. Write $b = j \nu(L)$ with $j\ge n^{(n-1)/2} 2^{n+1} - 2$.  By the above
    lemmata, the probability that $y_{k+1}$ is in a $k$-dimensional hyperplane is bounded from above
    by
    \begin{eqnarray*}
      P_k & := & \frac{n^{k/2} (b + 2 \nu)^k (2 \nu)^{n-k}}{\det(L)} \cdot \frac{\det(L)}{(b - 2 \nu)^n} \\
          & =  & \frac{n^{k/2} (b + 2 \nu)^k (2 \nu)^{n-k}}{(b - 2 \nu)^n} = n^{k/2} \frac{(j + 2)^k 2^{n - k}}{(j - 2)^n}.
    \end{eqnarray*}
    We now prove that $P_k \le 2^{-k}$ holds, which is equivalent to
    \begin{align*}
    	n^{k/2} (j + 2)^k 2^n  \le (j - 2)^n.
    \end{align*}
    Clearly, the left-hand side is maximal for $k=n-1$, giving the strictest condition
    \begin{align*}
    	n^{(n-1)/2} 2^n \le (j + 2) \biggl( \frac{j - 2}{j + 2} \biggr)^n\,.
    \end{align*}
    The right-hand side is bounded from below by $(j+2)/2$ provided that $j\ge 8n-2$ (this follows from Bernoulli's inequality).
  	Hence, the above condition is satisfied for $j\ge n^{(n-1)/2} 2^{n+1} - 2$.

    To conclude the proof, note that the probability we look for is therefore bounded from below
    by 
    \[ 
    	\prod_{i=1}^{n-1} (1 - 2^{-i}) \ge \prod_{i=1}^\infty (1 - 2^{-i}) > 0.289 > \frac{1}{4}, 
    \] 
    where the last two inequalities follows by Euler's Pentagon Number Theorem.
  \end{proof}
 
	\subsubsection{Probability of generating finite abelian groups}

	\begin{proposition}\label{prop:genFiniteAbelianGroup}
  	Let $G$ be a finite abelian group known to be generated by~$n$ elements. Then the probability that
  	$n + 1$ elements drawn uniformly at random from $G$ generate $G$ is at least 
  	\begin{align*}
  		\hat{\zeta} := \prod_{i=2}^\infty \zeta(i)^{-1} \ge 0.434\,,
  	\end{align*}
  	where $\zeta$ denotes the Riemann zeta function.
	\end{proposition}
        
        Note that for small~$n$, better lower bounds on the probability can be obtained. If $G$ can
        be created by $n$~elements, then a better lower bound is $\prod_{i=2}^{n+1} \zeta(i)^{-1}$;
        this is always larger than $\hat{\zeta}$. The following table gives explicit values for this
        product for small values of $n$, rounded down to a precision of $10^{-3}$:
        \begin{center}
          \begin{tabular}{c|ccccc}
            $n$ & $2$ & $3$ & $4$ & $5$ & $6$ \\\hline
            $\prod_{i=2}^{n+1} \zeta(i)^{-1}$ & 0.505 & 0.467 & 0.450 & 0.442 & 0.439 \\
          \end{tabular}
        \end{center}
        
	\begin{proof}
  	Let $p_1, \dots, p_k$ be the prime divisors of $|G|$, and let $G_i$ be the $p_i$-Sylow subgroup of
  	$G$. Then $G = G_1 \oplus \dots \oplus G_k$. Let $(g_1, \dots, g_{n+1}) \in G^{n+1}$ be $n+1$ elements
  	of $G$; then we can write $g_i = (g_{i1}, \dots, g_{ik}) \in G_1 \times \dots \times G_k$. Now \[
  	G = \langle g_1, \dots, g_{n+1} \rangle \Longleftrightarrow \forall j : G_j = \langle g_{1j},
  	\dots, g_{n+1,j} \rangle. \] Hence, it suffices to bound the probability for abelian $p$-groups.
  
  	In the proof of the theorem in \cite{pomerance-generate}, it is shown that the probability that $n +
        1$ elements in an abelian $p$-group of $p$-rank~$r$ generate the group is \[ \prod_{i=1}^r
        (1 - p^{-((n + 1 - r) + i)}) \ge \prod_{i=2}^{n+1} (1 - p^{-i}). \] We know that $r \le n$,
        since $G$ is generated by $n$~elements.
  
  	Therefore, the probability that $n$ elements of an arbitrary finite abelian group~$G$ which can be
  	generated by $n$ elements generate the group is at least 
  	\[ 
  		\prod_p \prod_{i=2}^{n+1} (1 - p^{-i}) =
  		\prod_{i=2}^{n+1} \prod_p (1 - p^{-i}) = 
  		\biggl( \prod_{i=2}^{n+1} \zeta(i) \biggr)^{-1}
  	\]
        using the Euler product representation of the Riemann zeta function.  Now \[
        \prod_{i=2}^{n+1} \zeta(i) \le \prod_{i=2}^{\infty} \zeta(i) = \hat{\zeta}^{-1}. \] The
        product $\prod_{i=2}^\infty \zeta(i)$ is well-known in group theory \cite{Sloane}.
	\end{proof}

        Note that it is essential for our proof to work that we use $n + 1$ elements instead of $n$,
        since if we choose just~$n$ elements randomly, the final product would include
        $\zeta(1)^{-1} = 0$ and the probability would drop down to zero. However, a different
        approach can result in a non-zero probability for $n$ elements, but this probability will
        not be constant anymore, but depend on $n$ or $\abs{G}$. For example, if $p_1, \dots, p_k$
        are distinct primes and $G = \prod_{i=1}^k \F_{p_i}^n \cong (\Z/p_1 \cdots p_k\Z)^n$, then
        $G$ can be generated by $n$ elements, but the probability that $n$~random elements from $G$
        generates~$G$ is exactly $\prod_{i=1}^k \prod_{j=1}^n (1 - p_i^j)$, which goes to zero if $k
        \to \infty$ for exactly the above reasons. Hence, any non-trivial bound of the probability
        must take $n$ or $p_1, \dots, p_k$ into account.

\subsubsection{Probability of generating the entire lattice $L$}\label{subsec:entireLattice}

\begin{lemma}[Sampling almost uniformly at random from $L/L_0$]\label{lem:samplingQuotient}
	Let $L_0$ be an arbitrary full-rank sublattice of $L$.  Assume that $b_0>2\nu(L_0)$ and 
	we can sample uniformly at random from
	\begin{align*}
	L \cap [0,b_0)^n\,.
	\end{align*}
	Denote the sample by $\lambda$.  Then, $\lambda + L_0$ is distributed almost uniformly at random over the quotient group $L/L_0$.
	More precisely, the total variation distance between the uniform distribution is at most
	\begin{align*}
		1 - \frac{(b_0-2\nu(L_0))^n}{(b_0+2\nu(L))^n}\,.
	\end{align*}
\end{lemma}

\begin{proof}
  Let again $V_{L_0}(\lambda_0)$ denote the open Voronoi cell of the lattice $L_0$
  centered around $\lambda_0$. First note that $V_{L_0}(\lambda_0) = \lambda_0 +
  V_{L_0}(0)$ and $\overline{V_{L_0}(\lambda_0)} = \lambda_0 +
  \overline{V_{L_0}(0)}$. Now, as $\bigcup_{\lambda_0 \in L_0} (\lambda_0 +
  \overline{V_{L_0}(0)}) = \R^n$ and two translates of $V_{L_0}(0)$ by different
  elements of $L_0$ do not intersect, there exists a set $V$ with $V_{L_0}(0) \subseteq
  V \subseteq \overline{V_{L_0}(0)}$ satisfying \[ \bigcup_{\lambda_0 \in L_0}
  (\lambda_0 + V) = \R^n \quad \text{and} \quad \forall \lambda_0 \in L_0 \setminus \{ 0 \} :
  (\lambda_0 + V) \cap V = \emptyset. \] Note that $\volume(V) = \volume(V_{L_0}(0)) = \det
  (L_0)$.

  Every translate of $V$ contains the same number of elements from $L$, and $\abs{V \cap
    L}$ equals
  \begin{align*}
    m = \det(L_0)/\det(L);
  \end{align*}
  this can be shown using asymptotic arguments similarly to the proof that any elementary
  parallelepiped of $L_0$ contains exactly $m$ elements of $L$ (see
  e.g.\ \cite{barvinok-notes}).
  
	For all $\lambda \in L \cap V$, the vectors $\lambda-\lambda_0$ form a transversal for $L/L_0$.
	
	As $V \subseteq \overline{B_{\nu(L_0)}(0)}$, there are at least
	\begin{align*}
		\ell_V = \frac{(b_0 - 2\nu(L_0))^n}{\det(L_0)}
	\end{align*}
	translates of $V$ that are contained inside the window $[0,b_0]^n$.
	
	There are at most
	\begin{align*}
		u_P = \frac{(b_0 + 2\nu(L))^n}{\det(L)} 
	\end{align*}
	points of $L$ inside $[0,b_0]^n$.  
	
	Let $d_{\max}= \lfloor u_p - m \ell_V \rfloor$ be the maximal possible deviation in the number of points of $L$ inside $[0,b_0]^n$ from the lower bound $m \ell_V$.  Let
	$d\in\{0,\ldots,d_{\max}\}$ be the actual deviation.
	
	Ideally, we would have the uniform distribution $p_j = 1/m$ on $L/L_0$.  But we only have the almost uniform distribution which necessarily has the form
	\begin{align*}
		\tilde{p}_j = \frac{\ell_V + d_j}{m \ell_V + d}
	\end{align*}
	for $j=1,\ldots,m$, where $d_1,\ldots,d_m$ are integers with $0\le d_j \le d$ and $\sum_{j=1}^m d_j = d$.  The total variation distance can be bounded as follows
	\begin{eqnarray*}
		\frac{1}{2} \sum_{j=1}^m | p_j - \tilde{p}_j | 
		& = & 
		\frac{1}{2} \sum_{j=1}^m \left| \frac{1}{m} - \frac{\ell_V + d_j}{m \ell_V + d} \right| \\
		& = &
		\frac{1}{2m} \sum_{j=1}^m \left| \frac{d - m d_j }{m \ell_V + d} \right| \\
		& \le & 
		\frac{1}{2m} \sum_{j=1}^m \frac{d + m d_j }{m \ell_V + d} \\
		& = & 
		\frac{d}{m \ell_V + d} \\
		& \le & 
		\frac{d_{\max}}{m \ell_V + d_{\max}} \le \frac{u_p - m \ell_V}{m \ell_V + u_p - m \ell_V} = 1 - \frac{m\ell_V}{u_P}\,.
	\end{eqnarray*}
	We have 
	\begin{align*}
		1 - \frac{m \ell_V}{u_P} = 1 - \frac{(b_0-2\nu(L_0))^n}{(b_0+2\nu(L))^n} \,.
	\end{align*}
        
        Note that so far, we have considered $[0, b_0]^n$ instead of $[0, b_0)^n$. As $L$ is
          discrete, there exists some $2 \nu(L_0) < b_0' < b_0$ with $[0, b_0']^n \cap L
          = [0, b_0)^n$. Applying the result above to $[0, b_0']^n$ and then using that $x \mapsto 1
            - \frac{(x - 2 \nu(L_0))^n}{(x + 2 \nu(L))^n}$ is increasing yields the
            stated claim for $[0, b_0)^n$.
  
\end{proof}

\begin{proposition}\label{prop:generateWholeLattice}
	Assume that $b\ge \max\{ 8 n - 2, n^{(n-1)/2} 2^{n+1} - 2 \} \cdot \nu(L)$ and $b_0\ge 8n^2(n+1)b$.  Let $Y$ be as in 
	Corollary~\ref{cor:probGenFullRankSubLattice} and $(y_1,\ldots,y_n)\in Y$.  Let 
	\begin{eqnarray*}
	  X_0 & := & \bigl(L \cap [0,b_0)^n \bigr)^{n+1} \\
		Z   & =  & \{(z_1,\ldots,z_{n+1})\in X_0^{n+1} \mid \mspan_\Z \{y_1,\ldots,y_n,z_1,\ldots,z_{n+1} \} = L \}.
	\end{eqnarray*}
	Then
	\begin{align*}
		\abs{Z} \ge \left(\hat{\zeta}-\frac{1}{4} \right)	\abs{X_0} \ge 0.184 \abs{X_0}.
	\end{align*}
\end{proposition}  

\begin{proof}
  Let $L_0$ be the full-rank sublattice generated by $y_1,\ldots,y_n$.  We have the following simple bound on the covering radius
	\begin{align*}
		\nu(L_0) \le \frac{\sqrt{n}}{2} \lambda_n(L_0) \le \frac{\sqrt{n}}{2} \max_{i=1,\ldots,n} \| y_i \|_\infty \le
		\frac{\sqrt{n}}{2} \sqrt{n} b = \frac{nb}{2}
	\end{align*}
	since the $y_i$ are linearly independent and the longest vector in $[0,b)^n$ is shorter than $\sqrt{n}b$.
  
  Let $z_i$ be uniformly distributed in $L\cap [0,b_0)^n$.  Then, Lemma~\ref{lem:samplingQuotient} implies that $z_i + L_0$ (for $i=n+1,\ldots,2n+1$) 
  are distributed almost uniformly at random from $L/L_0$.  The total variation distance from the uniform distribution is bounded from above as follows
  
	\begin{eqnarray*}
  	1 - \frac{(b_0-2\nu(L_0))^n}{(b_0+2\nu(L))^n} 
  	& \le &
		1 - \frac{(b_0-2\nu(L_0))^n}{(b_0+2\nu(L_0))^n} \\
		& = & 
		1 - \left( 1 - \frac{4\nu(L_0)}{b_0+2\nu(L_0)} \right)^n \\
		& \le &
		1 - \left( 1 - n \, \frac{4\nu(L_0)}{b_0+2\nu(L_0)} \right) \\ 
		& \le & 
		\frac{4n \nu(L_0)}{b_0} \le \frac{2 n^2 b}{b_0} \le \frac{1}{4(n+1)}\,.
	\end{eqnarray*}

	Consider now the uniform probability distribution on the $(n+1)$-fold direct product of
        $L/L_0$ and the probability distribution that arises from sampling almost uniformly at
        random on each of the components as above.  Then the total variation between these two
        distributions is bound from above by $(n+1) \cdot \frac{1}{4(n+1)} = \tfrac{1}{4}$.  This is
        because total variation distance is additive under composition provided that the components
        are independent (see e.g.\ \cite[Subsection 1.3 ``Statistical distance'' in Chapter
        7]{micciancio-goldwasser} for more information total variation distance).

	Clearly, the abelian group $L/L_0$ can be generated with only $n$ generators.  Hence, Proposition~\ref{prop:genFiniteAbelianGroup} 
	implies that $n+1$ samples (provided that they are
	distributed uniformly at random over the group) form a generating set with probability greater or equal to $\hat{\zeta}$.  Due to the deviation from 
	the uniform distribution on the $(n+1)$-fold direct product
	of $L/L_0$ this probability may decrease. However it is at least $\hat{\zeta}-1/4$ since the total variation distance is at most $1/4$.  The claim follows 
	now by translating the lower bound on the probability to a lower bound on the fraction of elements with the desired property.
	\end{proof}

	\begin{remark}
		The purpose of this proposition is similar to that of Satz~2.4.23 in \cite{arthurDiss}.  We emphasize that
		our bound on the success probability is constant, whereas the bound presented in Satz~2.4.23 decreases
		exponentially fast with the dimension $n$.  The first part of proof of Satz~2.4.23 (concerning the
		generation of a full-rank sublattice) is unfortunately not correct, but can be corrected as we have
		shown in our proof of Corollary~\ref{cor:probGenFullRankSubLattice}.  The idea behind the second part is completely
		different from our proof and cannot be used to prove a constant success probability.  Perhaps it could be used 
		to prove that only $2 n$ random elements (as opposed to $2 n + 1$ elements) are needed to guarantee a non-zero 
		success probability.
	\end{remark}

	Note that in \cite{hallgrenUnitgroup}, neither a bound is given on how many lattice elements have to be
	sampled nor the probability is estimated with which the lattice is generated.

	\begin{lemma}\label{lem:chooseWindowSizes}
	Assume
	\begin{align*}
		b   & \ge \max\{ 8 n - 2, n^{(n-1)/2} 2^{n+1} - 2 \} \cdot \frac{n}{2\lambda_1(\Lambda)} \;\;\; \text{and} \\
		b_0 & \ge 8 n^2 (n + 1) b. 
	\end{align*}
	Define
  \begin{align*}
 		X   &:= (\Lambda^\ast \cap [0, b)^n)^n \\	
  	Y   &:= \{ (\lambda^*_1, \dots, \lambda^*_n) \in X \mid \mspan_\R(\lambda^*_1, \dots, \lambda^*_n) = \R^n \}.
  \end{align*}
  For each $(\lambda^\ast_1,\ldots,\lambda^\ast_n)\in Y$, define
  \begin{align*}
  	X_0 &:= (\Lambda^\ast \cap [0, b_0)^n)^{n+1} \\	
  	Z   &:= \{ (\lambda^*_{n+1}, \dots, \lambda^*_{2n+1}) \in X_0 \mid \mspan_\Z(\lambda^*_1, \dots, \lambda^*_n,\lambda_{n+1}^*, \dots, \lambda^*_{2n+1}) = L \}.
  \end{align*}
  Then
  \[
  	\abs{Y} \ge 0.289 \, \abs{X} > \frac{1}{4} \abs{X} \;\;\;\; \text{and} \;\;\;\;
  	\abs{Z} \ge \bigl( \hat{\zeta} - \tfrac{1}{4} \bigr) \abs{X_0} \ge 0.184 \, \abs{X_0}
  \]
  \end{lemma}
        
  \begin{proof}
  The first lower bound follows from Corollary~\ref{cor:probGenFullRankSubLattice} and the inequality $\nu(\Lambda^\ast) \le \frac{n}{2\lambda_1(\Lambda)}$ and the second from
  Proposition~\ref{prop:generateWholeLattice}.
  \end{proof}

  By combining the more precise bounds listed below Corollary~\ref{cor:probGenFullRankSubLattice}
  and Proposition~\ref{prop:generateWholeLattice}, respectively, one obtains the following more
  precise bounds which depend on $n$:
  \begin{equation}
    \tag{$\ast$}\label{eq:betterboundsongenprop}
    \abs{Y} \ge \abs{X} \cdot \prod_{i=1}^{n-1} (1 - 2^{-i}) \quad \text{and} \quad \abs{Z} \ge
    \biggl( \prod_{i=2}^{n+1} \zeta(i)^{-1} - \frac{1}{4} \biggr) \cdot \abs{X_0}.
  \end{equation}

\subsection{Lattices of dimension one}
	
	We now discuss the special case~$n = 1$.  For this case, $2 n$ instead
  of $2 n + 1$ vectors from one window suffice to generate the lattice with a significantly higher probability.
  
  \begin{lemma}
    \label{lemma:latticegeneration1D}
    Let $L=\Z v$ be a one-dimensional lattice, where $v\in\R_{>0}$.  Assume that $b \ge 3v +
    1$. Then, two samples chosen uniformly at random in $L\cap[0,b)$ generate $L$ with probability
    greater than $\frac{3^3}{\pi^2 2^3} > \tfrac{1}{3}$. Note that $\det(L) = v = \lambda_1(L)$,
    $\nu(L) = \tfrac{1}{2} \det(L)$ and that $L^* = \frac{1}{v} \Z$.
  \end{lemma}
  
  \begin{proof}
    Clearly, the number of lattice elements in
    $[0, b - 1]$ is $1 + \lfloor \frac{b - 1}{v} \rfloor$, where 1 accounts for the zero
    vector.  Hence, the probability that a random element of $L \cap [0, b - 1]$ is non-zero is
    \[ 
    	\frac{\lfloor \frac{b - 1}{v} \rfloor}{1 + \lfloor \frac{b - 1}{v} \rfloor} = 
    	1 - \frac{1}{1 + \lfloor \frac{b - 1}{v} \rfloor},  
    \] 
    which greater or equal to $\tfrac{3}{4}$ for $b \ge 3 v + 1$.  Further, note that this condition
    ensures that there are at least $3$ non-zero elements.  Assume that we obtained two non-zero
    elements; these have the form $k v$ and $\ell v$, where $k,\ell$ are chosen uniformly at random
    in $\{1,\ldots,m\}$ with $m\ge 3$.  It is well-known that $\gcd(k,\ell)=1$ with probability
    greater than $\frac{6}{\pi^2}$.  This proves the bound $\tfrac{6}{\pi^2} (\tfrac{3}{4})^2 >
    \tfrac{1}{3}$.
  \end{proof}

	\section{Obtaining an approximate generating set of the dual lattice $\Lambda^*$}\label{sec:approximateDualLattice}
  
  \subsection{Lattices of dimension greater than one}
  
  The current result in Proposition~\ref{prop:generateWholeLattice} forces us to sample lattice vectors from windows of two different sizes.  Recall that the parameter $N$
  directly determines the size of the portion of the dual lattice $\Lambda^*$ from which we can sample. We refer to this parameter as $N$ in Subsection~\ref{subsec:fullRankSublatticeDual} 
  and as $N_0$ in Subsection~\ref{subsec:entireDual}.  The other parameters $q$ and $\kappa$ can be chosen to be the same. 
  
  \subsubsection{Generating a full-rank sublattice of the dual lattice}\label{subsec:fullRankSublatticeDual}  
	  
	\begin{lemma}\label{prop:samplingFromGoodSet}
  	Choose $q$, $N$, and $\kappa$ according to (III)--(V) and 
  	\begin{align}\tag{VI}\label{eq:VI}
  		N \ge{} & \frac{1}{\kappa} \left( \max\{ 8 n - 2, n^{(n-1)/2} \cdot 2^{n+1} - 2 \}
                \cdot \frac{n}{2\lambda_1(\Lambda)} + \frac{1}{2nq} \right), \\ \tag{VII}\label{eq:VII} N >{} &
                \frac{1}{\kappa} \left( \frac{1}{2 q} + \frac{n^2}{\lambda_1(\Lambda)} \right).
  	\end{align}
    Run the quantum algorithm $n$ times and denote the samples by $w_1,\ldots,w_n$.  Then, the probability that 
    there exists $\lambda^\ast_1,\ldots,\lambda^\ast_n\in\Lambda^\ast \cap [0,\kappa N - \tfrac{1}{2nq})^n$ with
    \begin{enumi}
    	\item the lattice vectors $\lambda^\ast_1,\ldots,\lambda^\ast_n$ span a full-rank sublattice of $\Lambda^\ast$ and
    	\item the samples $w_i$ approximate these lattice vectors $\lambda_i^\ast$ so that
    	\[
    		\norm{\frac{w_i}{2nq} - \lambda_i^\ast}_2 \le \frac{1}{2\sqrt{n} q} \;\; \text{for $i=1,\ldots,n$}
    	\]
    \end{enumi}
    is greater or equal to
	  \begin{align*}
  		& \;\;\;\;\;
  		\frac{1}{4} \left(\frac{2^{n-1} M_\ell L_\ell \, c}{W}\right)^n \\
  		& \ge 
  		\frac{1}{4} \left( \frac{c}{2} \right)^n \left( \frac{\kappa}{n} \right)^{n^2} 
  			\cdot \left[ 1 - \left( \frac{1}{2 q} + \frac{n^2}{\lambda_1(\Lambda)} \right) \frac{1}{\kappa N} \right]^n
  			\cdot \left[ 1 - \frac{3n}{qN} - \frac{2n\nu(\Lambda)}{q} \right]^n \\
  		& \approx \frac{1}{4} \left( \frac{c}{2} \right)^n \left( \frac{\kappa}{n} \right)^{n^2}.
  	\end{align*}
  	Here $c:=\cos^2\bigl( \pi (\tfrac{1}{4} + \tfrac{1}{2qN} + 2 \kappa n) \bigr)>0$ and 
  	$L_\ell$ is a lower bound on the cardinality of $\Lambda^\ast \cap [0,\kappa N - \tfrac{1}{2nq})^n$.  
  	The approximation $\approx$ indicates that $L_\ell$ and $M_\ell$ are close to $1$ 
  	provided that $b$, $N$ and $q$ are sufficiently large.

  \end{lemma}
  
  Here, the factor $\frac{1}{4}$ can be replaced with 0.289 or $\prod_{i=1}^{n-1} (1 - 2^{-i})$
  (compare Equation~\eqref{eq:betterboundsongenprop} on page~\pageref{eq:betterboundsongenprop}).
  
  \begin{proof}
    Observe that $\calR_{\lambda^\ast}\subset [0,2nq\kappa N]^n$ for all $\lambda^\ast\in\Lambda^\ast \cap [0,\kappa N - \tfrac{1}{2nq})^n$.  Set $b:=\kappa N - \tfrac{1}{2nq}$. 
    For all $\lambda^\ast\in\Lambda^\ast \cap [0,b)$, Proposition~\ref{prop:lbprobone} yields the lower bound
  	\[
  		\Pr(w_i \in \calR_{\lambda^\ast}) \ge 
  		\frac{2^{n-1} M_\ell c}{W}.
  	\]
  	Clearly, if $w_i\in\calR_{\lambda^\ast}$ then
  	\[
  		\left\| \frac{w_i}{2nq} - \lambda^\ast \right\|_2 \le \frac{1}{2\sqrt{n} q}.
  	\]
  	We obtain the lower bound
  	\[
                \sum_{(\lambda^\ast_1, \dots, \lambda^\ast_n)\in(\Lambda^\ast_b)^n} \Pr\bigl( w_1 \in \calR_{\lambda^\ast_1}, \ldots, w_n \in \calR_{\lambda^\ast_n} \bigr) \ge
  		\biggl(\frac{2^{n-1} M_\ell L_\ell c}{W}\biggr)^n
  	\]
  	where 
  	\[
  		L_\ell = (\kappa N)^n \det(\Lambda) \left[ 1 - \left( \frac{1}{2 q} + \frac{n^2}{\lambda_1(\Lambda)} \right) \frac{1}{\kappa N} \right]
  	\]
  	is a lower bound on on the cardinality of $\Lambda^\ast\cap [0,b)^n$.  We derive this particular lower bound by applying the argument based on Voronoi cells and
  	\begin{align*}
  		\frac{\bigl(\kappa N - \frac{1}{2nq} - 2\nu(\Lambda^\ast)\bigr)^n}{\det(\Lambda^\ast)} 
  		& = 
  		(\kappa N)^n \det(\Lambda) \left[ 1 - \left( \frac{1}{2 n q} + 2\nu(\Lambda^*) \right) \frac{1}{\kappa N} \right]^n \\
  		& \ge 
  		(\kappa N)^n \det(\Lambda) \left[ 1 - \left( \frac{1}{2 q} + 2 n \nu(\Lambda^*) \right) \frac{1}{\kappa N} \right] \\
  		& \ge 
  		(\kappa N)^n \det(\Lambda) \left[ 1 - \left( \frac{1}{2 q} + \frac{n^2}{\lambda_1(\Lambda)} \right) \frac{1}{\kappa N} \right]. 
  	\end{align*}
  	We used the Bernoulli inequality and the inequality
        $\lambda_1(\Lambda)\nu(\Lambda^*)\le\tfrac{1}{2} n$. Observe that (VII) implies that
        $L_\ell$ is nontrivial.
  	
  	Finally, (VI) implies that $b$ is greater than the lower bound in Lemma~\ref{lem:chooseWindowSizes}.  This shows that at least a fourth of the tuples $(\lambda^\ast_1,\ldots,\lambda^\ast_n)$ with 
  	$\lambda_i\in\Lambda^\ast\cap [0,b)^n$ for $i=1,\ldots,n$ are such that the lattice vectors generate a full-rank sublattice.
  \end{proof}
    
  \subsubsection{Generating the entire dual lattice}\label{subsec:entireDual}
  
  Now we combine Proposition~\ref{prop:samplingFromGoodSet} and
  Proposition~\ref{prop:generateWholeLattice}.  We use the same parameters $q$ and $\kappa$ as in
  the previous section.  We only have to use a larger value for $N$, which guarantees
  that we sample from a larger portion of the dual lattice $\Lambda^*$ to satisfy the premises of
  Proposition~\ref{prop:generateWholeLattice}.  We denote this larger value by $N_0$.  Note that
  with this choice the conditions (III) and (IV) on $q$, $N_0$, and $\kappa$ are
  automatically satisfied.  This is because it becomes easier to satisfy these conditions when $N$
  is made larger. 
    
  \begin{lemma}
    \label{lemma:genentireduallattice}
    Let $q$, $N$, and $\kappa$ be as in Lemma~\ref{prop:samplingFromGoodSet}.  Choose $N_0$ according to
    \begin{equation}\tag{VIII}\label{eq:VIII}
      N_0 \ge  8 n^2(n+1) N.
    \end{equation}
    Use the parameters $q$, $N_0$, and $\kappa$ for the quantum algorithm.  Run it $n+1$ times and
    denote the samples by $w_{n+1},\ldots,w_{2n+1}$.  Assume that $\lambda_1^*, \dots, \lambda_n^*$
    from Lemma~\ref{prop:samplingFromGoodSet} generate a full-rank sublattice of $\Lambda^*$.  Then,
    the probability that there exist
    $\lambda^\ast_{n+1},\ldots,\lambda^\ast_{2n+1}\in\Lambda^\ast\cap [0,\kappa N_0 -
    \frac{1}{2nq})^n$ with
    \begin{enumi}
      \item the lattice vectors $\lambda^\ast_{n+1},\ldots,\lambda^\ast_{2n+1}$ together with the
      lattice vectors $\lambda^\ast_1,\ldots,\lambda^\ast_n$ generate the entire dual lattice
      $\Lambda^\ast$ and
      \item the samples $w_{n+i}$ approximate these lattice vectors $\lambda^\ast_{n+i}$ so that
      \[
        \norm{\frac{w_{n+i}}{2nq} - \lambda_{n+i}^\ast}_2 \le \frac{1}{2\sqrt{n} q} \;\; \text{for $i=1,\ldots,n+1$}
      \]
    \end{enumi}
    is greater or equal to
    \begin{align*}
      & \;\;\;\;
      \Bigl(\hat{\zeta} - \frac{1}{4} \Bigr) 
      \left(\frac{2^{n-1} M_\ell L_\ell \, c_0}{W}\right)^{n+1} \\
      & \ge 
      \Bigl(\hat{\zeta} - \frac{1}{4} \Bigr) 
      \left( \frac{c}{2} \right)^{n+1}  \left( \frac{\kappa}{n} \right)^{n(n+1)}
      \cdot \left[ 1 - \left( \frac{1}{2 q} + \frac{n^2}{\lambda_1(\Lambda)} \right) \frac{1}{\kappa N_0} \right]^{n+1} \cdot \\
      & \quad\quad \left[ 1 - \frac{3n}{q N_0} - \frac{2n\nu(\Lambda)}{q} \right]^{n+1} \\
      & \approx 
      \Bigl(\hat{\zeta} - \frac{1}{4} \Bigr) \left( \frac{c}{2} \right)^{n+1} \left( \frac{\kappa}{n} \right)^{n(n+1)}.
    \end{align*}
  \end{lemma}
  
  Here, the factor $\hat{\zeta} - \frac{1}{4}$ can be replaced with $\prod_{i=2}^{n+1} \zeta(i)^{-1}
  - \tfrac{1}{4}$ (compare Equation~\eqref{eq:betterboundsongenprop} on
  page~\pageref{eq:betterboundsongenprop}).
  
  The proof of this lemma is basically the same as that of Lemma~\ref{prop:samplingFromGoodSet}.
  Here $L_\ell$ is the lower bound on $\Lambda^\ast\cap [0,b_0)^n$ where $b_0:=\kappa N_0 -
  \tfrac{1}{2nq}$, $M_\ell$ the lower bound on $M$ in Proposition~\ref{prop:Mestimate} (iii), and
  $W=(2n q N_0)^n$, and $c_0=\cos^2\bigl(\pi (\tfrac{1}{4} + \tfrac{1}{2 q N_0} + 2 \kappa
  n)\bigr)$.  The cosine factor $c_0$ is bounded from below by $c = \cos^2\bigl(\pi (\tfrac{1}{4} +
  \tfrac{1}{2 q N} + 2 \kappa n)\bigr)$ since $N_0>N$.  The approximation $\approx$ indicates that
  $L_\ell$ and $M_\ell$ are close to $1$ provided that $q$ and $N_0$ are sufficiently large.
  
  There is one point that should be explained in more detail.  It remains to verify that $b_0 \ge
  8n(n^2 + 1) b$ so that we can apply Lemma~\ref{lem:chooseWindowSizes}.  The condition on the
  relation of the window sizes is equivalent to
  \[
    \kappa N_0 - \frac{1}{2nq} \ge 8n (n^2+1) \Bigl(\kappa N - \frac{1}{2nq}\Bigr).
  \]
  This inequality is clearly satisfied due to (VIII).
  
  \subsubsection{Bounding the probability}
  \label{sec:boundingprobability}
  
  We replace condition (VII) by the stricter condition
  \begin{align}\tag{VII$_1$}\label{eq:VII1}
    N \ge \frac{1}{\kappa}  \left( \frac{n}{q} + \frac{2 n^3}{\lambda_1(\Lambda)} \right).
  \end{align}
  This, together with (VIII), implies 
  \[
    \left[ 1 - \left( \frac{1}{2 q} + \frac{n^2}{\lambda_1(\Lambda)} \right) \frac{1}{\kappa N}
      \right]^n \cdot \left[ 1 - \left( \frac{1}{2 q} + \frac{n^2}{\lambda_1(\Lambda)} \right)
      \frac{1}{\kappa N_0} \right]^{n+1} \ge \frac{1}{2^2}.
  \]
  Moreover, we replace condition~(IV) by the stricter
  condition
  \begin{equation}\tag{IV$_1$}\label{eq:IV1}
    q \ge \frac{6 n^2}{N} + 4 n (n + 1) \nu(\Lambda).
  \end{equation}
  This implies together with (VIII)
  \[
    \left[ 1 - \frac{3n}{q N} - \frac{2n\nu(\Lambda)}{q} \right]^n \cdot \left[ 1 - \frac{3n}{q N_0}
      - \frac{2n\nu(\Lambda)}{q} \right]^{n+1} \ge \frac{1}{2^2}.
  \]
  From the previous two subsections, under the assumption that (I)--(VIII) hold, we get that the
  probability that $2 n + 1$ samples from the algorithm generate the whole lattice~$\Lambda^*$ is at
  least
  \begin{align*}
    & 
    \frac{1}{4} \left(\hat{\zeta} - \frac{1}{4} \right)
    \left( \frac{c}{2} \right)^{2n+1}  \left( \frac{\kappa}{n} \right)^{2n^2+n}
    \cdot \left[ 1 - \left( \frac{1}{2 q} + \frac{n^2}{\lambda_1(\Lambda)} \right) \frac{1}{\kappa
        N} \right]^n \\
    & \cdot{} \left[ 1 - \frac{3n}{q N} - \frac{2n\nu(\Lambda)}{q} \right]^n
    \cdot \left[ 1 - \left( \frac{1}{2 q} + \frac{n^2}{\lambda_1(\Lambda)} \right) \frac{1}{\kappa
        N_0} \right]^{n+1} \\
    & \cdot{} \left[ 1 - \frac{3n}{q N_0} - \frac{2n\nu(\Lambda)}{q} \right]^{n+1},
  \end{align*}
  where $c = \cos^2\bigl(\pi (\tfrac{1}{4} + \tfrac{1}{4 q N} + 2 \kappa n)\bigr)$. Using the stricter conditions 
  ~(VII$_1$) and (IV$_1$) from above, this can be bounded from below by 
  \[ 
  	\frac{1}{2^6} \left( \hat{\zeta} - \frac{1}{4} \right) \left( \frac{c}{2} \right)^{2n+1} \left( \frac{\kappa}{n} \right)^{2n^2+n} \ge
  	\frac{1}{2^9} \left( \frac{c}{2} \right)^{2n+1} \left( \frac{\kappa}{n} \right)^{2n^2+n}.
  \]
  Here, the factor $\frac{1}{4} (\hat{\zeta} - \frac{1}{4})$ can be increased to 0.053176 or $\bigl(
  \prod_{i=2}^{n+1} \zeta(i)^{-1} - \tfrac{1}{4} \bigr) \cdot \prod_{i=1}^{n-1} (1 - 2^{-i})$
  (compare Equation~\eqref{eq:betterboundsongenprop} on
  page~\pageref{eq:betterboundsongenprop}). The latter would improve the lower bound on the
  probability that $2n+1$~samples from the algorithm generate the whole lattice~$\Lambda^\ast$ to \[
  \frac{1}{2^4} \biggl( \prod_{i=2}^{n+1} \zeta(i)^{-1} - \tfrac{1}{4} \biggr) \biggl(
  \prod_{i=1}^{n-1} (1 - 2^{-i}) \biggr) \left( \frac{c}{2} \right)^{2n+1} \left( \frac{\kappa}{n}
  \right)^{2n^2+n}. \]
  
  \subsection{Dimension one}
  \label{sec:boundingprop1D}
  
  Finally, we want to investigate the case $n = 1$ more closely. In this case, we have only one
  window and we sample only two vectors. If $b \ge 3 \det(L) + 1$,
  Lemma~\ref{lemma:latticegeneration1D} yields that two randomly sampled vectors from $\Lambda^*
  \cap [0, b)$ generate $\Lambda^*$ is larger than $\frac{1}{3}$. We proceed similarly to the proof
  of Proposition~\ref{prop:samplingFromGoodSet}. For $b = \kappa N - \frac{1}{2 q}$ to hold in
  conjunction with $b \ge 3 \det(L) + 1 = \frac{3}{\det(\Lambda)} + 1$, we must satisfy the new
  condition
  \begin{equation}\tag{VI$_2$}\label{eq:VI2}
    N \ge \frac{1}{\kappa} \biggl( \frac{3}{\det(\Lambda)} + 1 + \frac{1}{2 q} \biggr).
  \end{equation}
  Assume that the assumptions (I)--(V) and (VI$_2$) are satisfied.  Let $w_1, w_2$ be the two
  samples output by our quantum algorithm.  Then, the probability that all sampled $w_i$ correspond
  to lattice vectors $\lambda^*_i$ in $L_{[0,b)}$ for $i=1, 2$ and that they generate $L$ is at
  least
  \begin{eqnarray*}
    & & 
    \frac{1}{3} \, \left(\frac{M_\ell L_\ell c}{W} \right)^2  \\
    & \ge &
    \frac{1}{12} \kappa^2 c^2 \cdot \left[ 1 - \left( \frac{1}{2 q} + \frac{1}{\det(\Lambda)} \right)
      \frac{1}{\kappa N} \right]^2 \cdot \left[ 1 - \frac{3}{qN} - \frac{\det(\Lambda)}{q} \right]^2,
  \end{eqnarray*}
  where $L_\ell$ is the lower bound on $L_{[0,b)}$ in Proposition~\ref{prop:samplingFromGoodSet},
  $c$ the cosine-factor in Proposition~\ref{prop:samplingFromGoodSet}, $M_\ell$ the lower bound on
  $M$ in Proposition~\ref{prop:Mestimate}~(iii), and $W=2qN$.
  
  Let us introduce the two new assumptions
  \begin{align}\tag{IV$_2$}\label{eq:IV2}
    q \ge{} & \frac{12}{N} + 4 \det(\Lambda) \\
    \text{and} \qquad \tag{VII$_2$}\label{eq:VII2}
    N \ge{} & \frac{1}{\kappa} \left( \frac{2}{q} + \frac{4}{\det(\Lambda)} \right);
  \end{align}
  these imply~(IV), and allow us to bound 
  \[ 
  	\left(1 - \frac{1}{2 q \kappa N} -
  	\frac{1}{ \kappa N \det(\Lambda)} \right)^2 \ge \frac{1}{2} \quad \text{and} \quad 
  	\left(1 - \frac{3}{qN} - \frac{\det(\Lambda)}{q}\right)^2 \ge \frac{1}{2}. 
  \] 
  This yields the lower bound $\frac{1}{48} \kappa^2 c^2$ on the success probability.

	\section{Lattice theoretic tools -- Part 2}\label{sec:part2}

	First, we consider the problem to obtain an approximate basis of a lattice $L$ from an approximate generating set of $L$.  
	Second, we consider the problem to obtain an approximate basis of the dual lattice $L^{*}$
  from an approximate basis of $L$.

	\subsection{Computing an approximate basis of $L$ from an approximate generating set of $L$}

  We address the problem of computing an approximate basis from an approximate generating
  set.  In this subsection, we present \textsc{Buchmann's and Kessler's} approach in
  \cite{BK:93}.  Our exposition simplifies and improves their results.  Our more general analysis
  makes it possible to quantify the approximation quality when different lattice approximation
  algorithms can be used.  The analysis in \cite{BK:93} is written only for the LLL algorithm.
  In the context of our quantum algorithm it is more advantageous to use algorithms to compute
  Korkine-Zolotarev reduced bases.  In our analysis, the approximation quality is entirely
  expressed in terms of the lattice $L$.  In contrast, in \cite{BK:93} the approximation
  quality depends on the characteristics of some sublattice of $L$.

	\begin{remark}
	An approach based on \cite{BK:93} was already suggested in \cite{arthurDiss}.  However, our requirements on the precision of the approximation can be stated in much simpler terms than those made in 		\cite{arthurDiss}.  For instance, an important simplification is that we do not have to consider any sublattice (compare to \cite[Satz 2.4.24]{arthurDiss}).

	Note that \cite{hallgrenUnitgroup} suggested to use the precursor \cite{BP:87} for computing an approximate basis.  The problem is that this earlier work does not make any statements on the size of 	
	the entries of a certain unimodular transformation matrix.  Therefore, the results of this work cannot be directly applied because it not possible to quantify the quality of the resulting approximate 
	basis.  The major motivation for the follow-up work \cite{BK:93} to \cite{BP:87} was to bound the entries of the relevant transformation matrix (see \cite[Introduction]{BK:93}).

Observe that both \cite{BK:93} and \cite{BP:87} rely on the LLL basis reduction algorithm to compute
the transformation matrix.  However, for the quantum algorithm it is significantly better to compute
Korkine-Zolotarev-reduced bases in the classical post-processing step.  This makes it possible to
obtain a transformation matrix with exponentially smaller entries, which in turn yields an
exponentially better approximation of the basis of the period lattice of the infrastructure.  If the LLL
algorithm is used, then it is necessary to evaluate the function $f$ over an exponentially wider window to 
achieve the same quality of approximation of the period lattice.  Note that the cost of computing Korkine-Zolotarev bases
in the classical post-processing step is negligible compared to the time complexity of the quantum part.
\end{remark}

\newpage

Let $L$ be a lattice in $\R^n$ of rank $r\le n$.

\begin{definition}[Approximate basis]
We call $\fb'_1,\ldots,\fb'_r$ a $\delta$-approximate basis of $L$ if there exists a basis $\fb_1,\ldots,\fb_r$ of $L$ with
\begin{equation*}
\| \fb'_i - \fb_i \|_2 \le \delta
\end{equation*}
for $i=1,\ldots,r$.
\end{definition}

\begin{definition}[Approximate generating set]
We call $\fa'_1,\ldots,\fa'_k$ an $\varepsilon$-ap\-prox\-i\-mate generating set of $L$ if there exists a generating set $\fa_1,\ldots,\fa_k$ of $L$ with
\begin{equation}\label{eq:error}
\| \fa'_j - \fa_j \|_2 \le \varepsilon
\end{equation}
for $j=1,\ldots,k$.
\end{definition}

We assume 
\begin{eqnarray*}
\mu    & \le & \lambda_1(L) \\
\alpha & \ge & \max_{j=1,\ldots,k} \{ \| \fa_j \|_2 \}\,.
\end{eqnarray*}
We need these bounds to derive the method for computing an approximate basis from an $\varepsilon$-approximate generating set and to bound its corresponding $\delta$ in terms of $\varepsilon$, $\mu$, $\alpha$, $n$, and $k$.

\begin{remark}
The approximate generating set arises in the following way in our quantum algorithm.  We are given an algorithm that returns rational vectors of the special form $[t \mathbf{a}_j]$ where 
the vectors $\fa_1,\ldots,\fa_k$ generate the lattice.  The parameter $t$ specifies the quality of the approximation and is under our control.  The problem is to find a unimodular matrix $T\in\Z^{k\times r}$ that transforms the $\frac{\sqrt{n}}{2t}$ approximate generating set $\frac{1}{t} [t \fa_j]$ into an approximate basis of $L$ and to determine its corresponding $\delta$.
\end{remark}

We call a vector $\fz=(z_1,\ldots,z_k)\in\Z^k$ a (nontrivial) relation for the generating set if $\fz\neq \mathbf{0}$ and
\begin{equation}\label{eq:relation}
\sum_{j=1} z_j \fa_j = \mathbf{0}\,,
\end{equation}
where $\mathbf{0}$ denotes the (column) zero vector in either $\Z^k$ or $\Z^n$.

\begin{lemma}[Sufficient and necessary condition for relations]\label{lem:suffNecCondRel}
Let $\fz\in\Z^k$ and assume that 
\begin{equation}\label{eq:lowerBound1}
2 \varepsilon \|\fz \|_1  < \mu\,.
\end{equation}
Then $\fz$ is a relation for the generating set if and only if
\begin{equation}\label{eq:condRelation}
\biggl\| \sum_{j=1}^k z_j \fa'_j \biggr\|_2 \le \varepsilon \| \fz \|_1\,.
\end{equation}
\end{lemma}

\begin{proof}
Let $\fz$ be an arbitrary relation. The condition in (\ref{eq:condRelation}) follows then from (\ref{eq:error}) and (\ref{eq:relation}) 
\begin{equation}\label{eq:relationApproxGen}
\biggl\| \sum_{j=1}^k z_j \fa'_j \biggr\|_2 \le \sum_{j=1}^k |z_j| \, \|\fa'_j - \fa_j\|_2 + \, \biggl\| \sum_{j=1}^k z_j \fa_j \biggr\|_2 \le \varepsilon\| \fz \|_1\,.
\end{equation}
Now assume that (\ref{eq:condRelation}) holds for some (nonzero) vector $\fz\in\Z^k$.  Using (\ref{eq:error}) and (\ref{eq:lowerBound1}) we obtain
\begin{eqnarray*}
\biggl\| \sum_{j=1}^k z_j \fa_j \biggr\|_2  & \le & \biggl\| \sum_{j=1}^k z_j (\fa_j - \fa'_j) \biggr\|_2 + \biggl\| \sum_{j=1}^k z_j \fa'_j \biggr\|_2 \\
& \le &
2 \varepsilon \| \fz \|_1 < \mu\,.
\end{eqnarray*}
Since $\mu\le\lambda_1(L)$ we must have that $\sum_{j=1}^k z_j \fa_j=\mathbf{0}$.
\end{proof}

It is convenient to define the scaled approximation vectors 
\begin{equation*}
\ha_j = s\fa'_j\,,
\end{equation*}
where $s$ is a positive parameter that we fix later.  Clearly, $\| \ha_j - s \fa_j \|_2 \le s \varepsilon$.

\begin{definition}[Approximation lattice]
For $j=1,\ldots,k$, define the vectors $\ta_j\in\Z^k \oplus \R^n$ by
\begin{equation*}
\ta_j = \fe_j \oplus \ha_j\,,
\end{equation*}
where $\fe_j$ is the $j$th standard basis vector of $\Z^k$.  The vectors $\ta_1\,\ldots,\ta_k$ are linearly independent and form a basis of the approximation lattice 
\begin{equation*}
\tilde{L} = \bigoplus_{j=1}^k \Z \ta_j\,.
\end{equation*}
\end{definition}

The following lemma establishes that short lattice vectors of $\tilde{L}$ give rise to relations for the generating set of $L$.  For the sake of
generality we introduce the parameter $f$ that characterizes the approximation quality of basis reduction algorithms.  We have $f=2^{(k-1)/2}$ and 
$f=\frac{\sqrt{k+3}}{2}$ for the algorithms that compute LLL-reduced and Korkine-Zolotarev reduced bases.

\begin{lemma}[Sufficient condition for relations]\label{lem:suffCondRel}
Let $\lambda\ge 1$.  Assume that the approximation error $\varepsilon$ is bounded from above by
\begin{equation*}
\varepsilon \le \frac{\mu}{2 f \lambda \sqrt{k}}
\end{equation*}
and the scaling factor $s$ is chosen so that
\begin{equation}\label{eq:lowerBound2}
s > \frac{2 f \lambda}{\mu}\,.
\end{equation}
Let $\fz=(z_1,\ldots,z_k)\in\Z^k$ be an arbitrary vector and 
\begin{equation*}
\tx = \sum_{j=1}^k z_j \ta_j\,.
\end{equation*}
the corresponding lattice vector of $\tilde{L}$.
If 
\begin{equation*}
\| \tx \|_2 \le f \lambda
\end{equation*}
then $\fz$ is a relation for the generating set $\fa_1,\ldots,\fa_k$.
\end{lemma}
\begin{proof}
We prove the lemma by showing that the contraposition of the statement holds.  Assume that $\fz$ is not a relation.  We have to show that corresponding vector $\tx$ is strictly longer than $f \lambda$.

We write $\tx = \fz \oplus \hx$ with $\hx = \sum_{j=1}^k z_j \ha_j$.  Then we have
\begin{equation*}
\| \tx \|_2^2 = \| \fz \|_2^2 + \| \hx \|_2^2\,. 
\end{equation*}

If $\| \fz \|_2 > f\lambda$ holds then we are done.  Otherwise we have
\begin{eqnarray*}
\| \tx \|_2 
& \ge & 
\| \hx \|_2 
=
\biggl\| \sum_{j=1}^k z_j \ha_j \biggr\|_2 \\
& \ge &
s \biggl\| \sum_{j=1}^k z_j \fa_j \biggr\|_2 - \biggl\| \sum_{j=1}^k z_j (s\fa_j - \ha_j)\biggr\|_2 \\
& \ge &
s \, \mu - s \, \| \fz \|_1 \, \varepsilon
\ge
s \, \mu - s \, \| \fz \|_2 \, \sqrt{k} \varepsilon \\
& \ge  &
s \, \mu - s \, f \lambda \sqrt{k} \varepsilon
=
s \left( \mu - f \lambda \sqrt{k} \varepsilon \right) \\
& \ge &
s \left( \mu - \frac{\mu}{2} \right)
\ge
s \, \frac{\mu}{2}
>
f \lambda\,.
\end{eqnarray*}
\end{proof}

\begin{lemma}[Linearly independent relations of bounded norm]\label{lem:linIndepRel}
There exist $k-r$ linearly independent relations $\fm_1,\ldots,\fm_{k-r}$ of the generating set with
\begin{equation*}
\| \fm_j \|_{\infty} \le \frac{\alpha^r}{\det(L)}\,.
\end{equation*}
\end{lemma}
\begin{proof}
We construct an isometric embedding of $L$ into $\R^r$.  Let $\fb_1,\ldots,\fb_r$ be a basis of $L$ and $\fb^*_1,\ldots,\fb^*_r$ the corresponding orthonormal vectors obtained by the Gram-Schmidt process.
Let $\mathbf{w}_1,\ldots,\mathbf{w}_r$ be an arbitrary orthonormal basis of $\R^r$.  
The mapping $\Phi$ defined by 
\begin{equation*}
\Phi(\fb^*_i) = \mathbf{w}_i
\end{equation*}
for $i=1,\ldots,r$ is an isometry between $L$ and $L^{\Phi}:=\Phi(L)$ and we have $\det(L)=\det(L^{\Phi})$.  We set $\fa^{\Phi}_i := \Phi(\fa_i)$.  We assume w.l.o.g. that the
first $r$ vectors of the generating set $\fa_1,\ldots,\fa_k$ are linearly independent.  Define the matrices
\begin{eqnarray*}
A       & = & (\fa^{\Phi}_1|\cdots|\fa^{\Phi}_r|\cdots|\fa^{\Phi}_k) \\
C       & = & (\fa^{\Phi}_1|\cdots|\fa^{\Phi}_r)
\end{eqnarray*}
The submatrix $C\in\R^{r\times r}$ of $A\in\R^{r\times k}$ is nonsingular, which follows from the assumption that the first $r$ generators of $L$ are linearly independent.
Let $\mathbf{v}_j\in\R^r$ be the solutions of the linear system 
\begin{equation*}
C \mathbf{v}_j = \fa^{\Phi}_{r+j}
\end{equation*}
for $j=1,\ldots,k-r$.  
Define the (column) vectors 
\begin{equation*}
\fm_j= \frac{\det(C)}{\det(L^{\Phi})} \bigl( \mathbf{v_j} \oplus (-1) \mathbf{e}_j \bigr)\,,
\end{equation*}
where $\mathbf{e}_j$ are the standard basis vectors of $\R^{k-r}$ for $j=1,\ldots,k-r$.  Due to construction they are linearly independent and form a basis of the kernel of $A$ (which has dimension $k-r$) since 
\begin{equation*}
A \fm_j = \frac{\det(C)}{\det(L^{\Phi})} \left( C \mathbf{v}_j - \fa_{r+j} \right) = \frac{\det(C)}{\det(L^{\Phi})} \left( \fa_{r+j} - \fa_{r+j} \right) = \mathbf{0}
\end{equation*}
for $j=1,\ldots,k-r$.  Using Cramer's rule, we can express the coefficients $v_{ij}$ of the vector $\mathbf{v}_j$ as
\begin{equation*}
v_{ij} = \frac{\det(\fa^{\Phi}_1|\cdots|\fa^{\Phi}_{i-1}|\fa^{\Phi}_{r+j}|\fa^{\Phi}_{i+1}|\cdots|\fa^{\Phi}_r)}{\det(C)}\,.
\end{equation*}
Note that the values 
\begin{equation*}
\frac{\det(C)}{\det(L^{\Phi})} v_{ij} = \frac{\det(\fa^{\Phi}_1|\cdots|\fa^{\Phi}_{i-1}|\fa^{\Phi}_{r+j}|\fa^{\Phi}_{i+1}|\cdots|\fa^{\Phi}_r)}{\det(L^{\Phi})}
\end{equation*}
are always either $0$ or the indices of full-rank sublattices of $L^{\Phi}$.  The two mutually exclusive cases are: (i) $\fa^{\Phi}_{r+j}$ is contained in the span of $\fa^{\Phi}_1,\ldots,\fa^{\Phi}_{i-1},\fa^{\Phi}_{i+1},\ldots,\fa^{\Phi}_r$, implying that the determinant is $0$ and (ii) $\fa^{\Phi}_1,\ldots,\fa^{\Phi}_{i-1},\fa^{\Phi}_{r+j}, \fa^{\Phi}_{i+1},\ldots,\fa^{\Phi}_r$, implying that they generate a full-rank sublattice.  Therefore, all components of $\fm_j$ are integers. This concludes the proof that $\fm_1,\ldots,\fm_{k-r}$ are relations for the generating set.

The upper bound on the $\| \cdot \|_\infty$-norm of these relations follows directly from Minkowski's inequality.  We can bound the absolute value of the determinants by the product of the norms of the column vectors, which can be at most $\alpha^r$.
\end{proof}

\begin{lemma}[Upper bounds on minima of the approximate lattice]\label{lem:minima}
Assume we set 
\begin{equation*}
\lambda = 3 \sqrt{k} \frac{\alpha^r}{\det(L)}
\end{equation*}
and choose the scaling factor $s$ so that 
\begin{equation}\label{eq:upperBound}
s \le \frac{4 f \lambda}{\mu}\,.
\end{equation}

The first $(k-r)$ minima are bounded from above by 
\begin{equation*}
\lambda_j(\tilde{L}) \le \lambda
\end{equation*}
for $j=1,\ldots,k-r$. 

The last $r$ minima are bounded from above
\begin{equation*}
\lambda_j(\tilde{L}) \le \sqrt{s^2(\alpha+\varepsilon)^2 + 1} \le \frac{6.5f\lambda\alpha}{\mu}
\end{equation*}
for $j=k-r+1,\ldots,k$.
\end{lemma}

\begin{proof}
Let $\fm_j$ be the $(k-r)$ linearly independent relations constructed in the proof of Lemma~\ref{lem:linIndepRel}.
We define the vectors
\begin{equation*}
\tx_j = \sum_{i=1}^k m_{ij} \ta_i = \fm_j \oplus \sum_{i=1}^k m_{ij} \ha_i\,.
\end{equation*}
Obviously, the vectors $\tx_j$ are linearly independent.  Since $\fm_j$ is a relation we may apply the 
inequality in (\ref{eq:relationApproxGen}) from the first part of the proof of Lemma~\ref{lem:suffNecCondRel}.  We obtain
\begin{eqnarray*}
\| \tx_j \|_2 
& \le & 
\| \fm_j \|_2 + \biggl\| \sum_{i=1}^k m_{ij} \ha_i \biggr\|_2
=
\| \fm_j \|_2 + s\, \biggl\| \sum_{i=1}^k m_{ij} \fa'_i \biggr\|_2 \\
& \le &
\| \fm_j \|_2 + s\, \varepsilon\, \| \fm_j \|_1 
\le
\sqrt{k} \| \fm_j \|_{\infty} + s \, \varepsilon\, k \, \|\fm_j\|_{\infty} \\
& = &
\left( 1 + s\,\varepsilon \sqrt{k} \right) \sqrt{k} \, \| \fm_j \|_{\infty}
\le
\left( 1 + \frac{2}{\sqrt{k}}  \sqrt{k} \right) \sqrt{k} \, \| \fm_j \|_{\infty} \\
& \le &
3 \sqrt{k} \frac{\alpha^r}{\det(L)}
\le
\lambda\,.
\end{eqnarray*}
The upper bound on the last minima follows from 
\begin{equation*}
\lambda_j(\tilde{L}) \le \max_{i=1,\ldots,k} \| \ta_i \|_2 \le \sqrt{s^2(\alpha+\varepsilon)^2 + 1}\,.
\end{equation*}
The upper bound on the square root expression holds since the tangent to the square root at $s^2(\alpha+\varepsilon)^2>1$ has slope greater or equal to $1/2$ so a displacement by $1$ can increase the value by at most $1/2$ and $s\varepsilon \le 2/\sqrt{k}\le 1$.  This yields observations yield the upper bound $4f\lambda\alpha/\mu + 2.5$, which bounded from above by $6.5 f\lambda\alpha/\mu$.
\end{proof}

To simplify notation in the following we set
\begin{equation*}
\tilde{\alpha} = \sqrt{s^2(\alpha+\varepsilon)^2 + 1}.
\end{equation*}

We apply the basis reduction algorithm to the lattice basis $\ta_1,\ldots,\ta_k$ and obtain the reduced basis $\tb_1,\ldots,\tb_k$.  Denote by $M=(m_{ij})\in\Z^{k\times k}$ the corresponding (unimodular) transformation matrix.  We write the reduced basis vectors as
\begin{equation*}
\tb_j = (\fm_j,\hb_j)
\end{equation*}
where $\fm_j=(m_{1j},\ldots,m_{kj})\in\Z^k$ are the column vectors of $M$ and 
\begin{equation*}
\hb_j = \sum_{i=1}^k m_{ij} \ha_j\,.
\end{equation*}

The following lemma shows we can directly obtain a basis of $L$ with the help of the transformation matrix $M$.

\begin{lemma}[Basis and approximate bases for $L$]
Set 
\begin{equation*}
\lambda=3\sqrt{k} \frac{\alpha^r}{\det(L)}\,.
\end{equation*}  
Assume that the approximation error is bounded from above by
\begin{equation*}
\varepsilon \le \frac{\mu}{2 f \lambda \sqrt{k}}
\end{equation*}
and the scaling factor is bounded from below and above by
\begin{equation*}
\frac{2 f \lambda}{\mu} < s \le \frac{4 f \lambda}{\mu}\,.
\end{equation*}
Let $M$ be the transformation matrix returned by the basis reduction algorithm when applied to the basis $\ta_1,\ldots,\ta_k$ of the approximation lattice $\tilde{L}$.  

Define the vectors
\begin{eqnarray*}
\fb_j  & = & \sum_{i=1}^k m_{i,k-r+j} \fa_i  \\
\fb'_j & = & \sum_{i=1}^k m_{i,k-r+j} \fa'_i \\
\end{eqnarray*}
for $j=1,\ldots,r$.  Then we have
\begin{itemize}
\item The vectors $\fb_1,\ldots,\fb_r$ form a basis of $L$ and their norms are bounded from above by 
\begin{equation*}
\| \fb_j \|_2 \le f \sqrt{k} \, \tilde{\alpha} \, \alpha \,.
\end{equation*}
\item The vectors $\fb'_1,\ldots,\fb'_r$ form a $\delta$-approximate basis of $L$ with 
\begin{equation*}
\delta \le f \sqrt{k} \, \tilde{\alpha} \, \varepsilon \,.
\end{equation*}
\end{itemize}
\end{lemma}

\begin{proof}
We know that the reduced basis vectors $\tb_j$ satisfy
\begin{equation*}
\| \tb_\ell \|_2 \le f \lambda_\ell(\tilde{L})\,.
\end{equation*}
Using the upper bounds on the first $(k-r)$ minima in Lemma~\ref{lem:minima} we obtain 
\begin{equation*}
\| \tb_\ell \|_2 \le f \lambda
\end{equation*}
for $\ell=1,\ldots,k-r$.  These vectors are sufficiently short so that Lemma~\ref{lem:suffCondRel} applies.  We conclude that $\fm_1,\ldots,\fm_{k-r}$ are relations for 
the generating set $\fa_1,\ldots,\fa_k$ of $L$.  

Let $A=(\fa_1 | \cdots | \fa_k)\in\Z^{n\times k}$.  Then we have
\begin{equation*}
A M =
\left(
\underbrace{\mathbf{0} | \cdots | \mathbf{0}}_{k-r} | \fb_1 | \cdots | \fb_r
\right) \in \Z^{n\times k}
\end{equation*}
since the first $k-r$ columns of $M$ are relations of the generating set.  Since $M$ is unimodular the lattice generated by $\fb_1,\ldots,\fb_r$ is equal to $L$ and, thus, $\fb_1,\ldots,\fb_r$ form a basis.

We first determine an upper bound on the norm of the last $r$ column vectors of $M$.  For $j=1,\ldots,r$, we have
\begin{equation*}
\|\fm_{k-r+j}\|_2 \le \| \tb_j \|_2 \le f \lambda_{k-r+j}(\tilde{L}) \le f \tilde{\alpha}\,.
\end{equation*}

We have
\begin{eqnarray*}
\| \fb_j \|_2          & \le & \| \fm_{k-r+j} \|_1 \, \alpha   \le \sqrt{k} f \tilde{\alpha} \, \alpha \\
\| \fb'_j - \fb_j \|_2 & \le & \| \fm_{k-r+j} \|_1 \, \varepsilon \le \sqrt{k} f \tilde{\alpha} \, \varepsilon
\end{eqnarray*}
for $j=1,\ldots,r$.
\end{proof}

We assume in the following the lattice $L$ has full rank, i.e., $r=n$.  This situation occurs precisely in our quantum algorithm.  
To further simplify notation, we also set
\begin{equation*}
g := f \sqrt{k} \tilde{\alpha}\,.
\end{equation*}

\subsection{Computing an approximate basis of the dual lattice $L^\ast$ from an approximate basis of $L$}

\begin{lemma}
Let $\fb'_1,\ldots,\fb'_n$ be a $\delta$-approximate basis of $L$ with $\delta\le g \varepsilon$ as in the lemma above.  Then we can obtain a $\gamma$-approximate basis of the dual lattice $L^*$ with
\begin{equation*}
\gamma \le \frac{2 n^{5/2} g^{2n-1} \alpha^{2(n-1)}}{\det(L)^2}\, \varepsilon\,.
\end{equation*}
provided that
\begin{equation*}
	\varepsilon \le \frac{\det(L)}{2 n^{3/2} g^n \alpha^{n-1}}\,.
\end{equation*}
\end{lemma}
\begin{proof}
Let $B=(\fb_1|\cdots|\fb_n)$ and $B'=(\fb'_1|\cdots|\fb'_n)$ be the matrices whose columns form the basis of $L$ and the approximate basis of $L$, respectively.  
We compute the inverses of these matrices to obtain the basis and the approximate basis of the dual lattice $L^*$.

Denote the perturbation by $E=B'-B$.  We use \cite[Theorem 2.5]{SS:90} to estimate the sensitivity of the inverse under perturbation.  If $\|B^{-1}\|_1 \|E\|_1 < 1$, then $B+E$ is nonsingular and 
\begin{equation*}
\| {B'}^{-1} - B^{-1} \|_1 =
\| (B + E)^{-1} - B^{-1} \|_1 \le \frac{\|B^{-1}\|_1^2 \, \|E\|_1}{1-\|B^{-1}\|_1 \, \|E\|_1}\,.
\end{equation*}
We may apply the bound from this theorem because the matrix norm on $\R^{n\times n}$ defined by $\|X\|_1 = \max_{1\le j\le n} \sum_{i=1}^n x_{ij}$ is multiplicative.

Let $c_{ij}$ denote the entries of $B^{-1}$. Using Cramer's rule and Hadamard's inequality, we have
\begin{equation*}
|c_{ij}| 
=  
\left| \frac{\det(\fb_1,\ldots,\fb_{i-1},\fe_j,\fb_{i+1},\ldots,\fb_{n})}{\det(B)} \right| 
\le  
\frac{\prod_{i\neq j} \|\fb_i\|_2}{\det(L)} 
\le 
\frac{(g \alpha)^{n-1}}{\det(L)}\,.
\end{equation*}
This implies 
\begin{equation*}
\| B^{-1} \|_1 \le  \frac{n (g \alpha)^{n-1}}{\det(L)}\,.
\end{equation*}
Note that the Euclidean norm of the column vectors of $E$ is bounded by $\delta$ from above since these vectors are equal to $\fb'_i-\fb_i$.  
This implies $\|E\|_1 \le \sqrt{n} \delta \le \sqrt{n} g \varepsilon$.

Assume that 
\begin{equation}
	\varepsilon \le \frac{\det(L)}{2 n^{3/2} g^n \alpha^{n-1}}\,,
\end{equation}
which ensures that $\|B^{-1}\|_1 \|E\|_1 \le 1/2$.  Then we have
\begin{equation}
  \|B'^{-1} - B^{-1}\|_1 \le \frac{2 n^{5/2} g^{2n-1} \alpha^{2(n-1)}}{\det(L)^2}\, \varepsilon\,.
\end{equation}
This implies that the column vectors of $B'^{-1}$ form a $\gamma$-approximate basis of $L^*$ with
\begin{equation}
\gamma \le \frac{2 n^{5/2} g^{2n-1} \alpha^{2(n-1)}}{\det(L)^2}\, \varepsilon\,.
\end{equation}
\end{proof}

\begin{corollary}
  \label{cor:generatingsetbasis}
  Recall that the quantum algorithm returns a generating set with $\varepsilon\le 1 / (4 \sqrt{n} q)$. This and the above lemma imply that if
  \begin{equation}
    q \ge \max\left\{ \frac{n g^n \alpha^{n-1}}{ 2 \det(L)},\; \frac{n^2 g^{2n-1} \alpha^{2(n-1)}}{ 2 \det(L)^2} \cdot \frac{1}{\gamma} \right\}
  \end{equation}
  then we obtain a $\gamma$-approximate basis of the dual lattice $L^*$, where 
  \begin{equation}
    g \le \frac{19.5 k f^2 \alpha^{n+1}}{\det(L) \lambda_1(L)} \quad \text{and} \quad \alpha \ge \max_{j=1,\ldots,k} \{ \| \fa_j \|_2 \}\,.
  \end{equation}
\end{corollary}

\begin{proof}
  This follows with $g=\sqrt{k} f \tilde{\alpha} \le \sqrt{k} 6.5 f^2 \lambda \alpha/\mu$ and $\lambda=3\sqrt{k} \alpha^n/\det(L)$.
\end{proof}

  \section{Final analysis of the quantum algorithm}\label{sec:final}
  
  By combining all material from the previous sections, we obtain the following result:
  
  \begin{theorem}
    \label{thm:maintheorem}
    Assume that \ASM1--\ASM3 hold with $C \le 1$ and $A \ge 1$. Further, assume that $N, q, N_0, L
    \in \N$ are chosen such that
    \begin{align*}
      N \ge{} & \max\biggl\{ 32, \; \frac{8 (n + 1) n 2^n D A^{n-1}}{3 C^n}, \; \frac{9 n^2}{32} +
                \frac{18 n^4}{\lambda_1(\Lambda)}, \\
              & \qquad\quad \max\{8n-2,n^{(n-1)/2} 2^{n+1}-2\} \cdot \frac{9 n^2}{2\lambda_1(\Lambda)}
                + \frac{9}{64} \biggr\}, \\
      N_0 \ge{} & 8 n^2 (n + 1) N, \\
      q \ge{} & \max\biggl\{ 32, \; 9 A, \; \frac{6 n^2}{N} + 2 n^{(n+1)/2+1} (n + 1) \frac{\det
                (\Lambda)}{\lambda_1(\Lambda)^{n-1}}, \\
              & \qquad\quad \frac{19.5^n n^{n+3/2} (1 + \tfrac{5}{2n} + \tfrac{1}{n^2})^n
              N_0^{n^2+2n-1} \det(\Lambda)^{2n+1}}{2 \cdot 9^{n^2+2n-1} \lambda_1(\Lambda)^{n^2-n}}, \\
              & \qquad\quad \frac{19.5^{2n} n^{2n+3/2} (1 + \tfrac{5}{2n} + \tfrac{1}{n^2})^{2n-1}
              N_0^{2n^2+3n-3} \det(\Lambda)^{4n}}{\gamma \cdot 39 \cdot 9^{2n^2+3n-3}
              \lambda_1(\Lambda)^{2n^2-3n-1}} \biggr\} \\ \text{and} \quad
      L \ge{} & \frac{4 n D (q + A + C + 2)^n}{C^n}.
    \end{align*}
    Set $\kappa := \frac{1}{9 n}$ and assume that $s \in S$ is chosen uniformly at random. Then the
    probability that the algorithm described in Section~\ref{sec:algoutline}, applied $n$~times with
    the parameters~$N, q, \kappa$ and $n + 1$~times with the parameters~$N_0, q, \kappa$, returns an
    $\frac{1}{4\sqrt{n}q}$-approximate generating set of $\Lambda^*$ is at least
    \begin{align*}
      & \frac{\cos\bigl(\pi \tfrac{17417}{36864} \bigr)^{4n+2}}{2^{2n+6} 3^{4 n^2 + 2 n} n^{4 n^2 +
      2 n}} \biggl( \prod_{i=2}^{n+1} \zeta(i)^{-1} - \tfrac{1}{4} \biggr) \prod_{i=1}^{n-1} (1 -
      2^{-i}) \\
      {}\ge{} & \frac{6.198327 \cdot 1.54587777^n}{10^{6 n + 6} 81^{n^2} n^{4 n^2 + 2 n}}.
    \end{align*}
    If such an approximate generating set of $\Lambda^*$ is obtained, the algorithm described in
    Section~\ref{sec:part2} computes a $\gamma$-approximate basis of $\Lambda$.
  \end{theorem}
  
  We will prove this theorem further down (on page~\pageref{thm:maintheorem:proof}). In case $n =
  1$, we can improve the bound from Theorem~\ref{thm:maintheorem} significantly:
  
  \begin{proposition}
    \label{prop:onedimAlgorithm}
    Assume that $\Lambda \subseteq \R$, i.e., that $n = 1$. Further, assume that \ASM1--\ASM3 hold
    with $C \le 1$ and $A \ge 1$, and assume that $N, q, L \in \N$ are chosen such that
    \begin{align*}
      N \ge{} & \max\biggl\{ 32, \; \frac{4}{A}, \; \frac{32 D}{3 C}, \;
      \frac{36}{\det(\Lambda)} + \frac{9}{16}, \; \frac{27}{\det(\Lambda)} + 9 +
      \frac{9}{16} \biggr\}, 
      \\
      q \ge{} & \max\biggl\{ 32, \; 9 A, \; \frac{12}{N} + 4 \det(\Lambda), \; \frac{19.5}{9^2} N^2
      \det(\Lambda)^3 \cdot \max\left\{ 1, \; \frac{\det(\Lambda)}{\gamma} \right\} \biggr\} 
      \\ \text{and} \quad
      L \ge{} & \frac{4 D (q + A + C + 2)}{C}.
    \end{align*}
    Set $\kappa := \frac{1}{9}$ and assume that $s \in S$ is chosen uniformly at random. Then the
    probability that the algorithm described in Section~\ref{sec:algoutline}, applied two times with
    the parameters~$N, q, \kappa$, returns an $\frac{1}{4 q}$-approximate generating set of
    $\Lambda^*$ is at least \[ \frac{\cos^4\bigl(\pi \tfrac{17417}{36864} \bigr)}{7776} \ge 7.163
    \cdot 10^{-9}. \] If such an approximate generating set of $\Lambda^*$ is obtained, the
    algorithm described in Section~\ref{sec:part2} computes a $\gamma$-approximate basis of
    $\Lambda$.
  \end{proposition}
  
  We will also prove this proposition further down (on page~\ref{prop:onedimAlgorithm}).
  
  One important remark is that it is not possible to determine whether our algorithm actually
  returns the lattice $\Lambda$ or a proper sublattice of $\Lambda$. This is a problem of all such
  quantum algorithms, in particular the ones by \textsc{Hallgren} and \textsc{Schmidt and
    Vollmer}. In case the infrastructure is obtained from a global field, checking whether the
  lattice computed by our algorithm is a sublattice of $\Lambda$ can be done efficiently: one simply
  has to check whether the computed basis consists of units of the global field. However, even when
  one assumes that the Generalized Riemann Hypothesis holds, there is no efficient polynomial-time
  algorithm known which certifies that a given sublattice of $\Lambda$ equals $\Lambda$. But we
  assume that the case that a basis returned by our algorithm (and any of the other algorithms, for
  that it matters) is a proper sublattice of $\Lambda$ is somewhat pathological.
  
  Note that the lower bound on the success probability is very small even for moderate~$n$. More
  precisely, for $n = 1, \dots, 10$, the inverses of the probabilities, i.e., the expected
  number of iterations which have to be run, are bounded from above by
  \begin{align*}
    & 1.40 \cdot 10^8, & & 1.27 \cdot 10^{30}, & & 4.67 \cdot 10^{59}, & & 1.74 \cdot 10^{102}, & &
    6.47 \cdot 10^{158}, \\
    & 1.39 \cdot 10^{230}, & & 7.12 \cdot 10^{316}, & & 2.92 \cdot 10^{419}, & & 2.72 \cdot
    10^{538}, & & 1.43 \cdot 10^{674}.
  \end{align*}
  (Note that for $n = 1$, we used the algorithm described in Proposition~\ref{prop:onedimAlgorithm};
  the bound given by the formula in Theorem~\ref{thm:maintheorem} is $1.26 \cdot 10^{12}$.)  The
  success probability for the algorithm in \cite{arthurDiss} is bounded from below by $2^{-20 n^2 -
  12 n - 2} n^{-4 n^2}$, as stated there in Satz~6.2.6. Hence, the expected number of iterations for
  $n = 1, \dots, 10$ for this algorithm are bounded by
  \begin{align*}
    & 1.72 \cdot 10^{10}, & & 5.32 \cdot 10^{36}, & & 6.32 \cdot 10^{82}, & & 8.18 \cdot 10^{149}, &
    & 1.19 \cdot 10^{239}, \\
    & 1.18 \cdot 10^{351}, & & 3.45 \cdot 10^{486}, & & 1.02 \cdot 10^{646}, & & 9.05 \cdot
    10^{829}, & & 6.10 \cdot 10^{1038}.
  \end{align*}
  Note that in \cite{schmidt-vollmer} the success probability is given as $2^{-k n^{2 +
  \varepsilon}}$ for some $k \in \N$ and $\varepsilon > 0$ without making these explicitly; the
  behavior for $n \to \infty$ will be similar to the analysis in \cite{arthurDiss}. Finally, in
  \cite{hallgrenUnitgroup}, no success probability is given at all.  The current analyses can only
  prove expected running times which are impractical. Our analysis improves on the previous ones,
  though not substantially. We believe that it can be further optimized.
  
  Assuming that $n$ is constant, we obtain the following complexity theoretic result, which extends
  the results by {\sc Hallgren} and {\sc Schmidt and Vollmer} to a larger class of infrastructures:
  
  \begin{corollary}
    Assume that $n = O(1)$ and that $\calI$ is an infrastructure satisfying the
    assumptions~\ASM1--\ASM3. We obtain a quantum algorithm to compute $\Lambda$ with a success
    probability bounded away from 0 by a constant which runs in time polynomial in $\log \det
    (\Lambda)$, $\log \frac{1}{\lambda_1(\Lambda)}$, $\log \frac{1}{\gamma}$, $\log A$, $\log
    \frac{1}{C}$ and $\log D$. \qed
  \end{corollary}
  
  Note that $\log L$, $\log N$, $\log N_0$ and $\log q$ can all be chosen to be \emph{linear} in
  $\log \det(\Lambda)$, $\log \frac{1}{\lambda_1(\Lambda)}$, $\log \frac{1}{\gamma}$, $\log A$, $\log
  \frac{1}{C}$ and $\log D$.
  
  Finally, we want to conclude with the proofs of Theorem~\ref{thm:maintheorem} and
  Proposition~\ref{prop:onedimAlgorithm}.
  
  \begin{proof}[Proof of Theorem~\ref{thm:maintheorem}]\label{thm:maintheorem:proof}
    We have $C \le 1$, $A \ge 1$, $N, N_0 \ge 32$, $q \ge \max\{ 32, 9 A \}$ and $\kappa =
    \frac{1}{9 n}$. Clearly, with $q N_0 \ge q N \ge 32^2 > 18$ we get assumption~(V). Since
    $\frac{4}{A} \le 4 \le 32$ and since $N_0 \ge N \ge \frac{8 (n + 1) n 2^n D A^{n-1}}{3 C^n}$ we
    have assumption~(II) for $N$ and $N_0$. The requirement $N_0 \ge 8 n^2 (n + 1) N$ on $N_0$ is
    assumption~(VIII).
    
    Since $N_0 \ge N \ge \max\{8n-2,n^{(n-1)/2} 2^{n+1}-2\} \cdot \frac{9 n^2}{2\lambda_1(\Lambda)}
    + \frac{9}{64} \ge \max\{8n-2,n^{(n-1)/2} 2^{n+1}-2\} \cdot \frac{9 n^2}{2\lambda_1(\Lambda)} +
    \frac{9}{2q} \ge \frac{2 \sqrt{n}}{\lambda_1(\Lambda)}$ we have assumptions~(III) for $N$ and
    $N_0$ as well as assumption~(VI), and as $N \ge \frac{9 n^2}{32} + \frac{18
    n^4}{\lambda_1(\Lambda)} \ge 9 n \bigl( \frac{n}{q} + \frac{2 n^3}{\lambda_1(\Lambda)} \bigr)$
    we get assumption~(VII$_1$). Next, $q \ge 9 A \ge 9$ yields assumption~(II) for $q$. The third
    condition on $q$ yields assumption~(IV$_1$) using the bound $\nu(\Lambda) \le \frac{1}{2}
    n^{(n+1)/2} \frac{\det(\Lambda)}{\lambda_1(\Lambda)^{n-1}}$. That bound follows by Theorem~7.9
    in \cite{micciancio-goldwasser}, stating that $\nu(\Lambda) \le \frac{\sqrt{n}}{2}
    \lambda_n(\Lambda)$, and from \[ \lambda_n(\Lambda) \le n^{n/2} \frac{\det
    (\Lambda)}{\prod_{i=1}^{n-1} \lambda_i(\Lambda)} \le n^{n/2} \frac{\det
    (\Lambda)}{\lambda_1(\Lambda)^{n-1}} \] by Minkowski's second theorem
    \cite[Theorem~1.5]{micciancio-goldwasser}.

    The condition on $L$ ensures that assumption~(I), i.e.\ the hypotheses of
    Corollary~\ref{cor:computefcorrectly}, are satisfied. Hence, if $s \in S$ is uniformly picked,
    with probability at least $1/2$ we have $\Hgrid(1/(2NL)) \cap G(s) = \emptyset$, which
    guarantees that we can compute the function $f$ for all $v \in \calV$ exactly using \ASM3.
    
    Note that $\kappa = \frac{1}{9 n}$ yields $c = \cos^2\bigl(\pi (\tfrac{1}{4} + \tfrac{1}{4 q N}
    + 2 \kappa n)\bigr) \ge \cos^2\bigl(\pi \tfrac{17417}{36864} \bigr) \ge 0.00746$ as $q N \ge
    32^2$. Combining this with the bounds in Section~\ref{sec:boundingprobability} yields the lower
    bound \[ \frac{\cos\bigl(\pi \tfrac{17417}{36864} \bigr)^{4n+2}}{2^{2n+5} 3^{4 n^2 + 2 n} n^{4
    n^2 + 2 n}} p^* \ge \frac{1.239665 \cdot 1.54587777^n}{10^{6 n + 5} 81^{n^2} n^{4 n^2 + 2 n}} \]
    for the probability that $2n+1$ runs of the quantum algorithm (with fixed ``good''~$s$) yield a
    generating set of $\Lambda^*$; here, $p^* \ge \bigl( \prod_{i=2}^{n+1} \zeta(i)^{-1} -
    \tfrac{1}{4} \bigr) \cdot \prod_{i=1}^{n-1} (1 - 2^{-i}) \ge (\hat{\zeta} - \tfrac{1}{4}) \cdot
    0.289 \ge 0.184 \cdot \tfrac{1}{4}$ (compare Equation~\eqref{eq:betterboundsongenprop} on
    page~\pageref{eq:betterboundsongenprop}). This has to be multiplied by $1/2$ for the above
    mentioned probability that a uniformly chosen~$s \in S$ yields $\Hgrid(1/(2NL)) \cap G(s) =
    \emptyset$.
    
    In the context of Corollary~\ref{cor:generatingsetbasis}, we can bound $\alpha$ by $\sqrt{n} b_0
    = \sqrt{n} \kappa N_0 - \frac{\sqrt{n}}{2 n q} \le \frac{1}{9 \sqrt{n}} N_0$, and $k = 2 n + 1$
    is the number of generating elements. When using Korkine-Zolotarev reduction, we can use $f =
    \frac{1}{2} \sqrt{2 n + 4}$. Since $L = \Lambda^*$, we see that $\det(L) = (\det(\Lambda))^{-1}$
    and $\frac{1}{\lambda_1(L)} \le \lambda_n(\Lambda) \le \frac{n^{n/2} \det
    (\Lambda)}{\lambda_1(\Lambda)^{n-1}}$. This yields \[ g \le \frac{13 n^{3/2} (1 + \tfrac{5}{2 n}
    + \tfrac{1}{n^2}) \det(\Lambda)^2 N_0^{n+1}}{6 \cdot 9^n \lambda_1(\Lambda)^{n-1}}. \]
    Therefore, the algorithm in Section~\ref{sec:part2} computes a $\gamma$-approximate basis of
    $\Lambda$ from a $\frac{1}{4\sqrt{n}q}$-approximate generating set of $2n+1$ vectors in
    $\Lambda^*$ if
    \begin{align*}
    q \ge{} & \max\biggl\{ \frac{19.5^n n^{n+3/2} (1 + \tfrac{5}{2n} + \tfrac{1}{n^2})^n
              N_0^{n^2+2n-1} \det(\Lambda)^{2n+1}}{2 \cdot 9^{n^2+2n-1} \lambda_1(\Lambda)^{n^2-n}},
              \\
    & \qquad\quad \frac{19.5^{2n} n^{2n+3/2} (1 + \tfrac{5}{2n} + \tfrac{1}{n^2})^{2n-1}
              N_0^{2n^2+3n-3} \det(\Lambda)^{4n}}{\gamma \cdot 39 \cdot 9^{2n^2+3n-3}
              \lambda_1(\Lambda)^{2n^2-3n-1}} \biggr\}.
    \end{align*}
    But this is satisfied by the fourth and fifth condition on $q$.
  \end{proof}
  
  \begin{proof}[Proof of Proposition~\ref{prop:onedimAlgorithm}.]\label{prop:onedimAlgorithm:proof}
    We have $C \le 1$, $A \ge 1$, $N \ge 32$, $q \ge \max\{ 32, 9 A \}$ and $\kappa =
    \frac{1}{9}$. Clearly, with $q N \ge 32^2 > 18$ we get assumption~(V). The second and third
    assumption on $N$ yield the $N$-part of assumption~(II), the fourth yields assumption~(III) and
    (VII$_2$) and the fifth yields assumption~(VI$_2$). The second assumption on $q$ yields the
    $q$-part of assumption~(II), and the third part yields assumption~(IV$_2$). Note that
    $\lambda_1(\Lambda) = \det(\Lambda)$ and $\nu(\Lambda) = \frac{1}{2} \det(\Lambda)$.
    
    Note that $\kappa = \frac{1}{9}$ yields $c = \cos^2\bigl(\pi (\tfrac{1}{4} + \tfrac{1}{4 q N} +
    2 \kappa n)\bigr) \ge \cos^2\bigl(\pi \tfrac{17417}{36864} \bigr) \ge 0.00746$ as $q N \ge
    32^2$. Combining this with the bounds in Section~\ref{sec:boundingprop1D} yields the lower
    bound \[ \frac{1}{48} \kappa^2 c^2 \kappa^2 c^2 \ge \frac{\cos^4\bigl(\pi \tfrac{17417}{36864}
    \bigr)}{48 \cdot 9^2} \] for the probability that two runs of the quantum algorithm (with
    fixed~$s$) yield a generating set of $\Lambda^*$. This has to be multiplied by $1/2$ for the
    above mentioned probability that a uniformly chosen~$s \in S$ yields $\Hgrid(1/(2NL)) \cap G(s)
    = \emptyset$.
    
    In the context of Corollary~\ref{cor:generatingsetbasis}, we can bound $\alpha$ by $\frac{1}{9}
    N$, and $k = 2$ is the number of generating elements. Since in dimension~one, one can reduce
    perfectly, we can use $f = 1$. Since $L = \Lambda^*$, we have $\det(L) = (\det(\Lambda))^{-1}$ and
    $\lambda_1(L) = \det(L) = (\det(\Lambda))^{-1}$. Using this, the algorithm in
    Section~\ref{sec:part2} computes a $\gamma$-approximate basis of $\Lambda$ from a
    $\frac{1}{4q}$-approximate generating set of two vectors in $\Lambda^*$ if
    \begin{align*}
      q \ge \frac{19.5}{9^2} N^2 \det(\Lambda)^3 \cdot \max\left\{ 1, \; \frac{\det
      (\Lambda)}{\gamma} \right\}
    \end{align*}
    But this is satisfied by the last condition on $q$.
  \end{proof}

  \section*{List of assumptions}
  \addcontentsline{toc}{section}{List of assumptions}
  \setlength\fboxrule{0pt}
  \setlength\fboxsep{1.75pt}
  \newcommand{\borderbox}[1]{\fbox{#1}}
  \begin{center}
    \begin{tabular}{|c|c|l|}\hline
      Assumption & Page & Can be found in \\\hline\hline
      (I) & \pageref{eq:I} & Corollary~\ref{cor:computefcorrectly} \\\hline
      \multicolumn{3}{|c|}{\borderbox{ $L \ge \frac{4 n D (q + A + C + 2)^n}{C^n}$ and $\varepsilon
          \le \tfrac{1}{2NL}$ }}
      \\\hline\hline
      (II) & \pageref{eq:II} & Corollary~\ref{corr:insideprob12} \\\hline
      \multicolumn{3}{|c|}{\borderbox{ $q \ge 9 \max\{ 1, A \}$ and $N \ge \max\bigl\{ \frac{4}{A},
          \; \frac{8 (n + 1) n \cdot 2^n D A^{n-1}}{3 C^n} \bigr\}$ }}
      \\\hline\hline
      (III) & \pageref{eq:III} & Proposition~\ref{prop:Mestimate} \\\hline
      \multicolumn{3}{|c|}{\borderbox{ $N \ge \frac{2 \sqrt{n}}{\lambda_1(\Lambda)}$ }}
      \\\hline\hline
      (IV) & \pageref{eq:IV} & Proposition~\ref{prop:Mestimate} \\\hline
      \multicolumn{3}{|c|}{\borderbox{ $q > 2n\nu(\Lambda) + \frac{3n}{N}$ }}
      \\\hline\hline
      (IV$_1$) & \pageref{eq:IV1} & Section~\ref{sec:boundingprobability} \\\hline
      \multicolumn{3}{|c|}{\borderbox{ $q \ge \frac{6 n^2}{N} + 4 n (n + 1) \nu(\Lambda)$ }}
      \\\hline\hline
      (IV$_2$) & \pageref{eq:IV2} & Section~\ref{sec:boundingprop1D} \\\hline
      \multicolumn{3}{|c|}{\borderbox{ $q \ge \frac{12}{N} + 4 \det(\Lambda)$ }}
      \\\hline\hline
      (V) & \pageref{eq:V} & Proposition~\ref{prop:lbprobone} \\\hline
      \multicolumn{3}{|c|}{\borderbox{ $\kappa < \frac{1}{8n} - \frac{1}{4 n q N}$ }}
      \\\hline\hline
      (VI) & \pageref{eq:VI} & Lemma~\ref{prop:samplingFromGoodSet} \\\hline
      \multicolumn{3}{|c|}{\borderbox{ $N \ge \frac{1}{\kappa} \left( \max\{ 8 n - 2, n^{(n-1)/2}
          \cdot 2^{n+1} - 2 \} \cdot \frac{n}{2\lambda_1(\Lambda)} + \frac{1}{2nq} \right)$ }}
      \\\hline\hline
      (VI$_2$) & \pageref{eq:VI2} & Section~\ref{sec:boundingprop1D} \\\hline
      \multicolumn{3}{|c|}{\borderbox{ $N \ge \frac{1}{\kappa} \bigl( \frac{3}{\det(\Lambda)} + 1 +
          \frac{1}{2 q} \bigr)$ }}
      \\\hline\hline
      (VII) & \pageref{eq:VII} & Lemma~\ref{prop:samplingFromGoodSet} \\\hline
      \multicolumn{3}{|c|}{\borderbox{ $N > \frac{1}{\kappa} \bigl( \frac{1}{2 q} +
        \frac{n^2}{\lambda_1(\Lambda)} \bigr)$ }}
      \\\hline\hline
      (VII$_1$) & \pageref{eq:VII1} & Section~\ref{sec:boundingprobability} \\\hline
      \multicolumn{3}{|c|}{\borderbox{ $N \ge \frac{1}{\kappa} \bigl( \frac{n}{q} + \frac{2
            n^3}{\lambda_1(\Lambda)} \bigr)$ }}
      \\\hline\hline
      (VII$_2$) & \pageref{eq:VII2} & Section~\ref{sec:boundingprop1D} \\\hline
      \multicolumn{3}{|c|}{\borderbox{ $N \ge \frac{1}{\kappa} \bigl( \frac{2}{q} +
          \frac{4}{\det(\Lambda)} \bigr)$ }}
      \\\hline\hline
      (VIII) & \pageref{eq:VIII} & Lemma~\ref{lemma:genentireduallattice} \\\hline
      \multicolumn{3}{|c|}{\borderbox{ $N_0 \ge 8 n^2(n+1) N$ }}
      \\\hline
    \end{tabular}
  \end{center}
  
  \paragraph{Acknowledgements}
  P.W. gratefully acknowledges the support from the NSF grant CCF-0726771 and the NSF CAREER Award
  CCF-0746600. P.W. would also like to thank Joachim Rosenthal and his group members for their
  hospitality during his visit at the Institute of Mathematics, University of
  Zurich. F.F. gratefully acknowledges the support form Armasuisse and the SNF grant No.~132256.


\begin{thebibliography}{Buc87b}

\bibitem[Bar]{barvinok-notes}
A.~Barvinok.
\newblock Math669: Combinatorics, geometry and complexity of integer points.
\newblock \url{http://www.math.lsa.umich.edu/~barvinok/latticenotes669.pdf}.

\bibitem[BJP94]{buchmannjuntgenpohst-practialGLA}
J.~Buchmann, M.~J{\"u}ntgen, and M.~Pohst.
\newblock A practical version of the generalized {L}agrange algorithm.
\newblock {\em Experiment. Math.}, 3(3):199--207, 1994.

\bibitem[BK93]{BK:93}
J.~Buchmann and V.~Kessler.
\newblock Computing a reduced lattice basis from a generating set.
\newblock \url{http://www.cdc.informatik.tu-darmstadt.de/reports/reports/
  reduced\_basis.ps.gz}, 1993.

\bibitem[BP89]{BP:87}
J.~Buchmann and M.~Pohst.
\newblock Computing a lattice basis from a system of generating vectors.
\newblock {\em Proceedings of EUROCAL 1987, Lecture Notes in Computer Science},
  378, 1989.

\bibitem[Buc87a]{genlagrange}
J.~A. Buchmann.
\newblock On the computation of units and class numbers by a generalization of
  {L}agrange's algorithm.
\newblock {\em J. Number Theory}, 26(1):8--30, 1987.

\bibitem[Buc87b]{buchmann-ontheperiodlength}
J.~A. Buchmann.
\newblock On the period length of the generalized {L}agrange algorithm.
\newblock {\em J. Number Theory}, 26(1):31--37, 1987.

\bibitem[Buc87c]{buchmann-habil}
J.~A. Buchmann.
\newblock Zur {K}omplexit\"at der {B}erechnung von {E}inheiten und
  {K}lassenzahl algebraischer {Z}ahlk\"orper.
\newblock Habilitationsschrift, October 1987.

\bibitem[CM01]{cheung-mosca}
K.~K.~H. Cheung and M.~Mosca.
\newblock Decomposing finite abelian groups.
\newblock {\em Quantum Information {\&} Computation}, 1(3):26--32, 2001.

\bibitem[Die08]{diem-habil}
C.~Diem.
\newblock On arithmetic and the discrete logarithm problem in class groups of
  curves.
\newblock Habilitationsschrift. Available at
  \url{http://www.math.uni-leipzig.de/~diem/preprints/english.html}, May 2008.

\bibitem[EH12]{hallgren-eisentraeger}
K.~Eisentr\"ager and S.~Hallgren.
\newblock Computing the unit group, class group and compact representations in
  algebraic function fields.
\newblock To be presented at ANTS~X., 2012.

\bibitem[Fon11]{ff-tioagfoaur}
F.~Fontein.
\newblock The infrastructure of a global field of arbitrary unit rank.
\newblock {\em Math. Comp.}, 80(276):2325--2357, 2011.

\bibitem[Hal02]{hallgrenPell}
S.~Hallgren.
\newblock Polynomial-time quantum algorithms for {P}ell's equation and the
  principal ideal problem.
\newblock In {\em Proceedings of the {T}hirty-{F}ourth {A}nnual {ACM}
  {S}ymposium on {T}heory of {C}omputing}, pages 653--658 (electronic), New
  York, 2002. ACM.

\bibitem[Hal05]{hallgrenUnitgroup}
S.~Hallgren.
\newblock Fast quantum algorithms for computing the unit group and class group
  of a number field.
\newblock In {\em S{TOC}'05: {P}roceedings of the 37th {A}nnual {ACM}
  {S}ymposium on {T}heory of {C}omputing}, pages 468--474. ACM, New York, 2005.

\bibitem[Hes02]{hessRR}
F.~Hess.
\newblock Computing {R}iemann-{R}och spaces in algebraic function fields and
  related topics.
\newblock {\em J. Symbolic Comput.}, 33(4):425--445, 2002.

\bibitem[Ked06]{kedlaya-zeta}
K.~S. Kedlaya.
\newblock Quantum computation of zeta functions of curves.
\newblock {\em Computational Complexity}, 15(1):1--10, 2006.

\bibitem[MG02]{micciancio-goldwasser}
D.~Micciancio and S.~Goldwasser.
\newblock {\em Complexity of lattice problems}.
\newblock The Kluwer International Series in Engineering and Computer Science,
  671. Kluwer Academic Publishers, Boston, MA, 2002.
\newblock A cryptographic perspective.

\bibitem[Pom01]{pomerance-generate}
C.~Pomerance.
\newblock The expected number of random elements to generate a finite abelian
  group.
\newblock {\em Periodica Mathematica Hungarica}, 43(1--2):191--198, 2001.

\bibitem[Sch07]{arthurDiss}
A.~Schmidt.
\newblock {\em Zur {L}\"osung von zahlentheoretischen {P}roblemen mit
  klassischen und {Q}uantencomputern}.
\newblock Ph.{D}. thesis, Technische Universit\"at Darmstadt, 2007.

\bibitem[Sch08]{schoofArakelov}
R.~J. Schoof.
\newblock {\em Computing {A}rakelov class groups}, volume~44 of {\em MSRI
  Publications}, pages 447--495.
\newblock Cambridge University Press, Cambridge, 2008.

\bibitem[Seq]{Sloane}
Integer sequence {A021002}.
\newblock The on-line encyclopedia of integer sequence
  \url{http://oeis.org/A021002}.

\bibitem[SS90]{SS:90}
G.~W. Stewart and J.~G. Sun.
\newblock {\em Matrix perturbation theory}.
\newblock Academic Press, Inc., 1990.

\bibitem[SV05]{schmidt-vollmer}
A.~Schmidt and U.~Vollmer.
\newblock Polynomial time quantum algorithm for the computation of the unit
  group of a number field (extended abstract).
\newblock In {\em S{TOC}'05: {P}roceedings of the 37th {A}nnual {ACM}
  {S}ymposium on {T}heory of {C}omputing}, pages 475--480. ACM, New York, 2005.

\bibitem[SW11]{pradeep-pawel}
P.~Sarvepalli and P.~Wocjan.
\newblock Quantum algorithms for one-dimensional infrastructures.
\newblock \url{http://arxiv.org/abs/1106.6347}, 2011.

\bibitem[Thi95a]{thielDiss}
C.~Thiel.
\newblock {\em On the complexity of some problems in algorithmic algebraic
  number theory}.
\newblock Ph.{D}. thesis, Universit\"at des Saarlands, 1995.

\bibitem[Thi95b]{thiel-comprep}
C.~Thiel.
\newblock Short proofs using compact representations of algebraic integers.
\newblock {\em J. Complexity}, 11(3):310--329, 1995.

\end{thebibliography}
\end{document}